\documentclass[smallextended]{svjour3a}

\usepackage{graphicx,amsmath,amssymb,amsfonts,mathrsfs,hyperref,nicefrac,xcolor,bm,enumitem,xfrac}
\usepackage[paperwidth=7.5in,paperheight=10.64in,text={5.9in,8.9in},centering,top=.75in]{geometry}

\usepackage[normalem]{ulem}
\usepackage{esint}
\usepackage{fancyhdr}
\hypersetup{
 colorlinks,
 linkcolor={blue!90!black},
 citecolor={red!80!black},
 urlcolor={blue!50!black}
}

\numberwithin{equation}{section}
\numberwithin{theorem}{section}
\numberwithin{lemma}{section}

\usepackage{mathpazo}
\usepackage[T1]{fontenc}

\newcommand{\ignore}[1]{}
\newcommand{\opn}{\operatorname}

\newcommand{\re}{\operatorname{Re}}
\newcommand{\im}{\operatorname{Im}}

\newcommand{\mbb}[1]{\mathbb{#1}}
\newcommand{\mb}[1]{\mathbf{#1}}
\newcommand{\mc}[1]{\mathcal{#1}}
\newcommand{\pa}{\partial}
\newcommand{\der}[2]{\frac{\partial #1}{\partial #2}}

\newcommand{\jd}{\displaystyle}

\newcommand{\jt}{\textstyle}

\newcommand{\e}[1]{{(#1)}}

\newcommand{\wtil}[1]{\widetilde{#1}}

\newcommand{\sech}{\operatorname{sech}}

\newcommand{\lp}{\left(}
\newcommand{\rp}{\right)}

\makeatletter
\newcommand{\vast}{\bBigg@{3}}
\newcommand{\Vast}{\bBigg@{3.5}}
\newcommand{\VVast}{\bBigg@{4}}

\DeclareRobustCommand{\gaussk}{\DOTSB\gaussk@\slimits@}
\newcommand{\gaussk@}{\mathop{\vphantom{\sum}\mathpalette\bigcal@{K}}}

\newcommand{\bigcal@}[2]{%
  \vcenter{\m@th
    \sbox\z@{$#1\sum$}%
    \dimen@=\dimexpr\ht\z@+\dp\z@
    \hbox{\resizebox{!}{0.8\dimen@}{$\mathcal{K}$}}%
  }%
}
\newcommand{\cfracplus}{\mathbin{\cfracplus@}}
\newcommand{\cfracplus@}{%
  \sbox\z@{$\dfrac{1}{1}$}%
  \sbox\tw@{$+$}%
  \raisebox{\dimexpr\dp\tw@-\dp\z@\relax}{$+$}%
}
\newcommand{\cfracdots}{\mathord{\cfracdots@}}
\newcommand{\cfracdots@}{%
  \sbox\z@{$\dfrac{1}{1}$}%
  \sbox\tw@{$+$}%
  \raisebox{\dimexpr\dp\tw@-\dp\z@\relax}{$\cdots$}%
}

\makeatother

\fancypagestyle{plain}{%
\fancyhf{}
\fancyfoot[C]{\fontsize{10pt}{10pt}\selectfont\thepage}

}

\begin{document}

\title{The semi-analytic theory and computation of finite-depth
  standing water waves}

\titlerunning{Finite-depth standing water waves}

\author{Ahmad Abassi \and Jon Wilkening}

\institute{A. Abassi \at
  Department of Mathematics, University of California, Berkeley, CA 94720-3840\\
  \email{zaid\_abassi@berkeley.edu} \and
  J. Wilkening \at
  Department of Mathematics, University of California, Berkeley, CA 94720-3840\\
  \email{wilkening@berkeley.edu}
}

\date{\today}

\maketitle

\begin{abstract}
We propose a Stokes expansion ansatz for finite-depth standing water
waves in two dimensions and devise a recursive algorithm to compute
the expansion coefficients. We implement the algorithm on a
supercomputer using arbitrary-precision arithmetic.  The Stokes
expansion introduces hyperbolic terms that require exponentiation of
power series, which we handle efficiently using Bell polynomials.
Although exact resonances occur at a countable dense set of fluid
depths, we prove that for almost every depth, the divisors that arise
in the recurrence are bounded away from zero by a slowly decaying
function of the wave number.  A direct connection between small
divisors and imperfect bifurcations is observed. They are found to
activate secondary standing waves that oscillate non-uniformly in
space and time on top of the primary wave, with different amplitudes
and phases on each bifurcation branch.  We compute new families of
standing waves using a shooting method and find that Pad\'e
approximants of the Stokes expansion continue to converge to the
shooting method solutions at large amplitudes as new small divisors
enter the recurrence.  Closely spaced poles and zeros of the Pad\'e
approximants are observed, which suggests that the bifurcation branches
are separated by branch cuts.
\end{abstract}


\vspace*{0.6pc}

\thispagestyle{plain}
\fancyfoot[C]{\fontsize{10pt}{10pt}\selectfont\thepage}

\section{Introduction}\label{sec:intro}

Standing water waves have a long scientific history dating back at
least to 1831, when Faraday observed beautiful ink patterns at the
surface of milk driven by a tuning force. Standing waves in the ocean
are responsible for microseisms
\cite{higgins:50,ardhuin} and can play an important role in the
dynamics of wave breaking \cite{mcallister:24}. Their resonances must
be accounted for in the design of oscillating wave energy converters
\cite{renzi:12} and breakwaters \cite{peregrine:03} to maximize
efficiency and minimize violent impacts during storms. Two-dimensional
standing waves can be viewed as symmetric, time-periodic solutions of
the free-surface Euler equations in an enclosed container
\cite{penney:52} or as a superposition of identical
counter-propagating spatially periodic traveling waves
\cite{bridges2004,waterTS}.  Low-order perturbation expansion
techniques for standing waves were developed by Penney \&
Price~\cite{penney:52}, Tadjbakhsh \& Keller~\cite{tadjbakhsh}
(in the finite-depth case), Concus~\cite{concus:62,concus:64}
(considering the effects of surface tension), and Verma \& Keller
\cite{verma62} (considering standing waves in rectangular
  three-dimensional containers). Roberts \cite{roberts:83} and
Marchant \& Roberts \cite{marchant:87} carried out high-order
perturbation expansions numerically to study short-crested waves,
which include standing waves as a special case.

For the infinite-depth case with zero surface tension, Schwartz \&
Whitney \cite{schwartz1981semi} proposed an arbitrary-order
semi-analytic theory of standing waves using a conformal mapping
formulation of the equations. They computed their expansions to
twenty-ninth
order in quadruple-precision floating-point arithmetic. Amick \&
Toland \cite{amick1987semi} devised an ingenious implicit function
theorem argument to prove the `Schwartz and Whitney conjecture' that
their algorithm does not break down. They showed that although there
are infinitely many exact resonances, each resonant equation is
solvable the first time it enters the system, and the free parameter
associated with each resonance is uniquely determined by the
solvability condition at the next higher order. Iooss presented an
alternative proof of the Schwartz and Whitney conjecture based on
normal forms \cite{iooss:99} and generalized the results to the case
of several dominant modes.

An important open question that was not resolved by Amick and Toland
is whether the coefficients in these formal asymptotic expansions grow
slowly enough that the resulting series has a positive radius of
convergence.  In the present paper, we generalize the Schwartz and
Whitney algorithm to handle standing waves of finite depth and explore
the growth rates of the coefficients of the Stokes expansion through
numerical computation. In \cite{abassi:semi2}, we will show how to
include the effects of surface tension in infinite depth.  Each step
of these recursive algorithms involves computing forcing terms that
arise from lower-order terms in the expansion and dividing them by the
numbers
\begin{equation}\label{eq:lam:lamcap}
  \lambda_{p,j} = p\frac{\tanh(p\mu_0)}{\tanh\mu_0}-j^2 \qquad\quad
  \text{or} \qquad\quad
  \lambda^\text{cap}_{p,j} = \frac{1+Bp^2}{1+B}p-j^2,
\end{equation}
where the first formula is for the finite depth case while the second
is for gravity-capillary waves in infinite depth \cite{abassi:semi2}.
Here $\mu_0$ is a dimensionless fluid depth parameter,
$B=\frac{4\pi^2\tau}{\rho g L^2}$ is the (inverse) Bond number
\cite{shelton:stand}, $\tau$ is the surface tension, $\rho$ is the
fluid density, $L$ is the wavelength, $g$ is the acceleration of
gravity, and the integers $p$ and $j$ are the wave number and angular
frequency of the mode being computed (after non-dimensionalization).

In infinite depth with zero surface tension ($B=0$), every pair
$(p,j)$ with $p=j^2$ leads to a zero divisor that has to be treated
specially \cite{schwartz1981semi,amick1987semi,abassi:semi2}. Zero
divisors also arise at specific finite depths. Physically,
$\lambda_{p,j}=0$ means that within linear water wave theory, the
frequency of the $p$th spatial harmonic is an integer multiple, $j$,
of the fundamental frequency.  An interesting feature of standing
waves is that these resonant depths form a countable dense subset of
the positive real numbers \cite{concus:64}.  We prove that
there are no depths where the divisors $\lambda_{p,j}$ in equation
\eqref{eq:lam:lamcap} are uniformly bounded away from zero, but for
every $\delta>0$ and almost every fluid depth (in the Lebesgue sense),
there is an $a>0$ such that
$|\lambda_{p,j}|\ge\min\big(a,p^{-1/2-\delta}\big)$ for all $p\ge2$
and $j\in\mbb Z$. While this lower bound presumably does not ensure a
positive radius of convergence, it limits the growth rate of the
Stokes expansion coefficients sufficiently that Pad\'e approximants of
the Stokes expansion appear to be convergent at large amplitudes in
our numerical experiments. In the electronic supplementary material,
  we use a result from elliptic curve theory \cite{tate} to show
  that the density of resonant bond numbers does not imply that
  the divisors $\lambda_{p,j}^\text{cap}$ become arbitrarily small
  for every choice of $B$.

For both travelling waves \cite{chen:79,schwartz:79,chen:80,roberts:81}
and standing waves
\cite{vandenBroeck:84,mercer:94,smith:roberts:99,okamura:99,water2,rycroft:13},
harmonic resonance leads to non-uniqueness.  Combination waves
\cite{chen:79} with multiple dominant modes co-exist with pure waves
of one dominant mode, and there are often perfect or imperfect
bifurcations connecting the various families. The resulting branching
behavior of standing waves near resonant depths has been studied
extensively by Mercer \& Roberts \cite{mercer:94}, Smith \& Roberts
\cite{smith:roberts:99}, and Wilkening \& Yu \cite{water2}. In the
present work, we investigate the role of small divisors in the
formation of these bifurcation branches. We observe sudden changes in
the growth rate of the Stokes coefficients when especially small
divisors $\lambda_{p,j}$ enter the recursion. We investigate this in
detail for several fluid depths $\mu_0$.  For $\mu_0=1$, there is a
cluster of three small divisors that each yields an imperfect bifurcation
in the solution curve computed using a shooting method
\cite{mercer:92,water2}. Following the side branches associated with
the $(p,j)$ small divisor leads to visible secondary `standing waves
on standing waves' with $p$ spatial oscillations that execute $j$
temporal oscillations over one cycle of the primary wave. Similar
secondary waves have been reported in various settings
\cite{mercer:94,smith:roberts:99,okamura:99,water2,shelton:stand},
including standing waves in three-dimensional fluids
\cite{rycroft:13}. Solutions on the side branches differ in how the
amplitude and phase of the secondary wave matches up with the phase of
the primary wave. We explore the effects of nonlinearity on the shapes
of the secondary waves, which deviate from the sinusoidal patterns of
linear water wave theory that led to the small divisors.

In a model problem, Roberts \cite{roberts:81} showed that a nonlinear
Shanks transform can extend the validity of non-resonant asymptotic
expansions across discontinuities in the bifurcation curves associated
with nearby harmonic resonances. We adopt this strategy and study the
convergence of Pad\'e approximants of the Stokes expansions, which
continue to improve in accuracy (relative to the shooting method) as
more terms are included in their continued fraction representation,
even at large amplitudes where successive terms in the Stokes
expansion diverge wildly.  Poles in the Pad\'e approximation allow for
accurate branch jumping. We achieve errors between $10^{-32}$ and
$10^{-27}$ on both sides of the first two imperfect bifurcations we
observe in the $\mu_0=1$ case.  We use the poles to locate new
bifurcation branches and present a new method of identifying which
harmonic resonance is most strongly activated on each branch.  We
often find multiple Pad\'e poles in gaps between turning points
\cite{smith:roberts:99} of the wave height. This suggests that the
turning points are branch points and the poles on the branch cut act
as a quadrature formula to approximate a Cauchy integral with the same
branch point singularity structure at its endpoints
\cite{allen:75,stahl:97,yamada:14}.

\section{Preliminaries}
\label{FDSecConformalMap}

In this section we introduce the conformal map used to represent the
fluid motion in finite depth, non-dimensionalize the partial
differential equations governing water waves, propose an ansatz for a
Stokes expansion of the solution in powers of an amplitude parameter,
derive the governing equations of the spatial Fourier modes of the
solution, and show how to use Bell polynomials to efficiently
re-expand the hyperbolic sine or cosine of a power series.

\subsection{The conformal map and governing equations}

We consider standing waves on an inviscid, irrotational,
incompressible two-dimensional fluid of finite depth. We denote the
velocity potential in physical space by $\phi(x,y,t)$, where the fluid
velocity satisfies $\mb u=\nabla\phi$.  We identify $\mbb R^2$ with
the complex plane and parameterize the free surface and surface
velocity potential by
\begin{equation}\label{eq:eta:phi:start}
  \zeta(\alpha,t) = \xi(\alpha,t) + i\eta(\alpha,t) \quad \text{and} \quad
  \varphi(\alpha,t) = \phi\big(\xi(\alpha,t),\eta(\alpha,t),t\big).
\end{equation}
The kinematic condition and dynamic Bernoulli equation governing
the time evolution of the free surface are
\begin{equation}\label{eq:evol1}
  \begin{aligned}
    &\zeta_t\cdot\mb n = \mb u\cdot\mb n = \der\phi n, \\
    &\phi_t=  -\frac12|\nabla\phi|^2-g\eta+C(t),
  \end{aligned}
\end{equation}
where $\mb n$ is the outward normal to the free surface, $g$ is the
acceleration of gravity, the subscript $t$ is a partial derivative,
and $C(t)$ is an arbitrary function of time but not space, which
accounts for the fact that only gradients of the velocity potential
have physical significance. This term can be set to zero, but we find
that it is useful to retain it in the finite-depth problem. Here we
neglect the effects of surface tension, which would introduce a
curvature term in the Bernoulli equation; see \cite{abassi:semi2}. The
governing equation for the surface velocity potential is obtained from
the Bernoulli equation using
\begin{equation}\label{eq:convect}
  \varphi_t = \phi_t + \nabla\phi\cdot\zeta_t,
\end{equation}
where $\pa\phi/\pa n$ is computed from $\varphi$ by applying the
Dirichlet-Neumann operator \cite{craig:sulem}.

\begin{figure}[t]
  \begin{center}
    \includegraphics[scale=.62]{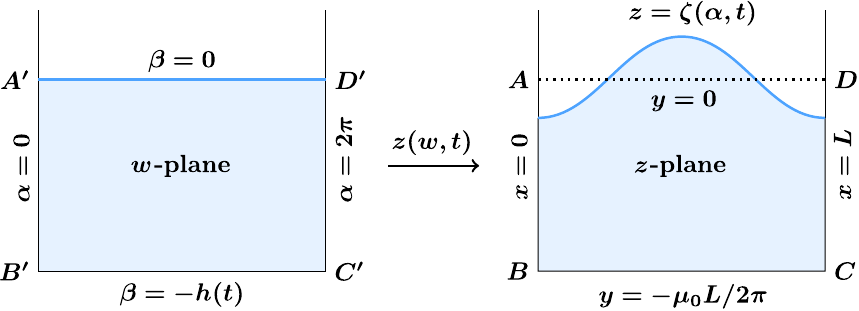}
  \end{center}
  \caption{The conformal map transforms conformal space (left)
    to physical space (right). The dotted line on the right
    illustrates a mean free-surface height of zero in physical space.}
  \label{FDConformalMapImage}
\end{figure}

We assume the standing wave and fluid velocity are spatially periodic
with period $L$ in physical space. In the infinite-depth case,
Schwartz \& Whitney \cite{schwartz1981semi}, Amick \& Toland
\cite{amick1987semi}, and Dyachenko et al
\cite{dyachenko1996nonlinear} introduce a conformal map to pull back
the fluid domain to the complex lower half-plane. We follow the same
plan, but since the fluid depth is finite, the pre-image of the
conformal map is a strip rather than a half-plane.  As illustrated in
figure~\ref{FDConformalMapImage}, we introduce
\begin{equation}\label{eq:w:alpha:beta}
  w = \alpha+i\beta
\end{equation}
as the spatial variable in conformal space and let
\begin{equation}
  z(w,t) = x(w,t) + iy(w,t) \quad \text{and} \quad \zeta(\alpha,t)=z(\alpha,t)
\end{equation}
denote the conformal map and its restriction to the real axis, which
is mapped to the free surface of the fluid.  We choose the period of
the conformal domain to be fixed at $2\pi$ and denote the lower
boundary of the strip by $\beta=-h(t)$, which evolves in time
\cite{ruban:2005,turner2016time,quasi:finite}.  We denote the
conjugate harmonic function to $\phi$ in physical space by $\psi$,
which is the stream function, and define the complex velocity
potential in conformal space by
\begin{equation}\label{eq:Phi:def}
  \Phi(w,t) = \phi\big(x(w,t),y(w,t),t\big) + i\psi\big(x(w,t),y(w,t),t\big).
\end{equation}
The Cauchy-Riemann equations give $\phi_y=-\psi_x$, so
\begin{equation}\label{eq:udotn}
  \mb u\cdot\mb n = (\phi_x,-\psi_x)\cdot\frac{(-\eta_\alpha,\xi_\alpha)}{s_\alpha} =
  -\frac{\im\{\Phi_w\}}{s_\alpha},
\end{equation}
where $s_\alpha=|\zeta_\alpha(\alpha,t)|$ is the arclength element of
the parameterization. From equations \eqref{eq:evol1}, \eqref{eq:convect},
\eqref{eq:Phi:def} and \eqref{eq:udotn}, we obtain
\begin{equation}\label{eq:gov:conf}
  \begin{aligned}
    \xi_\alpha\eta_t - \eta_\alpha\xi_t = -\im\{\Phi_w\}, \\
    \pa_t\re\{\Phi\} - \re\left\{\frac{\Phi_w}{z_w}\zeta_t\right\} +
    \frac12\left|\frac{\Phi_w}{z_w}\right|^2 + g\eta = C(t).
  \end{aligned}
\end{equation}
As shown in \cite{quasi:finite}, it follows from Cauchy's theorem and
the fact that $z_t/z_w$ is an analytic function in the strip
$-h(t)<\beta<0$ that
\begin{equation}\label{eq:ht}
  h_t = -\frac1{2\pi}\int_0^{2\pi} \frac{\im\{\Phi_w(\alpha,t)\}}{s_\alpha^2}\,d\alpha.
\end{equation}
In the construction of this paper, it is not necessary to impose equation
\eqref{eq:ht} explicitly since it follows from equation \eqref{eq:gov:conf} and
Cauchy's theorem.

We are searching for standing water waves, which are symmetric
time-periodic solutions of equation \eqref{eq:gov:conf}. Let $T$ denote the
temporal period. Following Schwartz \& Whitney \cite{schwartz1981semi},
we non-dimensionalize the variables via
\begin{equation}\label{eq:ZFW:def}
  Z(w,t) = \frac{2\pi}L z\left(w,\frac T{2\pi}t\right), \qquad
  F(w,t) = \frac{2\pi T}{L^2}\Phi\left( w,\frac T{2\pi}t\right), \qquad
  S = \frac{gT^2}{2\pi L}
\end{equation}
and introduce an auxiliary function, $W(w,t)=F_w/Z_w$, which, when
conjugated, is a dimensionless velocity pulled back from physical to
conformal space,
\begin{equation}
  \overline W(w,t) = \frac TL \mb u\big(x(w,t),y(w,t),t\big).
\end{equation}
Expressing equation \eqref{eq:gov:conf} in terms of the dimensionless variables
gives
\begin{subequations}\label{eq:aux:kin:ber}
  \begin{alignat}{3}
    \label{fdCRE}
    -F_w + WZ_w &= 0, & \; \; \; \; &\text{} & \; \; \; \big({-}h(t) \leq \beta 
      & \leq 
      0 \big)
    \text{,} \\
    \label{fdKE}
    \im{ \lp F_{\alpha}-Z_{\alpha} \overline{Z}_t \rp } &= 0, & \; \; \; 
    \; 
    &\text{} & \; \; \; 
    \big(\beta&=0 \big)
    \text{,} \\
    \label{fdBE}
    \re \lp F_t+\frac{1}{2} W \overline{W}-iSZ-WZ_t \rp &= C(t), & \; \; 
    \;\; 
    &\text{} & \; 
    \;\; 
    \big(\beta&=0 \big) \text{.}
  \end{alignat}
\end{subequations}

\subsection{The Stokes expansion ansatz} \label{FDSecAnsatzStokes}

Building on the infinite-depth conformal mapping framework of Schwartz
\& Whitney \cite{schwartz1981semi,amick1987semi} and finite-depth
graph-based formulations \cite{tadjbakhsh,marchant:87,okamura:99}, we
propose the following ansatz for the Fourier representations of $Z$,
$W$, and $F$:
\begin{subequations}\label{ansatz}
  \begin{align}
    \label{ansatzZ}
    Z(w,t) &= w +ih(t) -i\mu_0 + \sum_{p=1}^{\infty} 
    a_p(t)\frac{\sin[p(w+ih(t))]}{\cosh(p\mu_0)}, \\
    \label{ansatzW}
    W(w,t) &= \sum_{p=1}^{\infty} 
    b_p(t)\frac{\sin[p(w+ih(t))]}{\cosh(p\mu_0)}, \\
    \label{ansatzF}
    F(w,t) &= \sum_{p=0}^{\infty} 
    c_p(t)\frac{\cos[p(w+ih(t))]}{\cosh(p\mu_0)}.
  \end{align}
\end{subequations}
On the bottom boundary, where $w=\alpha-ih(t)$, we have
\begin{equation}
  \frac{2\pi}L\im\{z(w,t)\} = \im\{Z(w,t)\}=-\mu_0, \quad\;
  \im\{W(w,t)\}=0, \quad\;
  \im\{F(w,t)\}=0.
\end{equation}
This shows that $\mu_0L/2\pi$ is the fluid depth in physical space,
that the vertical component of velocity is zero on the bottom
boundary, and that the stream function is constant (in fact zero) on
the bottom boundary, indicating that there is no fluid flux crossing
this boundary.

We employ identical $\epsilon$-expansions to those of Schwartz \&
Whitney \cite{schwartz1981semi,amick1987semi} for the coefficients
$a_p(t)$, $b_p(t)$ and $c_p(t)$, and for the dimensionless period
parameter $S$.  However, in the finite-depth setting we must also
expand the strip width $h(t)$ in terms of the standing wave amplitude:
\begin{subequations} \label{Stokes}
  \begin{gather}
    \label{stokesA}
    a_p(t) = \sum_{n=0}^{\infty} \alpha_{p,n}(t) \epsilon^{p+2n}, \quad
    b_p(t) = \sum_{n=0}^{\infty} \beta_{p,n}(t) \epsilon^{p+2n}, \quad
    c_p(t) = \sum_{n=0}^{\infty} \gamma_{p,n}(t) \epsilon^{p+2n}, \\
    \label{stokesS}
    S = \sum_{n=0}^{\infty} \sigma_{n} \epsilon^{2n}, \qquad
    h(t) = \mu_0 + \sum_{n=1}^{\infty} \mu_{n}(t) \epsilon^{2n}.
  \end{gather}
\end{subequations}
Here $\epsilon$ is the standing wave amplitude; the functions
$\alpha_{p,n}(t)$, $\beta_{p,n}(t)$, $\gamma_{p,n}(t)$, and
$\mu_{n}(t)$ are real-valued; the coefficients $\sigma_n$ are real
constants; and the coefficient $\mu_0$ is a positive constant, the
non-dimensionalized fluid depth in physical space (see figure
  \ref{FDConformalMapImage}).  As in the infinite-depth problem,
the amplitude is defined as half the vertical crest to trough height
after non-dimensionalization, i.e.,
\begin{equation}\label{eq:eps:def}
  \epsilon = \frac{1}{2} \im \Big( Z(0,0)-Z(\pi,0) \Big) =
  \sum_{m=0}^{\infty} a_{2m+1}(0)\frac{\sinh[(2m+1) 
		h(0)]}{\cosh[(2m+1)\mu_0]},
\end{equation} 
where the right-hand side must still be expanded in powers of
$\epsilon$ with the coefficient of the linear term evaluating to 1 and
the others evaluating to 0.

On the free surface, the sine and cosine terms in equation \eqref{ansatz} may
be written
\begin{equation}\label{eq:trig:idents}
  \begin{aligned}
    \sin[p(\alpha+ih)] &= \cosh(ph)\sin(p\alpha) + i\sinh(ph)\cos(p\alpha), \\
    \cos[p(\alpha+ih)] &= \cosh(ph)\cos(p\alpha) - i\sinh(ph)\sin(p\alpha),
  \end{aligned}
\end{equation}
so the spatial Fourier coefficients of the real and imaginary parts of
$Z(\alpha,t)$, $W(\alpha,t)$ and $F(\alpha,t)$ involve products
of the form
\begin{equation}\label{eq:re:expand}
  \frac{u_p(t)\cosh[ph(t)]}{\cosh(p\mu_0)} \qquad \text{ or } \qquad
  \frac{u_p(t)\sinh[ph(t)]}{\cosh(p\mu_0)},
\end{equation}
where $u_p(t)$ represents $a_p(t)$, $b_p(t)$ or $c_p(t)$.  One
advantage of the conformal mapping approach over previous graph-based
formulations \cite{penney:52,tadjbakhsh,marchant:87,okamura:99} is
that the argument $ph(t)$ of the hyperbolic functions in equation
\eqref{eq:re:expand} is an $\epsilon$-expansion with terms depending
only on time (and not also space); see \S\ref{FDBellPolynomials}
below.  An advantage of the graph-based formulation is that it extends
to three-dimensional short-crested waves \cite{marchant:87}, covering
standing waves as a special case.

\subsection{Time-evolution of the spatial Fourier modes}
\label{sec:time:evol:fourier}

Substitution of the ansatz \eqref{ansatz} in the auxiliary
equation \eqref{fdCRE} gives
\begin{equation} \label{CRAnsatzed}
  \begin{split}
    \frac{(pc_p+b_p)}{\cosh(p\mu_0)}  &+ \sum_{k=1}^{p-1} 
    \frac{k a_{k} b_{p-k}}{2\cosh(k\mu_0)\cosh\!\big[(p-k)\mu_0\big]}
    \\ & \quad + \sum_{k=1}^{\infty} 
    \frac{k a_kb_{p+k}-(k+p)a_{k+p}b_k}{2\cosh(k\mu_0)\cosh\!\big[(p+k)\mu_0\big]}
    =0,
  \end{split} \qquad
  \Big( p\in\mbb N \Big).
\end{equation}
Similarly, the kinematic free-surface equation \eqref{fdKE} gives
\begin{subequations}
  \begin{equation}	\label{fdKFSp0}
    \begin{aligned} 
      \dot{h} + \sum_{k=1}^{\infty} 
      ka_k\dot{a}_k\frac{\sinh(2kh)}{2\cosh^2(k \mu_0)} +\dot{h} 
      \sum_{k=1}^{\infty} k^2a^2_k\frac{\cosh(2kh)}{2\cosh^2(k\mu_0)} = 0
    \end{aligned}
  \end{equation}
  and
  \begin{equation}	\label{fdKFSp1}
    \begin{aligned}
      & (\dot{a}_p - pc_p )\frac{\sinh(ph)}{\cosh(p\mu_0)}
      + 2\dot{h} p a_p\frac{  \cosh(ph)}{\cosh(p\mu_0)}
      - \sum_{k=1}^{p-1}  \frac{ (p-k) a_{p-k} \dot{a}_k\,
	\sinh\!\big[(p-2k)h\big]
      }{2\cosh\!\big[(p-k)\mu_0\big]\cosh(k\mu_0)}
      \\ & \;\; + \sum_{k=1}^{\infty} 
      \Big(\lp p+k\rp 
        a_{p+k}\dot{a}_k+ka_k\dot{a}_{p+k}\Big)\frac{\sinh\!\big[(p+2k)h\big]}{2\cosh(k\mu_0)
        \cosh\!\big[(p+k)\mu_0\big]}
      \\ & \;\; + \dot{h} \sum_{k=1}^{\infty} 
      \Big((p+k)a_{p+k}ka_k+ka_k(p+k)a_{p+k}\Big)
      \frac{\cosh\!\big[(p+2k)h\big]}{2\cosh(k\mu_0)\cosh\!\big[(p+k)\mu_0\big]}
      \\ & \;\; +\dot{h} \sum_{k=1}^{p-1} 
	 (p-k)a_{p-k}ka_k\frac{\cosh\!\big[(p-2k)h\big]}{2\cosh\!\big[(p-k)\mu_0\big]\cosh(k\mu_0)} = 0,
         \qquad \Big(p\in\mbb N\Big),
    \end{aligned}
  \end{equation}
\end{subequations}
where a dot represents a time-derivative.  In the Bernoulli equation
\eqref{fdBE}, we choose the integration constant $C(t)$ so that
$\dot{c}_0=0$, which allows us to set $c_0(t)=0$ and commence the
series for $F$ in equation (\ref{ansatzF}) from $p=1$. This leads to
\begin{equation} \label{BernoulliAnsatzed}
  \begin{aligned}
    & \dot{c}_p\frac{\cosh(ph)}{\cosh(p\mu_0)}
    + S a_p\frac{ \sinh(ph)}{\cosh(p\mu_0)} 
    + \sum_{k=1}^{\infty} \frac{ b_k b_{p+k}\,
      \cosh\!\big[(p+2k)h\big] 
      }{2\cosh(k\mu_0)\cosh\!\big[(p+k)\mu_0\big]} \\
    & \quad - \sum_{k=1}^{p-1}
    \frac{ b_{p-k}b_k\, \cosh\!\big[(p-2k)h\big] 
      }{4\cosh\!\big[(p-k)\mu_0\big]\cosh(k\mu_0)}
    -\sum_{k=1}^{\infty} \frac{b_{p+k} \dot{a}_k+b_k\dot{a}_{p+k} 
    }{2\cosh(k\mu_0)\cosh\!\big[(p+k)\mu_0\big]} \cosh(ph)
    \\ & \quad + \sum_{k=1}^{p-1} \frac{b_{p-k} \dot{a}_k 
    }{2\cosh\!\big[(p-k)\mu_0\big]\cosh(k\mu_0)}\cosh(ph) = 0, \qquad \Big(p\in\mbb N\Big).
  \end{aligned}
\end{equation}
The $p=0$ term in the ansatz $Z(w,t)=w+i\sum_{p=0}^\infty a_p(t)
e^{-ipw}$ for the infinite-depth case
\cite{amick1987semi,abassi:semi2} has been replaced by $[h(t)-\mu_0]$
in equation \eqref{ansatzZ}. We only need to solve for $a_p(t)$, $b_p(t)$ and
$c_p(t)$ with $p\ge1$ since we set $c_0(t)=0$ above and $b_0(t)$ is
absent in the ansatz \eqref{ansatzW} due to $\sin(0)=0$.

\subsection{Bell polynomials and the exponential of a power series}	
\label{FDBellPolynomials}

In the equations of the previous section, there appear terms involving
the hyperbolic sine and cosine of integer multiples of the strip
width, $h(t)$, which has a Stokes expansion in powers of
$\epsilon$. An efficient formula \cite{bell1934exponential} to
re-expand the exponential of a power series is given by
\begin{equation} \label{BellExponentialPowerSeries}
	\exp \lp \sum_{k=0}^{\infty} a_k x^k \rp = 
	e^{a_0}\sum_{n=0}^{\infty} \frac{B_n(a_1 1!,\ldots,a_n n!)}{n!} 
	x^n,
\end{equation}
where the complete Bell polynomials $B_n(x_1,\dots,x_n)$ are defined
recursively by
\begin{equation} \label{completeBellPolynomials}
  B_0=1, \qquad B_{n+1}(x_1, \ldots,x_{n+1}) = \sum_{i=0}^{n} \binom{n}{i}
  B_{n-i} (x_1,\ldots,x_{n-i}) x_{i+1}, \quad (n\ge0).
\end{equation}
For our specific setting, we need $\cosh(qh)$ and $\sinh(qh)$
for various integers $q \in \mbb{Z}$, so we expand
\begin{equation}\label{eq:cqh:sqh}
  \begin{alignedat}{2}
  \exp(qh) &= \sum_{n=0}^\infty B_{q,n}(t)\epsilon^{2n},\qquad &
  &B_{q,n}(t)=\frac{e^{q\mu_0}}{n!}B_n\big( q\mu_1(t)1!,\dots,q\mu_n(t)n!\big), \\
  \cosh(qh) &= \sum_{n=0}^{\infty} c_{q,n}(t) \epsilon ^{2n}, \qquad &
  &\sinh(qh) = \sum_{n=0}^{\infty} s_{q,n}(t) \epsilon ^{2n},
  \end{alignedat}
\end{equation}
where
\begin{equation}\label{eq:cqn:sqn:bell}
  c_{q,n}(t) = \frac{B_{q,n}(t)+B_{-q,n}(t)}{2}, \qquad
  s_{q,n}(t) = \frac{B_{q,n}(t)-B_{-q,n}(t)}{2}.
\end{equation}
Using equation \eqref{completeBellPolynomials}, we obtain
\begin{equation}\label{eq:Bqn:recur}
  B_{q,0}(t) = e^{q\mu_0}, \qquad
  B_{q,n}(t) = q\sum_{i=1}^n \frac{i}{n}\,B_{q,n-i}(t)\,\mu_i(t), \quad (n\ge1).
\end{equation}
Roberts \cite{roberts:83} and Marchant \& Roberts \cite{marchant:87}
derived an identical recurrence from first principles (without
  employing Bell polynomials) in a graph-based formulation of the
short-crested wave problem. In the special case of standing waves in
this graph-based approach, one has to evaluate factors of $e^{qy}$
(infinite-depth) or $\cosh[q(\mu_0+y)]$ (finite-depth) in the velocity
potential expansion at $y=\eta^\text{graph}(x,t)=
\sum_{\nu\ge1}\eta^\text{graph}_\nu(x,t)\,\epsilon^\nu$.  Replacing
$n$ by $\nu$ and $\mu_i(t)$ by $\eta^\text{graph}_i(x,t)$ in equation
\eqref{eq:Bqn:recur} and calling the result $\wtil B_{q,\nu}(x,t)$
leads to a function that depends on both $x$ and $t$.  $B_{q,n}(t)$ is
represented by $O(n)$ temporal Fourier coefficients while $\wtil
B_{q,\nu}(x,t)$ contains $O(\nu^2)$ non-zero Fourier modes in both
time and space. Only even powers of $\epsilon$ are present in equation
\eqref{eq:cqh:sqh}, so the sum \eqref{eq:Bqn:recur} contains half as
many terms as the corresponding sum for $\wtil B_{q,\nu}(x,t)$ at a
given order $\epsilon^\nu$ with $\nu=2n$. This reduces the memory and
computational costs of the data structures required to re-expand the
hyperbolic functions in the conformal mapping approach of the present
work.

\section{ODEs for the Stokes coefficients and a recursive algorithm}
\label{FDSecSolutionAlgorithm}

Substitution of the Stokes expansions \eqref{Stokes} into the
equations \eqref{CRAnsatzed}--\eqref{BernoulliAnsatzed} governing the
time-evolution of the spatial Fourier modes yields a system of ODEs
for the Stokes coefficients,
\begin{gather}
  \dot{\mu}_{n} + T^1_{0,n} = 0,	\tag{I} \label{eq:I} \\
  \beta_{p,n}+p\gamma_{p,n} + T^2_{p,n} = 0,	\tag{II} \label{eq:II} \\
  \dot{\alpha}_{p,n} - p\gamma_{p,n} + T^3_{p,n} = 0,	\tag{III} \label{eq:III} \\
  \dot{\gamma}_{p,n} + \sigma_0\tanh(p\mu_0) \alpha_{p,n} + T^4_{p,n} = 
  0,
  \tag{IV} \label{eq:IV}
\end{gather}
where $p\ge1$ and $n\ge0$.  Formulas for the forcing terms $T^r_{p,n}$
are derived in the electronic supplementary material.  We require that
solutions of this system have certain symmetries and functional forms,
namely, that $\mu_n(t)$ and $\alpha_{p,n}(t)$ are even trigonometric
polynomials of the form
\begin{alignat}{2}
  \label{fdFourierSeriesMu}
  \mu_n(t) &= \sum_{j\in E_{2n}} \mu_{n,j}e^{ijt} =
  \sideset{}{'}\sum_{j=0}^{2n} \mu_{n,j} H_j \lp e^{ijt} + e^{-ijt} \rp, & \qquad
  &\big(\,\mu_{n,j}\in\mbb R \, \big), \\
  \label{fdFourierSeriesAlpha}
  \alpha_{p,n}(t) &= \sum_{j\in E_{p+2n}} \alpha_{p,n,j}e^{ijt} =
  \sideset{}{'}\sum_{j=0}^{p+2n} \alpha_{p,n,j} H_j \lp e^{ijt} + e^{-ijt} \rp, & \qquad
  &\big( \, \alpha_{p,n,j} \in \mbb R \, \big),
\end{alignat}
where $\mu_{n,-j}=\mu_{n,j}$,\, $\alpha_{p,n,-j}=\alpha_{p,n,j}$,
\begin{equation} \label{eq:Enu:Hj:def}
  \begin{aligned}
    E_\nu &= \big\{\, \nu-2m\;\,\big\vert\;\,0\le m\le\nu \,\big\} \\[-2pt]
    &= \big\{\,{-}\nu,\,-\nu+2,\,\dots,\,\nu-2,\,\nu \,\big\},
    \end{aligned} \qquad\quad
  H_j = \begin{cases} 1/2, & j=0, \\ 1, & j\ge1, \end{cases}
\end{equation}
and a prime on a sum indicates that terms in the given range should be
included only if the summation index has the same parity as the upper
limit. Moreover, $\beta_{p,n}(t)$ and $\gamma_{p,n}(t)$ are odd
trigonometric polynomials of the form
\begin{alignat}{2}
  \label{fdFourierSeriesBeta}
  \beta_{p,n}(t) &= \sum_{j\in E_{p+2n}} i\beta_{p,n,j}e^{ijt} =
  i \sideset{}{'}\sum_{j=1}^{p+2n} \beta_{p,n,j} \lp e^{ijt} - e^{-ijt} \rp, & \qquad
  & \big( \, \beta_{p,n,j} \in \mbb R \, \big), \\
  \label{fdFourierSeriesGamma}
  \gamma_{p,n}(t) &= \sum_{j\in E_{p+2n}} i\gamma_{p,n,j}e^{ijt} =
  i \sideset{}{'}\sum_{j=1}^{p+2n} \gamma_{p,n,j} \lp e^{ijt} - e^{-ijt} \rp, &\qquad
  &\big( \, \gamma_{p,n,j} \in \mbb R \, \big),
\end{alignat}
where $\beta_{p,n,0}=\gamma_{p,n,0}=0$,\,
$\beta_{p,n,-j}=-\beta_{p,n,j}$,\, and\,
$\gamma_{p,n,-j}=-\gamma_{p,n,j}$.  The symmetry assumptions
\eqref{fdFourierSeriesMu}, \eqref{fdFourierSeriesAlpha},
\eqref{fdFourierSeriesBeta} and \eqref{fdFourierSeriesGamma} mostly
take the place of initial conditions for the ODEs, but we also need to
impose
\vspace*{-10pt}
\begin{gather}
	\alpha_{1,0}(0)=\coth(\mu_0), \qquad
	\sum_{q=0}^{n}\sum_{k=0}^{n-q} 
	\frac{\alpha_{2q+1,k}(0)s_{2q+1,n-q-k}(0)}{\cosh\!\big[(2q+1)\mu_0\big]}
	= 0, \quad \big(n \in \mbb{N}\big),
	\tag{i} \label{eq:i} \\
	\mu_n(0) + \sum_{q=1}^{n} \sum_{k=0}^{n-q} \sum_{l=0}^{n-q-k} 
	\frac{q}{4\cosh^2(q\mu_0)} \alpha_{q,k}(0)\alpha_{q,l}(0) s_{2q,n-q-k-l}(0) 
	= 0, \quad \big(n \in \mbb{N} \big).	\tag{ii} \label{eq:ii}
\end{gather}
Here \eqref{eq:i} is a consequence of the amplitude definition
(\ref{eq:eps:def}) and \eqref{eq:ii} ensures that the fluid depth is
independent of $\epsilon$. As shown in
figure~\ref{FDConformalMapImage} above, the bottom boundary is at
$y=-\mu_0L/2\pi$, but we also need to specify the mean free-surface
height. It is easy to show that $(\partial/\partial
  t)\int_0^{2\pi}\eta\xi_\alpha\,d\alpha=0$, so mass is conserved in
time and the mean free-surface height remains zero if it is zero
initially. We obtain \eqref{eq:ii} using equation \eqref{ansatzZ}
together with \eqref{eq:trig:idents} in $\int_0^{2\pi}
\im\{Z(\alpha,0)\}\re\{Z_w(\alpha,0)\}\,d\alpha=0$.

At this stage, following \cite{amick1987semi} for the infinite-depth
case, it is useful to replace \eqref{eq:II} and \eqref{eq:III} by the
equivalent conditions
\begin{align}
  \ddot{\alpha}_{p,n} + p \sigma_0 \tanh(p\mu_0) \alpha_{p,n} + S_{p,n} &= 0, \qquad
  \big( S_{p,n} = \dot{T}^3_{p,n}+pT^4_{p,n} \big),	\tag{II*} \label{eq:II:star} \\
  \beta_{p,n}+\dot{\alpha}_{p,n} + T^2_{p,n} + T^3_{p,n} &= 0.	\tag{III*} \label{eq:III:star}
\end{align}
This allows us to solve \eqref{eq:II:star}, \eqref{eq:III:star} and
\eqref{eq:IV} sequentially to obtain $\alpha_{p,n}$, $\beta_{p,n}$ and
$\gamma_{p,n}$, respectively.  Here we have eliminated $\gamma_{p,n}$
from the equations for $\alpha_{p,n}$ and $\beta_{p,n}$, though
lower-order terms $\gamma_{p,j}$ with $j<n$ appear in the formulas for
the forces $T^3_{p,n}$ and $T^4_{p,n}$.

Our goal in the remainder of this section is to demonstrate the
existence of a solution of \eqref{eq:I}, \eqref{eq:II:star},
\eqref{eq:III:star}, \eqref{eq:IV}, \eqref{eq:i} and \eqref{eq:ii} by
proposing an algorithm in the spirit of \cite{amick1987semi}.  Let
\begin{equation}	\label{fdOmegaPN}
  \Omega_{p,n} = \Big\{(q,m) \;\,\Big\vert\;\, q\ge1, \;
  0 \leq m \leq n, \; q+m \leq p+n, \; (q,m) \neq (p,n) \Big\},
\end{equation}
which are the integer lattice points in the region shown in figure
\ref{FDOmegaSetImage}.  We also define the following sets of functions
and real numbers for $n\in\mbb{N}$:
\begin{equation}	\label{fdMuN}
  \mc M_n = \big\{\mu_0,\,\mu_1(t),\dots,\mu_n(t)\big\}, \qquad
  \Sigma_n = \big\{\sigma_0,\dots,\sigma_n\big\}.
\end{equation}
Each $\mu_k(t)$ is required to be of the form
\eqref{fdFourierSeriesMu}, and $\mu_0$ is a given parameter of the
problem statement, namely, the fluid depth in physical space after
non-dimensionalization.

\begin{figure}[t]
  \begin{center}
    \includegraphics[scale=0.62]{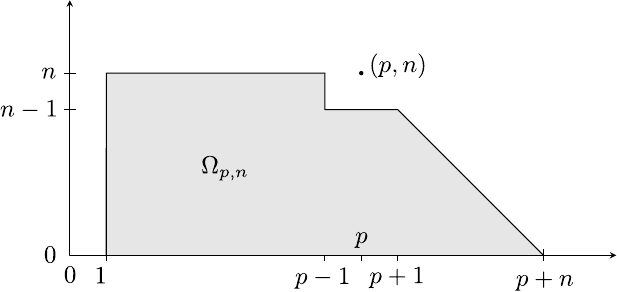}
  \end{center}
  \caption{The set $\Omega_{p,n}$ consists of the integer lattice points 
    in the shaded region, including its boundary.}
  \label{FDOmegaSetImage}
\end{figure}

Similar to the infinite-depth case \cite{amick1987semi}, the forcing
terms $T^1_{0,n}$, $T^2_{p,n}$, $T^3_{p,n}$, and $T^4_{p,n}$ are
functions of $\alpha_{q,m}$, $\beta_{q,m}$, and $\gamma_{q,m}$ for
$(q,m) \in \Omega_{p,n}$. In addition, $T^1_{0,n}$ depends on $\mc
M_{n-1}$; $T^3_{p,n}$ and $T^4_{p,n}$ depend on $\mc M_n$; and
$T^4_{p,n}$ depends on $\Sigma_n$. To say that $\Omega_{p,n}$ is known
means that $\alpha_{q,m}$, $\beta_{q,m}$, and $\gamma_{q,m}$ are known
for all $(q,m) \in \Omega_{p,n}$.  When $n=0$, we have for $p\ge2$
\begin{equation} \label{fdOmegaP0}
	\Omega_{p,0} = \big\{(1,0),\dots,(p-1,0)\big\}.
\end{equation}
For convenience and consistency, we define $\Omega_{0,n} =
\Omega_{1,0}=\emptyset$ for $n\ge0$.  In particular, for
$(p,n)=(1,0)$, the forcing terms $T^r_{1,0}$ are zero for
$r=2,3,4$. This allows us to immediately solve \eqref{eq:II:star},
\eqref{eq:III:star}, \eqref{eq:IV} and \eqref{eq:i} to conclude that
\begin{equation}\label{eq:pn:10}
  \begin{alignedat}{2}
    \sigma_0 &= \coth(\mu_0), & \qquad \beta_{1,0} &= \coth(\mu_0)\sin(t), \\
    \alpha_{1,0} &= \coth(\mu_0)\cos(t), & \qquad \gamma_{1,0} &= -\coth(\mu_0)\sin(t).
  \end{alignedat}
\end{equation}
From this point, $\sigma_0$ will be considered a known constant of the
problem.

\subsection{Non-resonant depths and small divisors}
\label{sec:res:depths}

Equation \eqref{eq:II:star} is a second-order, linear, non-homogeneous
ODE with constant coefficients, much like the analogous equation (IV*)
in \cite{amick1987semi}.  The forcing term $S_{p,n}$ in
\eqref{eq:II:star} will be shown to be an even trigonometric
polynomial of the form $\sum_{j\in E_{p+2n}}S_{p,n,j}e^{ijt}$, where
$E_\nu$ was defined above in equation \eqref{eq:Enu:Hj:def} and
$S_{p,n,j}=pT^4_{p,n,j}-jT^3_{p,n,j}$ for $j\in E_{p+2n}$. Thus,
\eqref{eq:II:star} may be written
\begin{equation}\label{eq:lam:pj:def}
  \lambda_{p,j}\alpha_{p,n,j} = -S_{p,n,j}, \qquad
  \lambda_{p,j} = p\frac{\tanh(p\mu_0)}{\tanh\mu_0} - j^2, \qquad
  \big(j\in E_{p+2n}\big).
\end{equation}
It has a unique solution of the form \eqref{fdFourierSeriesAlpha}
provided that $\lambda_{p,j}\ne0$ for $j\in E_{p+2n}$. Here $p\ge1$,
$n\ge0$, and we can restrict attention to $j\ge0$ since
$\alpha_{p,n,-j}=\alpha_{p,n,j}$.  We always have $\lambda_{1,1}=0$,
which is a special case that determines $\sigma_n$ for $n\ge1$, as
shown in the proof of Lemma~\ref{fdPNInductionStep} below.  Since
$1<\frac{\tanh(p\mu_0)}{\tanh\mu_0}< p$ for $p\ge2$ and
$\mu_0\in(0,\infty)$, any solution of $\lambda_{p,j}=0$ with $p\ge2$
satisfies $\sqrt{p}< j< p$. It follows from $j<p$ that if $j$ and
$p$ have the same parity, then $j\in E_{p+2n}$ for all $n\ge0$. Thus,
$S_{p,n}=\sum_{l\in E_{p+2n}}S_{p,n,l}e^{ilt}$ contains the term
$S_{p,n,j}e^{ijt}$ already when $n=0$.  (This is not true for
  gravity-capillary waves \cite{abassi:semi2}, where
  $\lambda^\text{cap}_{p,j}$ in \eqref{eq:lam:lamcap} can be zero with
  $j>p$.) The Stokes expansion ansatz \eqref{Stokes} is expected to
break down at resonant depths since it would be surprising if a
cancellation caused $S_{p,0,j}=0$ to occur at exactly the same depth
that led to $\lambda_{p,j}=0$.

In the infinite depth case with zero surface tension,
$\lambda_{p,j}=p-j^2$ is zero whenever $p\ge2$ is a perfect square and
$j=\sqrt p$. Nevertheless, imposing compatibility conditions leads to
existence and uniqueness of a formal expansion solution.  This is a
key point and challenge in the work of \cite{amick1987semi}.  In the
finite-depth case, Concus \cite{concus:64} proved that for any real
interval $(a,b)$ with $0<a<b<\infty$, there exists a $\mu_0\in(a,b)$
and integers $p\ge2$, $j\ge1$ such that $\lambda_{p,j}$ in equation
\eqref{eq:lam:pj:def} is zero. His proof is easily adapted to produce
$j$ and $p$ of the same parity. It follows that the resulting resonant
depths are dense in the positive real numbers.  They are enumerated by
$p\ge5$ and $\sqrt p<j<p$ with $p-j$ even since
$\tanh(p\mu_0)/\tanh\mu_0$ decreases monotonically from $p$ to 1 as a
function of $\mu_0\in(0,\infty)$.  Tables containing the first several
resonant depths with this enumeration are given in \cite{marchant:87}.

The complement of this countable dense set of resonant depths has full
Lebesgue measure, and consists of depths $\mu_0$ for which the
recursive algorithm described below will not lead to a division by
zero at any order. For these non-resonant depths, it is desirable to
know how small the $\lambda_{p,j}$ may become and what effect such
small divisors have on the recursive solution.  Let us define
\begin{equation}\label{eq:lam:p:def}
  \lambda_p = \min_{j\in p+2\mbb Z} \big| \lambda_{p,j} \big|.
\end{equation}
When necessary for clarity, we will write $\lambda_p(\mu_0)$ and
$\lambda_{p,j}(\mu_0)$.  The next lemma and theorem show that all
rational depths are non-resonant, and almost every depth $\mu_0$ leads
to a sequence $\{\lambda_p(\mu_0)\}_{p=2}^\infty$ that is bounded
below by a slowly decaying function of $p$.

\begin{lemma} \label{fdTranscendentalNumberTheoryLemma}
  Let $\mu_0$ be a positive algebraic number and let $p\ge2$ be an integer.
  Then $p\frac{\tanh(p\mu_0)}{\tanh\mu_0}$ is transcendental and $\lambda_{p,j}$
  in equation \eqref{eq:lam:pj:def} is non-zero for all integers $j$.
\end{lemma}

\begin{proof}
  Suppose $\mu_0$ and $p$ satisfy the hypotheses and that
  $r=p\frac{\tanh(p\mu_0)}{\tanh\mu_0}=p \frac{e^{2\mu_0}+1}{e^{2\mu_0}-1} 
    \frac{e^{2p\mu_0}-1}{e^{2p\mu_0}+1}$ is an algebraic number.
  After rearranging, we obtain
  \begin{gather*}
    (r-p)e^{2(p+1)\mu_0}-(r+p)e^{2p\mu_0}+(r+p)e^{2\mu_0}+(p-r)e^0=0.
  \end{gather*}
  All four exponents are distinct algebraic numbers. By the
  Lindemann-Weierstrass theorem \cite{baker:tnt}, the coefficients of
  the exponentials are zero, implying that $p=r=-r$, a contradiction
  to $p\ge2$. So $r$ is transcendental and there is no integer $j$
  satisfying $r=j^2$. \hspace*{\fill} \raisebox{-3pt}{\qed}
\end{proof}

\begin{theorem} \label{thm:nonresonant}
  For each $\delta>0$, the set
  \begin{equation}\label{eq:E:delta:def}
    \mc E_\delta = \Big\{\mu_0>0 \;\,\Big\vert\;\, \exists \; a>0 \;\,
    \text{such that} \;\, \forall\;p\ge2, \; 
    \lambda_{p}(\mu_0)\ge \min\big(a,p^{-\frac12-\delta}\big) \Big\}
  \end{equation}
  has full Lebesgue measure. If $\delta>\frac12$ and $\mu_0>0$ is
  rational, then $\mu_0\in\mc E_\delta$. For $\delta\le0$, $\mc
  E_\delta$ has Lebesgue measure 0.  For $\delta\le-\frac12$, $\mc
  E_\delta$ is the empty set.
\end{theorem}

We prove this theorem in the electronic supplementary material and
outline the key steps of the proof here. The first assertion makes
precise the claim that for almost every fluid depth, $\min_{2\le q\le
  p}\lambda_q$ is positive for $p\ge2$ and does not decay to zero much
faster than $1/\sqrt{p}$ as $p\to\infty$.  To prove it, we show that
$\mu_0\not\in\mc E_\delta\;\Rightarrow\; \tanh\mu_0\in\mc F_\delta$,
where
\begin{equation}\label{eq:Fdelta:def:0}
  \mc F_\delta = \Big\{ x\in\mbb R \;\,\Big\vert\;\, \exists \;\,
  \text{infinitely many pairs}
  \;\,(p,j)\in\mbb Z\times\mbb N \;\,\text{s.t.} \;\,
  \Big| x - \frac{p}{j^2} \Big| < \frac1{j^{3+\delta}} \Big\},
\end{equation}
which has been proved \cite{borosh} to have Hausdorff dimension
$\frac{3}{3+\delta}$ and Lebesgue measure zero. It follows
\cite{rudin:cx} that $\mc E_\delta$ has full Lebesgue measure.  We use
a theorem of Schmidt \cite{schmidt:64} to prove that $\mc E_\delta$
has measure zero for $\delta\in\big({-}\frac12,0\big]$. For the
$\delta\le-\frac12$ result, we use Weyl's equidistribution theorem
\cite{stein} that if $x$ is irrational then $\{j^2x\,|\,j\in\mbb N\}$
is equidistributed on $[0,1]$ modulo 1. Our proof that rational depths
belong to $\mc E_\delta$ for $\delta>\frac12$ makes use of Lambert's
continued fraction \cite{lorentzen:book} for $\tanh\mu_0$ to establish
that the irrationality exponent of $\tanh\mu_0$ is 2 via the method of
\cite{hancl:15}.  It may be possible to replace $\delta>\frac12$ by
$\delta>0$ for rational depths, but we do not know how to take
advantage of $p/j^2$ appearing with $j$ squared in equation
\eqref{eq:Fdelta:def:0}. One can estimate values of $\delta$ for which
$a$ is not too small numerically. For example, $\mu_0=1/16$ appears to
belong to $\mc E_\delta$ with $\delta=0.07$ and $a=0.0155$, based on
checking $\lambda_p$ for $2\le p\le 6.24\times 10^{22}$.  This is
shown in the electronic supplementary material, where we also argue
that these lower bounds on small divisors are important for the
convergence of Pad\'e approximants of the Stokes expansion.  There is
an asymmetry in which $a$ approaches 0 as $\mu_0$ approaches a fixed
resonant depth through the rationals, but $a$ is positive for a fixed
rational depth $\mu_0$, even though there are sequences of resonant
depths approaching $\mu_0$.

\subsection{Recursive algorithm}\label{sec:recursive:alg}

In this section we assume $\mu_0$ is not a resonant depth.
For every integer $\nu\ge1$, we define a set of lattice points
\begin{equation} \label{fdLN}
  L_\nu = \big\{\, (p,n)\;\,\big\vert\;\, n\ge0\,,\,p\ge1\,,\,
  p+2n\le\nu \,\big\},
\end{equation}
as well as a corresponding set of ordered triples
\begin{equation}	\label{fdGammaN}
  \Gamma_\nu = \big\{\,(\alpha_{p,n},\beta_{p,n},\gamma_{p,n})
  \;\,\big\vert\;\, (p,n)\in L_\nu\,\big\},
\end{equation} 
where $\alpha_{p,n}$, $\beta_{p,n}$ and $\gamma_{p,n}$ are assumed to
be of the form \eqref{fdFourierSeriesAlpha},
\eqref{fdFourierSeriesBeta} and \eqref{fdFourierSeriesGamma},
respectively.  We now state an induction hypothesis,
$\mathscr{P}_\nu$, for $\nu\ge1$.  The proof establishes the validity
of the algorithm, and thus the existence and uniqueness of a solution
of \eqref{eq:I}, \eqref{eq:II:star}, \eqref{eq:III:star},
\eqref{eq:IV}, \eqref{eq:i} and \eqref{eq:ii}. \\[-6pt]

\noindent

$\mathscr{P}_\nu$ (induction hypothesis): with
$N=\lfloor(\nu-1)/2\rfloor$, there exist unique Stokes expansion
coefficients $\Gamma_\nu$, $\mc M_N$, and $\Sigma_N$ satisfying
\eqref{eq:II:star}, \eqref{eq:III:star} and \eqref{eq:IV} for
$(p,n)\in L_\nu$; satisfying \eqref{eq:I} and \eqref{eq:i} for $0\le
n\le N$; satisfying \eqref{eq:ii} for $1\le n\le N$; and satisfying
\begin{align}	\label{fdOrthogonalityForSigmaN}
	\int_{0}^{2\pi} \cos(t) S_{1,n}(t) \, dt = 0, \qquad (0\le n\le N).
\end{align}

$S_{p,n}$ is the forcing term in \eqref{eq:II:star}, so the
orthogonality condition \eqref{fdOrthogonalityForSigmaN} ensures
solvability of \eqref{eq:II:star} at $p=1$ by eliminating secular
terms in the solution that destroy time-periodicity.  It is also the
constraint needed to uniquely determine the $\sigma_n$ values. We will
prove $\mathscr{P}_\nu$ inductively and exhibit the algorithm through
the proof. The computational scheme is illustrated in figure
\ref{FDAlgorithmImage}.

\begin{theorem}		\label{fdTheoremPNHoldsEverywhere}
	$\mathscr{P}_\nu$ holds for all $\nu \in \mbb{N}$.
\end{theorem}

\begin{proof}
  We have already established in equation \eqref{eq:pn:10} that $\mathscr{P}_1$
  holds. Thus, it follows from Lemma~\ref{fdPNInductionStep} below
  that $\mathscr{P}_\nu$ holds for all $\nu\ge1$. \hspace*{\fill} \raisebox{-3pt}{\qed}
  \\[-6pt]
\end{proof}

\begin{figure}[t]
  \begin{center}
    \includegraphics[scale=0.61,trim=0 23 0 0,clip]{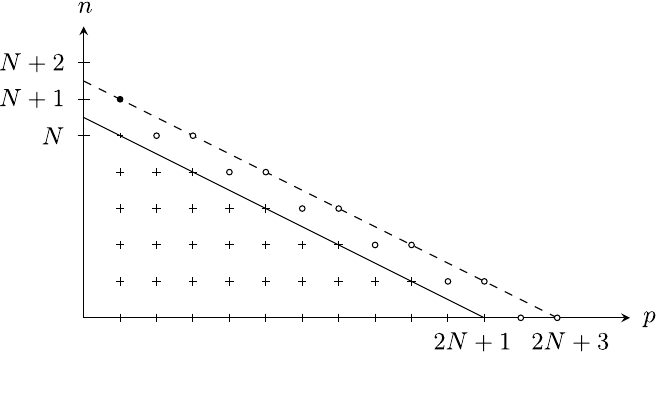}
  \end{center}
  \caption{Points marked with $+$ represent $\Gamma_{2N+1}$ while
    points marked with $\circ$ and $\bullet$ extend $\Gamma_{2N+1}$ to
    $\Gamma_{2N+3}$. On each iteration, we use \eqref{eq:II:star},
    \eqref{eq:III:star} and \eqref{eq:IV} to compute $\alpha_{p,n}$,
    $\beta_{p,n}$ and $\gamma_{p,n}$ at the points marked by $\circ$
    in two batches, first those in $L_{2N+2}^\circ$, followed by those
    in $L_{2N+3}^\circ$. The latter points lie on the dashed line.
    Next we compute $\mu_{N+1}$ using \eqref{eq:I} and
    \eqref{eq:ii}. Finally, we reach $\bullet$, where we compute
    $\sigma_{N+1}$ using the orthogonality condition
    (\ref{fdOrthogonalityForSigmaN}), followed by $\alpha_{1,N+1}$,
    $\beta_{1,N+1}$, and $\gamma_{1,N+1}$ via \eqref{eq:II:star} and
    \eqref{eq:i}, \eqref{eq:III:star}, and \eqref{eq:IV},
    respectively.}\label{FDAlgorithmImage}
\end{figure}

\begin{lemma}\label{fdPNInductionStep}
  If $N\ge0$ and $\mathscr{P}_{2N+1}$ holds, then $\mathscr{P}_{2N+2}$
  and $\mathscr{P}_{2N+3}$ also hold, and the extensions (in the
    set-theoretic sense) of $\Gamma_{2N+1}$, $\mc M_N$, and $\Sigma_N$
  to $\Gamma_{2N+3}$, $\mc M_{N+1}$, and $\Sigma_{N+1}$, respectively,
  are unique.
\end{lemma}

\begin{proof}
We assume $\mathscr{P}_{2N+1}$ holds and $\Gamma_{2N+1}$, $\mc M_N$,
and $\Sigma_N$ are known.  We begin by extending $\Gamma_{2N+1}$ to
$\Gamma_{2N+2}$, i.e., to the open circles of figure
\ref{FDAlgorithmImage} that do not lie on the dashed line.  We denote
these lattice points by $L_{2N+2}^\circ$, where
\begin{equation}
  L_\nu^\circ = \big\{\,(p,n) \;\,\big\vert\;\, n\ge0\,,\,p\ge2\,,\,
  p+2n=\nu\,\big\}.
\end{equation}
For each $(p,n)\in L_{2N+2}^\circ$, $\Omega_{p,n}$ is a subset of
$L_{2N+1}$, so by the induction hypothesis, each $\alpha_{q,m}$,
$\beta_{q,m}$ and $\gamma_{q,m}$ with $(q,m)\in \Omega_{p,n}$ is known
and has the form \eqref{fdFourierSeriesAlpha},
\eqref{fdFourierSeriesBeta} or \eqref{fdFourierSeriesGamma} with $p$
replaced by $q$ and $n$ replaced by $m$. Since $\mu_n(t)$ has the form
\eqref{fdFourierSeriesMu} for $1\le n\le N$, it follows from the
recurrence \eqref{eq:Bqn:recur} that $B_{q,m}(t)$ in equation
\eqref{eq:cqh:sqh} has the form
\begin{equation}\label{eq:Bqm:form}
  B_{q,m}(t) = \sideset{}{'}\sum_{j=0}^{2m} B_{q,m,j} H_j  \lp
  e^{ijt} + e^{-ijt} \rp, \qquad
  \big(0\le m\le N\,,\, q\in\mbb Z\big),
\end{equation}
where the $B_{q,m,j}$ are real coefficients and, as before, a prime on
a sum indicates that indices of opposite parity to the upper limit
should be excluded. The case $q=0$ is trivial since $\exp(0h)=1$ in
equation \eqref{eq:cqh:sqh}. As a result, $B_{0,m}(t)$ still has the form
\eqref{eq:Bqm:form}, but all the coefficients $B_{0,m,j}$ are zero
except for $B_{0,0,0}=1$. From equation \eqref{eq:cqn:sqn:bell}, we then also
have
\begin{equation}\label{eq:cqm:sqm:form}
  c_{q,m}(t) = \sideset{}{'}\sum_{j=0}^{2m} c_{q,m,j} H_j \lp e^{ijt} + e^{-ijt} \rp, \quad
  s_{q,m}(t) = \sideset{}{'}\sum_{j=0}^{2m} s_{q,m,j} H_j \lp e^{ijt} + e^{-ijt} \rp
\end{equation}
for $0\le m\le N$ and $q\in\mbb Z$, where $c_{q,m,j}$ and
$s_{q,m,j}$ are real. 
The coefficients $c_{q,m,j}$ and
$s_{q,m,j}$ that arise in the forcing terms $T^r_{p,n}$ with
$(p,n)\in L_{2N+3}\setminus L_{2N+1}$ satisfy $|q|\le
2N+3$.

From the definitions \eqref{eq:T2:def}, \eqref{kinematicGreekT3} and
\eqref{fdBernoulliGreekT4} of the forces, still assuming $(p,n)\in
L_{2N+2}^\circ$, we see that $T^2_{p,n}$ and $T^3_{p,n}$ are odd
trigonometric polynomials of the form \eqref{fdFourierSeriesBeta}
while $T^4_{p,n}$ and $S_{p,n}$ are even trigonometric polynomials of
the form \eqref{fdFourierSeriesAlpha}. For example, one of the terms
that appears in $T^3_{p,n}$ is a multiple of
$\alpha_{p-j,k}\alpha_{j,l}\dot\mu_m c_{p-2j,n-k-l-m}$, which is a
trigonometric polynomial of degree
\begin{equation}
  \big[(p-j)+2k\big]+\big(j+2l\big)+\big(2m\big)+\big(2n-2k-2l-2m\big)=p+2n.
\end{equation}
It is an odd function as $\dot\mu_m(t)$ is odd while the other factors
are even.  And it includes only terms $e^{ijt}$ with $j$ of the same
parity as $p+2n$. Indeed, each factor (indexed by $q$) has the form
$e^{i\nu_q t}P_q(e^{-2it})$ where $P_q$ is a polynomial of degree
$\nu_q$, so the product is also of this form.  Since $\mu_0$ is not a
resonant depth, we may solve \eqref{eq:II:star} uniquely to obtain
$\alpha_{p,n}$ of the form \eqref{fdFourierSeriesAlpha}, then
\eqref{eq:III:star} uniquely to obtain $\beta_{p,n}$ of the form
\eqref{fdFourierSeriesBeta}, and finally \eqref{eq:IV} uniquely to
obtain $\gamma_{p,n}$ of the form \eqref{fdFourierSeriesGamma}. This
establishes $\mathscr P_{2N+2}$.  Given $\mathscr P_{2N+2}$, identical
arguments show that \eqref{eq:II:star}, \eqref{eq:III:star} and
\eqref{eq:IV} uniquely determine $\alpha_{p,n}(t)$, $\beta_{p,n}(t)$
and $\gamma_{p,n}(t)$ of the form \eqref{fdFourierSeriesAlpha},
\eqref{fdFourierSeriesBeta} and \eqref{fdFourierSeriesGamma},
respectively, at the lattice points in $L_{2N+3}^\circ$, which are the
$\circ$ markers on the dashed line in figure~\ref{FDAlgorithmImage}.

The next step is to compute $T^1_{0,N+1}$, which, by
equation \eqref{fdKinematicGreekP0} and the above reasoning, is an odd
trigonometric polynomial of degree $2(N+1)$ that omits terms $e^{ijt}$
with $j$ odd. We then solve \eqref{eq:I} and \eqref{eq:ii} uniquely
for $\mu_{N+1}$ of the form \eqref{fdFourierSeriesMu}. We also learn
that equations \eqref{eq:Bqm:form} and \eqref{eq:cqm:sqm:form} hold for $m=N+1$
in addition to the cases $m\le N$ established above.  This fact is
needed for the last lattice point $(p,n)=(1,N+1)$ to conclude that
$T^3_{1,N+1}(t)$ and $T^4_{1,N+1}(t)$ are, respectively, odd and even
trigonometric polynomials of degree $2N+3$ that omit terms $e^{ijt}$
with $j$ even. The analogous conclusion for $T^2_{1,N+1}(t)$ follows
directly from $\mathscr{P}_{2N+2}$ since $\mu_m$, $c_{q,m}$ and
$s_{q,m}$ do not appear in \eqref{CRGreekT2}. At this point, all terms
in $S_{1,N+1}$ in \eqref{eq:II:star} are known except $\sigma_{N+1}$,
which is determined using the orthogonality condition
\begin{equation}\label{eq:cos:S1}
  \int_{0}^{2\pi} \cos(t) S_{1,N+1}(t) dt = 0.
\end{equation}
The term in $S_{1,N+1}$ that contains $\sigma_{N+1}$ is
$(1/c_{1,0})\alpha_{1,0} \sigma_{N+1}
s_{1,0}=\sigma_{N+1}\cos(t)$, so \eqref{eq:cos:S1} is a linear
equation in $\sigma_{N+1}$ whose coefficient is not zero. This
eliminates secular growth in the solution of \eqref{eq:II:star} for the lattice
point $(1,N+1)$ and uniquely determines $\alpha_{1,N+1}$ of the
form \eqref{fdFourierSeriesAlpha}, up to an arbitrary real
multiple of $\cos(t)$. To determine this unknown coefficient, we
compute $\alpha_{1,N+1}(0)$ from \eqref{eq:i}, where all other quantities
are known.  Finally, we use \eqref{eq:III:star} and \eqref{eq:IV} to compute
$\beta_{1,N+1}$ and $\gamma_{1,N+1}$ of the forms
\eqref{fdFourierSeriesBeta} and \eqref{fdFourierSeriesGamma},
respectively.

We have shown that the necessary extensions exist, are unique, and
preserve the trigonometric polynomial structure of the induction
hypothesis, thus proving the lemma. \hspace*{\fill} \raisebox{-3pt}{\qed}
\end{proof}

\section{Numerical results}	\label{fdSecNumericalResults}

We computed the expansion coefficients $\alpha_{p,n,j}$,
$\beta_{p,n,j}$, $\gamma_{p,n,j}$, $\mu_{n,j}$ and $\sigma_n$ for
dimensionless fluid depths
\begin{equation}\label{eq:mu0:list}
  \mu_0 \; \in \; \left\{ \, \frac1{16} \;,\; \frac14 \;,\; \frac35 \;,\;
  1 \;,\; 4 \;,\; 10 \;,\; 16 \;,\; \infty \, \right\}
\end{equation}
up to order $\nu=p+2n=109$, and then again to $\nu=149$ for
$\mu_0\in\{3/5,1,\infty\}$ to further explore the convergence of the
Pad\'e approximants studied in \S\ref{sec:pade} below.  In
infinite depth, we implemented a variant of the Schwartz \& Whitney
algorithm \cite{schwartz1981semi} that will be explained in detail
elsewhere \cite{abassi:semi2} as a special case of standing
gravity-capillary waves in infinite depth. Our code employs the MPFR
multiple precision library \cite{mpfr:toms} with a fixed mantissa
size. We implemented it on a supercomputer using a hybrid MPI/OpenMP
parallel framework \cite{chopp:book}. We ran each calculation at least
twice, with different precisions, to observe how floating-point errors
accumulate, estimate these errors, and repeat with more precision if
necessary. The precisions used were 64, 90, 144, 192 and 256
  digits (212, 300, 480, 638 and 850 bits). Computational aspects and
implementation details of the algorithm are given in the electronic
supplementary material along with a discussion of the generation,
propagation and estimation of floating-point errors.

\subsection{Growth of the coefficients in the asymptotic expansion}
\label{sec:growth}

It is useful to consolidate the $\epsilon$-expansions of $a_p(t)$ and
$h(t)$ in the formula \eqref{ansatzZ} for $Z(w,t)$. This gives a
single $\epsilon$-expansion of the non-dimensionalized wave profile,
which we denote by
\begin{equation}\label{eq:tilde:eta:series}
  \tilde\eta(\alpha,t) =
  \frac{2\pi}{L}\eta\Big(\alpha,\frac{T}{2\pi}t\Big) =
  \im\big\{ Z(\alpha,t) \big\} =
  \sum_{\nu=1}^\infty \tilde\eta^\e{\nu}(\alpha,t) \epsilon^\nu.
\end{equation}
As in equation \eqref{eq:ZFW:def}, the dimensionless variables $(\alpha,t)$
range over the torus $\mbb T^2$ rather than over a domain that depends
on $\epsilon$.  We denote the Fourier representation of
$\tilde\eta^\e\nu(\alpha,t)$ by
\begin{equation}
  \tilde\eta^\e\nu(\alpha,t) = \sum_{p,j\in E_\nu} \tilde \alpha^\e{\nu}_{p,j}
  e^{ip\alpha}e^{ijt},
\end{equation}
where $E_\nu$ was defined in \eqref{eq:Enu:Hj:def}.  We also introduce
the area-weighted $L^2$-norm
\begin{equation}\label{eq:A:nu:def}
  A_\nu 
  = \left(\frac1{(2\pi)^2}\int_0^{2\pi}\int_0^{2\pi}
    \big[\tilde\eta^\e\nu(\alpha,t)\big]^2\,d\alpha\,dt\right)^{1/2}
  = \sqrt{\sum_{j,p\in E_\nu} \Big| \tilde \alpha^\e\nu_{p,j} \Big|^2}
\end{equation}
to measure the growth of successive terms in equation
\eqref{eq:tilde:eta:series}.  Using equations \eqref{Stokes},
\eqref{fdFourierSeriesMu} and \eqref{fdFourierSeriesAlpha}, it follows
from
\begin{equation}\label{eq:im:Z:expand}
  \im\big\{ Z(\alpha,t) \big\} =
    h(t) - \mu_0 + \sum_{p=1}^\infty a_p(t)\frac{\sinh\!\big(ph(t)\big)}{\cosh(p\mu_0)}
    \cos(p\alpha)
\end{equation}
that $\tilde\alpha^\e{2n}_{0,j} = \mu_{n,|j|}$ for $n\ge1$, $j\in
E_{2n}$; that $\tilde\alpha^\e\nu_{p,j}=\tilde\alpha^\e\nu_{|p|,|j|}$
for $p,j\in E_\nu$; and that
\begin{equation}
  \sum_{j\in E_{p+2n}}\tilde\alpha^\e{p+2n}_{p,j}e^{ijt} =
  \sum_{m=0}^n \frac{\alpha_{p,m}(t)s_{p,n-m}(t)}{2\cosh(p\mu_0)}, \qquad
  p\ge1\,,\,n\ge0.
\end{equation}
We compute the right-hand side in real space on a uniform grid in the
$t$ variable with enough gridpoints to avoid aliasing errors. The
coefficients $\tilde\alpha^\e{p+2n}_{p,j}$ are then easily obtained
using the Fast Fourier Transform (FFT) \cite{recipes,fft:2345}. In the
infinite-depth case, the formula is simpler:
$\tilde\alpha^\e{|p|+2n}_{p,j} = [1/(2H_{|p|})]\alpha_{|p|,n,j}$,
where $\alpha_{p,n,j}$ is still related to $a_p(t)$ via equations
\eqref{stokesA} and \eqref{fdFourierSeriesAlpha} but the ansatz
\eqref{ansatz} is replaced by equation (2.12a) from \cite{schwartz1981semi},
i.e., $Z(w,t)=w+i\sum_{p=0}^\infty a_p(t) e^{-ipw}$.

Figure~\ref{fig:growthFD}(a,b) shows the growth rate factors
$\sqrt{A_\nu/A_{\nu-2}}$ of the norms $A_\nu$ for different fluid
depths. Plotting $\sqrt{A_\nu/A_{\nu-2}}$ instead of $A_\nu/A_{\nu-1}$
decouples the even and odd orders, which eliminates oscillations that
obscure the plots.  In the cases we studied, the growth rates approach
limiting values separated by occasional `stairstep jumps' from one
plateau height to another over a narrow transition region. Notable
jumps occur for depths $\mu_0\in\{1/4,3/5,1\}$ near
orders $\nu\in\{70,102,66\}$, respectively.  We show a
connection between these jumps in growth rate and new small divisors
entering the recurrence in the electronic supplementary material.
Figure~\ref{fig:growthFD}(c) shows the norms $A_\nu\epsilon^\nu$ of
successive terms of the series \eqref{eq:tilde:eta:series} for
$\mu_0=1$ and $\mu_0=3/5$ for various choices of the amplitude
$\epsilon$. Jumps in the growth rate in
figure~\ref{fig:growthFD}(b) lead to kinks in the plots of
figure~\ref{fig:growthFD}(c). Successive terms of the
  series decay geometrically until the inverse growth rate factor
\begin{equation}\label{eq:rho:nu:def}
  \rho_\nu = \sqrt{A_{\nu-2}/A_\nu}, \qquad\quad (\nu\ge3)
\end{equation}
drops below $\epsilon$, after which they grow geometrically.
This is illustrated with $\epsilon=0.03$ for $\mu_0=3/5$ in
figure~\ref{fig:growthFD}(c). The other two curves in
  figure~\ref{fig:growthFD}(c) show that when $\epsilon=\rho_\nu$ in a
plateau region where $\rho_\nu$ is nearly constant, the norms
  $A_\nu\epsilon^\nu$ also remain nearly constant.  As a
  function of $\epsilon$, if the series is evaluated through order
  $\nu_\text{max}$ by direct summation, it begins to grow rapidly once
  $\epsilon$ exceeds $\rho_{\nu_\text{max}}$.  For fixed $\epsilon$,
with direct summation,
the series should be truncated at or before the last $\nu$ for which
$\rho_\nu>\epsilon$. However, we find in \S\ref{sec:pade}
  below that Pad\'e approximants of the series continue to improve in
  accuracy as the order is
  increased, without requiring $\rho_\nu>\epsilon$.

\begin{figure}[t]
  \begin{center}
    \includegraphics[width=\linewidth]{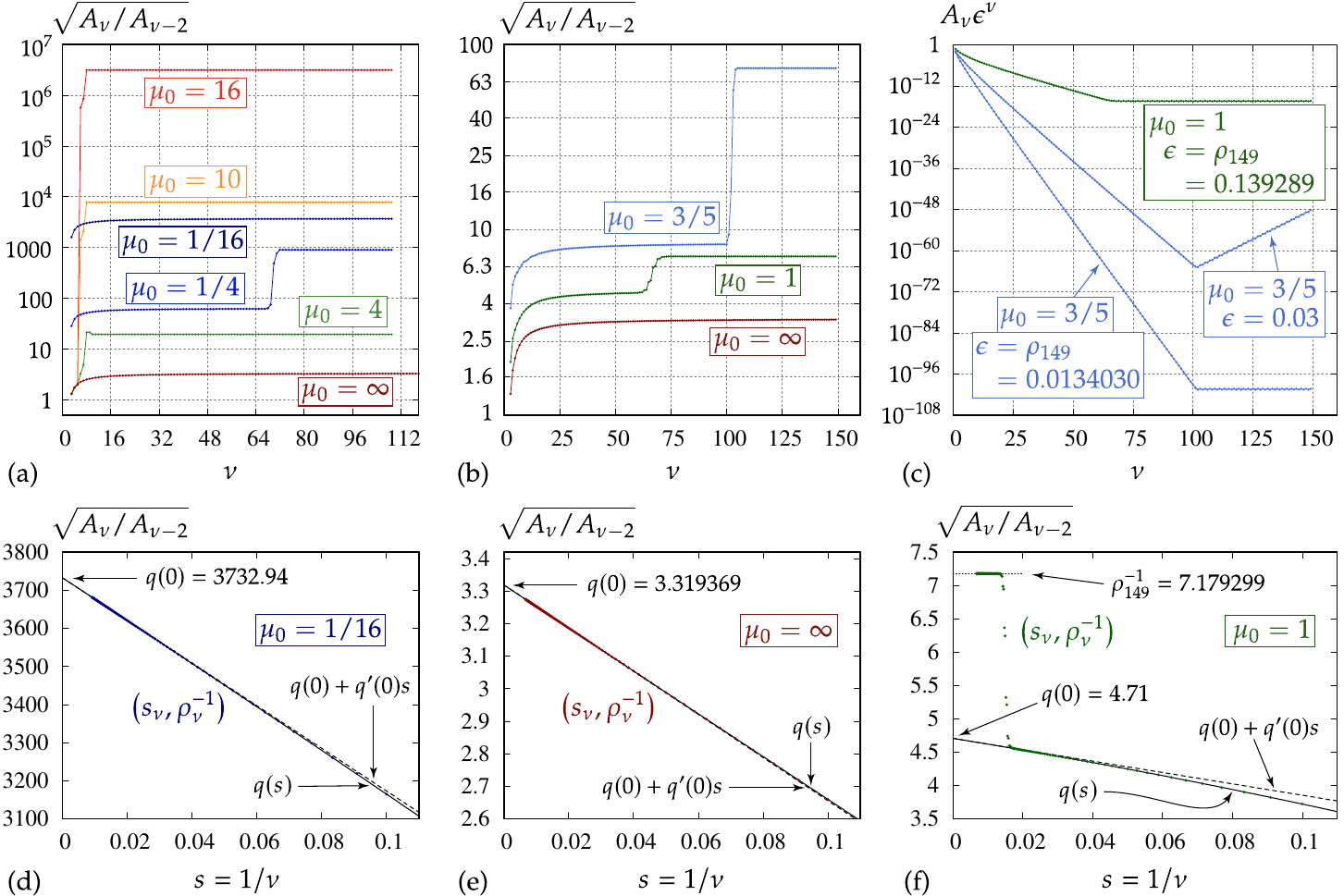}
  \end{center}
  \caption{\label{fig:growthFD} Growth rate factors
    $\rho_\nu^{-1}=\sqrt{A_\nu/A_{\nu-2}}$ and norms
    $A_\nu\epsilon^\nu$ of successive terms in the asymptotic
    expansion. (a,b) Jumps in the growth rate occur for
    $\mu_0\in\{1/4,3/5,1,4,10,16\}$ when small divisors enter the
    recurrence and lead to new growth patterns among the coefficients
    $\tilde\alpha^\e\nu_{p,j}$ in equation \eqref{eq:A:nu:def}. (c) Optimal
    truncations of the asymptotic series occur at kinks in the curves
    where successive terms stop decreasing due the jump in
    $\rho_\nu^{-1}$ from (b). (d,e,f) Domb-Sykes plots to estimate
    the location of singularities in the family of solutions.
    }
\end{figure}

Figure~\ref{fig:growthFD}(d,e) shows Domb-Sykes plots
  \cite{domb:sykes:57} of $\rho_\nu^{-1}$ versus $s_\nu=1/\nu$ for
the two cases where a jump is not observed, $\mu_0=1/16$ and
$\mu_0=\infty$. The solid curves show low-degree polynomials
$q(s)$ that were fit to the data points
$\big(s_\nu,\rho_\nu^{-1}\big)$ as described in the electronic
supplementary material. The dashed lines show $q(0)+q'(0)s$,
which is the extrapolated estimate of the leading-order asymptotic
behavior as $\nu\to\infty$ and $s_\nu\to0^+$. This predicts, by
  the ratio test, a radius of convergence of $\rho=0.000267885$ for
$\mu_0=1/16$ and $\rho=0.301262103$ for $\mu_0=\infty$. (We report the
  number of digits that appear justified from the polynomial fit.)
If the series corresponds to a family of solutions that depend
  analytically on the amplitude $\epsilon$, we expect a singularity in that
  family at some $\epsilon_*\in\mbb C$ with $|\epsilon_*|\approx\rho$.  If
  $\rho_\nu^{-1}$ jumps to a new plateau height, it indicates that a
  new singularity has been detected that is even closer to the
  origin. Figure~\ref{fig:growthFD}(f) demonstrates this for the
  depth $\mu_0=1$.  Extrapolation to $s=0$ through order
  $\nu_\text{max}=55$ suggests there is a singularity $\epsilon_*$ with
  $|\epsilon_*|\approx\rho$ and $\rho^{-1}=q(0)=4.71$. But there is a
  transition region where $\rho_\nu^{-1}$ stops following
  $q(s_\nu)$ extrapolated from $10\le\nu\le55$ and instead jumps rapidly
  from $\rho_{59}^{-1}=4.562$ to $\rho_{73}^{-1}=7.172$. It then
  stabilizes at $\rho_{\nu}^{-1}\approx\rho_{149}^{-1}=7.179299$ for
  $85\le\nu\le149$, suggesting another singularity $\epsilon_*$ with
  $|\epsilon_*|\approx\rho_{149}$. We find that the plateau regions after a
  jump occurs are extremely flat.  For example, all 17 digits we
recorded for $\rho_\nu^{-1}$ remain unchanged for $10\le\nu\le109$ in
the case of $\mu_0=16$.  It is not helpful to fit the data after a
jump with anything but a constant function.

We expect that at every non-resonant finite depth, there will
eventually be infinitely many stairstep jumps in
  $\rho_\nu^{-1}$ that cause $\rho=\lim_{\nu\to\infty}\rho_\nu=0$.
This is consistent with previous studies
\cite{roberts:81,roberts:83,marchant:87} that concluded that
asymptotic expansions of standing waves and short-crested waves have a
zero radius of convergence for all depths.  However,
Theorem~\ref{thm:nonresonant} shows that for almost every fluid depth,
the small divisors that lead to the jumps in $\rho_\nu^{-1}$ arise
infrequently as $p$ increases, which limits how fast $\rho_\nu$
approaches zero. In the case $\mu_0=1/16$, which we did not optimize
in advance, we find that $\lambda_p\ge\lambda_2\approx0.0155$ for
$2\le p\le 24773$. So one will not encounter a divisor that is smaller
than the first one in practice. There are no bifurcations associated
with $|\lambda_{2,2}|$ being small when $\mu_0=1/16$ since
$(p,j)=(5,3)$ is the first harmonic resonance in finite depth;
see \S\ref{sec:res:depths} above and the tables of
  resonant depths in \cite{marchant:87}. Consistent with this, we find
  that the closest Pad\'e poles to the origin lie on the imaginary
axis up to the order we computed ($\nu=109$). These results on
imaginary Pad\'e poles and a discussion of the importance of
$\rho_\nu$ approaching zero slowly as $\nu\to\infty$ for the
convergence of the Pad\'e approximants at larger amplitudes are
included in the supplementary material and will be investigated
further in future work.

\subsection{Imperfect bifurcations computed using a shooting method}
\label{sec:bif}

To verify the correctness of the asymptotic expansions and benchmark
their accuracy, we compare them quantitatively to standing waves
computed via numerical continuation using the overdetermined shooting
method of Wilkening \& Yu \cite{water2}.  We focus on the cases
$\mu_0=3/5$ and $\mu_0=1$ as they both possess interesting bifurcation
structures associated with nearby harmonic resonances
\cite{mercer:94,smith:roberts:99}. The depth $\mu_0=1/16$ is studied
in the electronic supplementary material. Following
\cite{mercer:92,water2}, we exploit a symmetry to cut the simulation
time down to a quarter period. This rules out symmetry-breaking
  bifurcations and enforces the ansatz
  \eqref{ansatz} and \eqref{fdFourierSeriesMu}--\eqref{fdFourierSeriesGamma}.
In the current paper, $t=0$
corresponds to a maximum-amplitude `rest' state. With this
convention, the initial conditions of the shooting method are imposed
at $t_0=-\pi/2$ and the objective function of the shooting method
drives the velocity potential to zero at the final time, $t_N=0$. Here
we discretize time into $N\ge1$ segments $[t_{n-1},t_n]$ with
$-\pi/2=t_0<t_1<\cdots<t_N=0$ and use a uniform grid with $I_n$
timesteps and $M_n$ gridpoints on each segment,
\begin{equation}\label{eq:adapt:grids}
  \begin{aligned}
    &t_{ni} = t_{n-1}+i\Delta t_n, \\[-2pt]
    &x_{nm} = 2\pi m/M_n,
  \end{aligned}
  \qquad
  \Big(\Delta t_n = \frac{t_n-t_{n-1}}{I_n}\,,\,
      1\le n\le N\,,\, 0\le i\le I_n\,,\, 0\le m< M_n \, \Big).
\end{equation}
We use the eighth-order Dormand/Prince Runge-Kutta method
\cite{hairer:I} for double-precision calculations and a
fifteenth-order spectral deferred correction method \cite{dutt}
for quadruple-precision calculations. The shooting method employs a
graph-based formulation of the water wave equations expressed in terms
of wave height $\eta^\text{graph}(x,t)$ and velocity potential
$\varphi^\text{graph}(x,t)$. In the code, the time variable is $\breve
t = \big(\frac14+\frac{t}{2\pi}\big)T$, which evolves from $\breve
t_0=0$ to $\breve t_N=T/4$, but we use dimensionless time $t$
here for simplicity. To compute standing waves, we minimize the
objective function
\begin{equation}\label{eq:obj:fcn}
  f(\theta) = \frac12 r(\theta)^Tr(\theta), \qquad r_m(\theta) =
  \frac1{\sqrt{M_N}}\varphi^\text{graph}(x_{Nm},t_N), \qquad
  \left(\begin{aligned} 0\le m &< M_N \\[-4pt] t_N &= 0 \end{aligned} \right),
\end{equation}
where $r(\theta)$ is the vector in $\mbb R^{M_N}$ with components $r_m(\theta)$, and
$\theta$ contains the period and initial Fourier modes of the
solution up to a given order $d$,
\begin{equation}\label{eq:dof}
  \theta \; = \;
  \Big( T, \;\,
  \big\{\hat\eta^\text{graph}_{2l}(t_0)\big\}_{l=1}^{\lfloor d/2 \rfloor}, \;\,
  \big\{\hat\varphi^\text{graph}_{2l-1}(t_0)\big\}_{l=1}^{\lceil d/2 \rceil}\,\Big), \qquad
    (t_0=-\pi/2).
\end{equation}
Here $\hat\eta^\text{graph}_k(t) =
\frac1{M_n}\sum_{m=0}^{M_n-1}\eta^\text{graph}(x_{nm},t)e^{-ikx_{nm}}$
and $\hat\varphi^\text{graph}_k(t)$ are computed via the FFT from the
grid values of the wave profile and surface velocity potential
(assuming $t_{n-1}\le t\le t_n$).  The floor and ceiling functions
satisfy $\lfloor d/2\rfloor+\lceil d/2\rceil=d$ for all integers
$d\ge1$.

The components of $\theta$ in equation \eqref{eq:dof} are real and all other
Fourier modes of the initial condition are set to zero.  This imposes
the desired symmetry \cite{mercer:92,water2} that
$\eta^\text{graph}(x,t_0)$ is an even function of $x$ that remains
unchanged if $x$ is shifted by $\pi$ while
$\varphi^\text{graph}(x,t_0)$ is an even function that changes sign
when $x$ is shifted by $\pi$.  One of the degrees of freedom in equation
\eqref{eq:dof} is specified as an amplitude parameter in the numerical
continuation algorithm and is removed from the list of unknowns when
minimizing the objective function, so $\theta\in\mbb R^d$. We use
$\hat\varphi^\text{graph}_1(t_0)$ as a default; $T$ to line up the
periods of labeled solutions such as $ABC$ in
figure~\ref{fig:bif:evol:06}; and the most resonant component of
$\theta$ to navigate turning points in
$\hat\varphi^\text{graph}_1(t_0)$, e.g., on bifurcation branches.
Alternative amplitude parameters include crest acceleration
\cite{mercer:92,mercer:94,smith:roberts:99,water1} and energy
\cite{waterTS}.  Further details on the boundary integral method used
to evolve the water wave equations, our nonlinear least-squares
solver, and the variational equations used to compute $J=\nabla_\theta
r$ are given in \cite{water2}.

The shooting method results need to be converted to conformal
variables in order to compare them to the asymptotic expansions of
\S\ref{sec:recursive:alg}. Focusing on the initial conditions,
we use Newton's method to solve $F[\eta]=0$, where $\eta(\alpha)$ is
shorthand for $\eta(\alpha,t_0)$ and
\begin{equation}\label{eq:graph2conf}
  F[\eta](\alpha) = \eta(\alpha) - \eta^\text{graph}\big( \xi(\alpha),t_0 \big),
  \qquad \xi(\alpha) = \alpha + H^{h,\opn{coth}}[\eta](\alpha).
\end{equation}
Here $H^{h,\opn{coth}}$ is the variant of the Hilbert transform with
symbol $\hat H^{h,\opn{coth}}_k=-i\coth(kh)$, and
$h=\mu_0+\frac1{2\pi}\int_0^{2\pi}\eta(\alpha)\,d\alpha$ is the fluid
depth in conformal space, which is calculated from $\eta$ as a
preliminary step in the evaluation of $F[\eta]$. The shooting method
places the bottom boundary at $y=-\mu_0$ and ensures that
$\int_0^{2\pi}\eta^\text{graph}(x,t_0)\,dx=0$.  Equation
\eqref{eq:graph2conf} is imposed at the collocation points
$\alpha_j=2\pi j/M_1$, $0\le j<M_1$, with $M_1$ as in equation
\eqref{eq:adapt:grids}, and $h$ is computed via the trapezoidal rule
at these same points, which preserves the spectral accuracy of the
solution. We also compute $\varphi(\alpha,t_0) =
\varphi^\text{graph}\big(\xi(\alpha),t_0\big)$ to convert the surface
velocity potential to conformal variables.

\begin{figure}[t]
  \begin{center}
    \includegraphics[width=\linewidth,trim=9 0 8 0,clip]{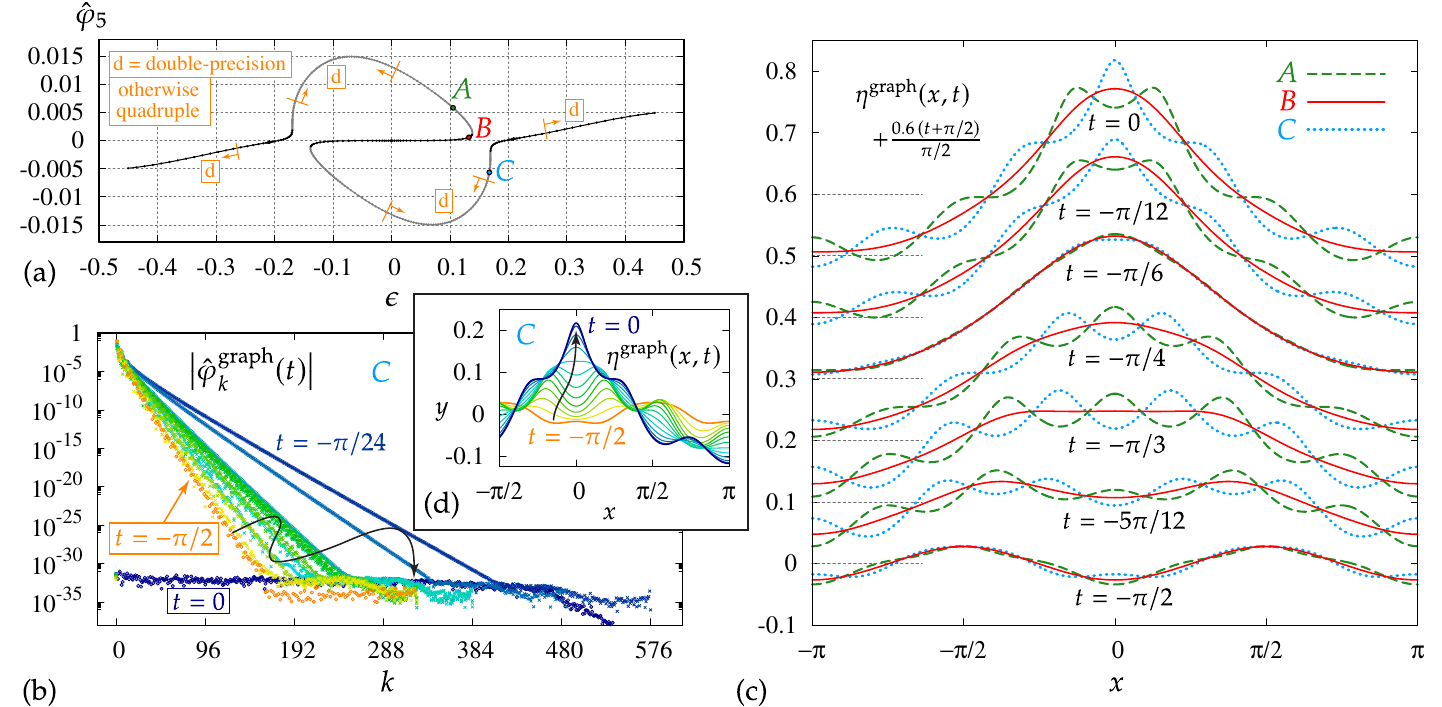}
  \end{center}
  \caption{\label{fig:bif:evol:06} Standing waves of depth
    $\mu_0=3/5$.  (a) The nearby $(5,3)$ harmonic resonance at
    $\mu_0=0.6232354$ leads to an imperfect bifurcation with a large
    gap separating solution $B$ from solution $C$. (b) Snapshots of
    the surface velocity potential (in Fourier space) and the wave
    profile (inset figure) for $t\in\mathscr{T}_{12}$ for solution
    $C$.  (c) The wave profile for solutions $A$, $B$ and $C$ at times
    $t\in\mathscr{T}_6$. The dashed gray horizontal lines near
    $x=-\pi$ are vertical offsets added for clarity. }
\end{figure}

Figure~\ref{fig:bif:evol:06} shows the results of the shooting method
for $\mu_0=3/5$, which is close to the resonant depth
$\mu_0=0.6232354$ where $\lambda_{53}=0$. Standing waves near this
resonance have been studied before
\cite{mercer:94,okamura:99,smith:roberts:99}, but we have new results
to report.  After converting the initial conditions from the shooting
method to conformal variables, we compute the Fourier expansions
\begin{equation}\label{eq:shooting:fourier}
  \eta(\alpha,t_0) = \sum_{p\in2\mbb Z} \hat\eta_pe^{ip\alpha}, \qquad
  \varphi(\alpha,t_0) = \sum_{p\in1+2\mbb Z} \hat\varphi_pe^{ip\alpha}, \qquad
  \big(t_0 = -\pi/2\big)
\end{equation}
numerically via the FFT, up to $|p|\le M_1/2$.
Figure~\ref{fig:bif:evol:06}(a) shows a bifurcation plot of
$\hat\varphi_5$ versus the amplitude,
$\epsilon=\frac12\big[\eta(0,0)-\eta(\pi,0)\big]$. A spatial shift by
$\pi$ leads to another standing wave with $\epsilon$ replaced by
$-\epsilon$. At $t=t_0$, this gives
$\varphi(\alpha,t_0;-\epsilon)=\varphi(\alpha-\pi,t_0;\epsilon)
=-\varphi(\alpha,t_0;\epsilon)$, which explains the odd symmetry of
the plot. There is an imperfect bifurcation near $\epsilon=0.15$ that
leads to a bubble structure in the bifurcation plot. Okamura observed
a similar structure within weakly nonlinear theory near this same
resonant depth \cite{okamura:99}.  Solutions $A$, $B$ and $C$
demonstrate typical behavior \cite{smith:roberts:99} of standing waves
near a resonant depth. Figure~\ref{fig:bif:evol:06}(c) shows
snapshots of these three solutions plotted on top of each other
at the dimensionless times
\begin{equation}\label{eq:scrT:def}
  \mathscr{T}_n = \Big\{\Big(\frac{j-n}{n}\Big)\Big(\frac\pi2\Big)
  \;\;\Big\vert\;\; 0\le j\le n\Big\},
\end{equation}
where $n=6$ in this plot. They were selected to have identical
periods, $T=8.45592$. The non-uniqueness is due to three possible
amplitudes of a secondary standing wave that evolves on top of the
primary wave and has features similar to the nearby harmonic resonance
($p=5$ spatial cycles and $j=3$ temporal cycles).  For solution $C$,
the secondary wave is in phase with the primary wave, which sharpens
the crest at $t=0$.  For solution $A$, it is out of phase, causing a
dimple to form at the wave crest at $t=0$. Since the secondary wave is
not active for solution $B$, we define the primary wave to be solution
$B$. Solutions $A$ and $C$ appear to oscillate around solution $B$,
though each is its own standing-wave solution of the fully nonlinear
water wave equations.

Figure~\ref{fig:bif:evol:06}(b) shows the time-evolution of the
Fourier modes of the surface velocity potential of solution $C$ in the
graph-based formulation of the shooting method. The modes decay exponentially
with respect to the wave number $k$, but the decay rate
fluctuates in time following the wavy black arrow in
figure~\ref{fig:bif:evol:06}(b). The
modes are also color coded, evolving from orange to yellow to green to
blue to navy, matching the time evolution of
  figure~\ref{fig:bif:evol:06}(d).  At the final time, $t=0$, the
velocity potential is driven nearly to zero by minimizing $f(\theta)$
in equation \eqref{eq:obj:fcn} to $4.8\times 10^{-62}$ so that all the Fourier
modes $\hat\varphi^\text{graph}_k(t_N)$ are below $10^{-31}$.  Except
in the regions indicated in figure~\ref{fig:bif:evol:06}(a),
all solutions were computed in quadruple-precision with $f(\theta)$
minimized below $10^{-60}$.

Figure~\ref{fig:bif:evol:1} shows the shooting method results for the
$\mu_0=1$ case. There are three nearby resonant depths
$\mu_0\in\{1.0397,\,0.9730,\,0.9962\}$ that lead to a cluster of small
divisors $\lambda_{p,j}$ with $(p,j)\in\{(7,3),\,(12,4),\,(19,5)\}$;
see figure~\ref{fig:smallDiv} of the electronic supplementary
material.  Figure~\ref{fig:bif:evol:1}(a,b,c) shows how the
nearly resonant Fourier modes of the initial condition
($\hat\varphi_7$, $\hat\eta_{12}$ and $\hat\varphi_{19}$) depend on
the amplitude $\epsilon$. These plots show different projections of
the same set of standing wave solutions and reveal a rich bifurcation
structure that has not been reported on before. As with the
$\mu_0=3/5$ case, when three branches meet at an imperfect pitchfork
bifurcation, solutions on the two side branches exhibit
higher-frequency, secondary standing waves oscillating with one of two
temporal phases on top of the primary wave. Solutions on the center
branch remain calm, without exciting this secondary wave. Similar
solutions with secondary standing waves have been reported previously
in
\cite{mercer:94,smith:roberts:99,okamura:99,water2,shelton:stand,rycroft:13}.
These secondary waves can deviate visibly from their form in the
linear water wave regime. This is demonstrated in
figures~\ref{fig:evol1defBig}--\ref{fig:bif027fit} in the electronic
supplementary material, which show solutions $DEF$, $HIJ$ and $KLM$ in
figure~\ref{fig:bif:evol:1} as well as the secondary wave associated
with another bifurcation at $\epsilon=0.27380806$.
In figure~\ref{fig:bif:evol:1}(a,b,c),
solution $G$ is the highest wave (with the
  largest crest-to-trough height) for
$\mu_0=1$, which will be discussed further in \S\ref{sec:pade}
below. Solution $O$ is the zero-amplitude flat rest state.

\begin{figure}[t]
  \begin{center}
    \includegraphics[width=\linewidth,trim=9 0 0 0,clip]{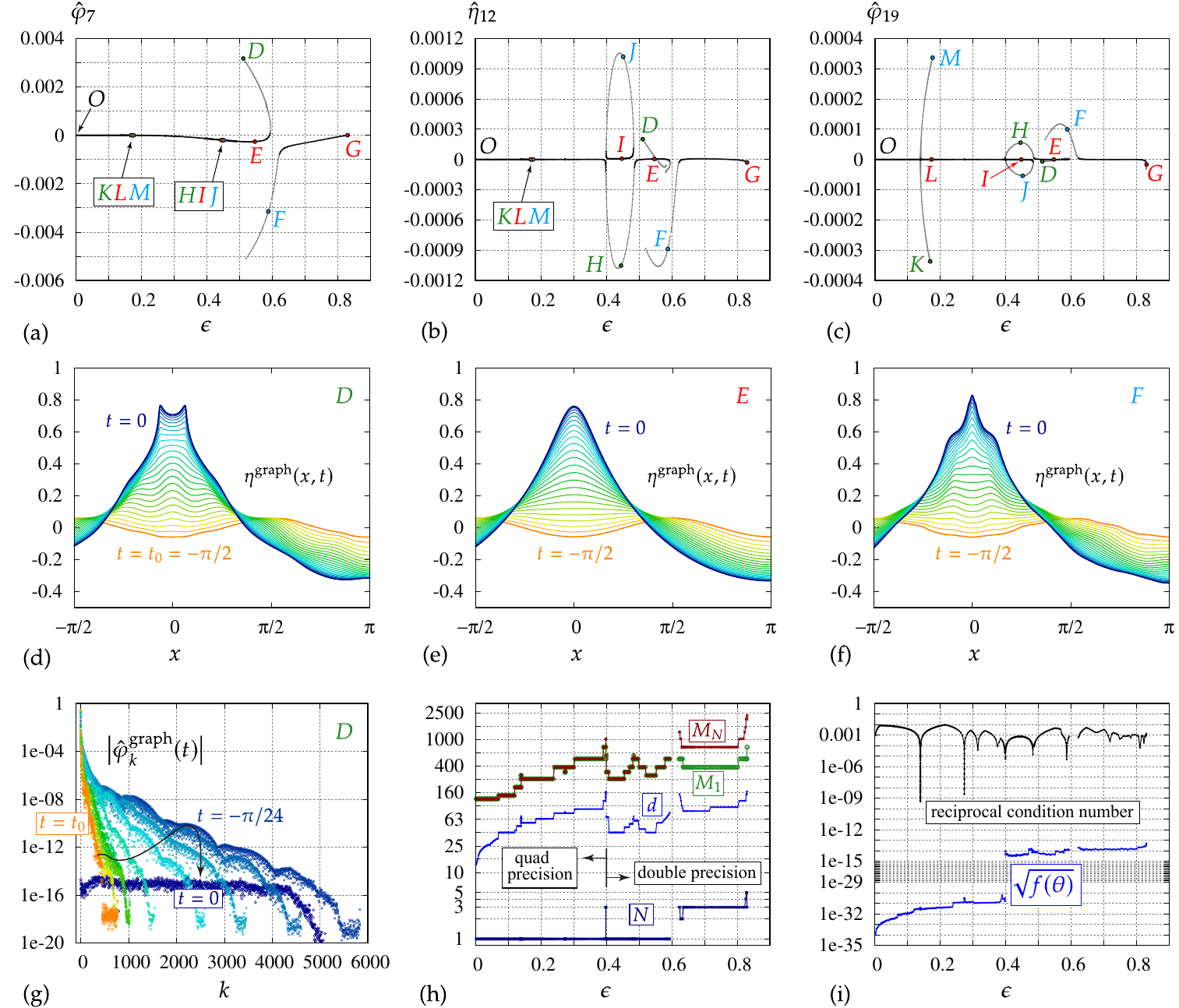}
  \end{center}
  \caption{\label{fig:bif:evol:1} Standing waves of depth $\mu_0=1$.
    (a,b,c) Bifurcation plots of $\hat\varphi_7$, $\hat\eta_{12}$ and
    $\hat\varphi_{19}$ versus $\epsilon$ reveal numerous imperfect
    bifurcations.  (d,e,f) Wave profile evolution for solutions $DEF$
    at dimensionless times $t\in\mathscr{T}_{24}$.  (g) Fourier mode
    evolution of $\hat\varphi^\text{graph}_k(t)$ for solution $D$ at
    times $t\in\mathscr{T}_{12}$.  (h) Parameters $d$, $N$, $M_1$ and
    $M_N$ of the shooting method for the solutions corresponding to
    the black markers in (a,b,c).  (i) Condition number and
    minimized value of $\sqrt{f(\theta)}$. }
\end{figure}

The $(7,3)$ resonance leads to the imperfect bifurcation separating
solution $F$ from solution $E$ in figure~\ref{fig:bif:evol:1}(a). It
became too expensive to maintain double-precision accuracy beyond
solution $D$, but if $\mu_0$ is increased to 1.03, this branch can be
continued further and meets up with the odd reflection (under
  $\epsilon\to-\epsilon$) of the branch passing through solution $F$;
see \cite{water2}.  Solutions $E$ and $F$ were chosen to match the
period, $T=7.26730$, of solution
$D$. Figure~\ref{fig:bif:evol:1}(d,e,f) shows solutions $DEF$
at dimensionless times $t\in\mathscr{T}_{12}$ from equation
\eqref{eq:scrT:def}. At $t=0$, the secondary standing wave causes a
dimple to form at the crest of solution $D$ and sharpens the crest of
solution $F$. The $(12,4)$ resonance leads to two imperfect
bifurcations on either side of solution $I$ in
figure~\ref{fig:bif:evol:1}(b) while the $(19,5)$ resonance
leads to the side branches passing through solutions $K$ and $M$ in
figure~\ref{fig:bif:evol:1}(c). The value of $\rho_\nu$ for
$81\le\nu\le149$ in figure~\ref{fig:growthFD}(b) is 0.139289, and
coincides with the amplitude $\epsilon$ where the imperfect
bifurcation to the $K$ and $M$ side branches occurs.  These side
branches (gray markers) show no sign of reconnecting with the main
branch (black markers), so we stopped when the calculations became
expensive. We will refer to the connected components of the main
branch as the center branches. To stay on the main branch at an
imperfect bifurcation, one has to jump from one center branch to the
next.

The shooting method parameters of these main branch solutions are
shown in figure~\ref{fig:bif:evol:1}(h). The side branch
solutions in figure~\ref{fig:bif:evol:1}(a--c)
are omitted to make the plot in figure~\ref{fig:bif:evol:1}(h)
single-valued. We switched from
quadruple-precision to double-precision at $\epsilon=0.4$ due to the
high cost of carrying out the shooting method with larger grid
sizes. We also used adaptive grids with $N\ge2$ in equation
\eqref{eq:adapt:grids} for the larger problem sizes in double and
quadruple-precision. The Fourier mode evolution of solution $D$, which
has the most `active' Fourier modes among the solutions we computed,
is shown in figure~\ref{fig:bif:evol:1}(g).  There are only 450
modes of magnitude larger than $10^{-14}$ at $t=t_0=-\pi/2$, whereas
there are close to 5000 at $t=0$. By evolving from $t=-\pi/2$ to 0
instead of $0$ to $\pi/2$, as was previously done for standing waves
\cite{mercer:94}, we reduced the dimension $d$ of $\theta$ in equation
\eqref{eq:dof} by a factor of $11$. By increasing the grid size
adaptively from $M_1=1536$ to $M_7=11664$ in this case, 60\% of the
cost goes into evolving the solution and its first variation with
respect to $\theta$ (a matrix with $d=450$ columns) through the last
$17\%$ of the simulation time. Figure~\ref{fig:bif:evol:1}(i)
shows the square root of the minimized value of the objective
  function $f(\theta)$ from equation \eqref{eq:obj:fcn} for each
solution on the main branch.  We minimize $f(\theta)$ until
floating-point error prevents further reduction. In all cases,
$\sqrt{f(\theta)}$ was reduced below $10^{-30}$ in quadruple-precision
and below $7\times 10^{-14}$ in double-precision.

Also plotted in figure~\ref{fig:bif:evol:1}(i) is the
reciprocal of the condition number of the Jacobian in the shooting
method on the final iteration of each Levenberg-Marquardt
minimization.  Each downward spike corresponds to an approximate
resonance of the nonlinear problem, where there are solutions of the
linearization about the standing wave that behave like secondary
standing waves. At a perfect bifurcation, the Jacobian is singular
\cite{quasi:bif}, and near an imperfect bifurcation, the Jacobian is
nearly singular. In the electronic supplementary material, we show how
to use the right singular vector corresponding to the smallest
singular value of the Jacobian to identify which harmonic resonance
is activated by following an imperfect bifurcation, and to
observe how strongly the wave profile of the associated secondary wave
is distorted away from being a multiple of $\cos(px)\cos(jt)$ due to
nonlinear interactions with the primary wave and itself.

\subsection{Pad\'e approximation}
\label{sec:pade}

Next we compare the unit-depth shooting method results for the
  period $T$ and a nearly resonant initial Fourier mode of the surface
  velocity potential, namely $\hat\varphi_{19}$ from equation
  \eqref{eq:shooting:fourier}, to Pad\'e approximants of their Stokes
expansions. We use continued fractions \cite{cuyt,lorentzen:book}
to efficiently represent the Pad\'e approximants of a power
series. Following \cite{cuyt,lorentzen:book}, we employ the notation
\begin{equation}
  \gaussk_{n=0}^\infty\frac{a_n}{b_n} =
  \frac{a_0}{b_0\jd+\frac{a_1}{b_1\jd+\frac{a_2}{b_2+\raisebox{-7pt}{$\ddots$}}}}, \qquad
  \frac1{\epsilon^2}\gaussk_{n=0}^2\frac{d_n\epsilon^2}{1} =
  \frac{d_0}{1\jd+\frac{d_1\epsilon^2}{1\jd+\frac{d_2\epsilon^2}{1}}},
\end{equation}
where the latter formula illustrates a finite truncation $\gaussk_{n=0}^N\cdots$
with $N=2$.  We expand the period first as a power
series and then as a continued fraction
\begin{equation}\label{eq:T:expand}
  T \; = \;
  \sum_{n=0}^\infty \tau_n\epsilon^{2n}
  \; = \; \frac1{\epsilon^2}\gaussk_{n=0}^\infty\frac{d_n\epsilon^2}{1},
\end{equation}
where the equal signs are intended in the sense of formal power series
\cite{cuyt}. Setting $g=1$ and $L=2\pi$ in equation \eqref{eq:ZFW:def} to match
the parameters used in the shooting method gives $T=2\pi\sqrt{S}$. For
any $N\ge0$, the coefficients $\tau_0,\dots,\tau_N$ are uniquely
determined from $\sigma_0,\dots,\sigma_N$ in the expansion
\eqref{stokesS} of $S$ by matching terms in
$(\sqrt{S})(\sqrt{S})=S$. We then use the quotient-difference (qd)
algorithm for continued fractions \cite{cuyt,lorentzen:book} to
compute $d_0,\dots,d_N$ from $\tau_0,\dots,\tau_N$. Note that $d_N$
only affects $\tau_n$ for $n\ge N$ in equation \eqref{eq:T:expand}. Similarly,
let $\tilde\tau_{p,n}$ and $\tilde d_{p,n}$ denote the coefficients of
the expansions
\begin{equation}\label{eq:hat:eta:phi:cfrac}
  \left\{ \!\!\!\begin{array}{cc} \hat\eta_p, & p\text{ even} \\[2pt] \hat\varphi_p, & p\text{ odd} \end{array}
  \!\!\!\right\}
  \quad = \quad \sum_{n=0}^\infty \tilde\tau_{p,n}\epsilon^{p+2n}
  \quad = \quad \frac{\epsilon^p}{\epsilon^2}\gaussk_{n=0}^\infty\frac{\tilde d_{p,n}\epsilon^2}{1},
  \qquad \big(p\ge0\big).
\end{equation}
We compute the $\tilde\tau_{p,n}$ from
$\eta(\alpha,t_0)=\im\{Z(\alpha,t_0)\}$ and $\varphi(\alpha,t_0) =
\frac{L^2}{2\pi T}\,\re\big\{F(\alpha,t_0)\big\}$ as follows. Setting
$L=2\pi$, we use equation \eqref{eq:im:Z:expand} and the analogous
equation for $\re\{F(\alpha,t_0)\}$ to obtain
\begin{equation}
  \hat\eta_p = a_p(t_0)\frac{\sinh\!\big(ph(t_0)\big)}{
    2\cosh\!\big(p\mu_0\big)}, \quad
  \left(\substack{\jd p\ge2\\[3pt]\jd p\text{ even}}\right), \qquad
  T\hat\varphi_p = 2\pi c_p(t_0) \frac{\cosh\!\big(ph(t_0)\big)}{
    2\cosh\!\big(p\mu_0\big)}, \quad
  \left(\substack{\jd p\ge1\\[3pt]\jd p\text{ odd}}\right)
\end{equation}
and  $\hat\eta_0 = [h(t_0)-\mu_0]$.
The $\epsilon^{p+2n}$ term of $a_p(t_0)\sinh\!\big(ph(t_0)\big)$ is
$\sum_{m=0}^n \alpha_{p,m}(t_0)s_{p,n-m}(t_0)$, with a similar formula
for $c_p(t_0)\cosh\!\big(ph(t_0)\big)$. Since the expansion of $T$ is
known from equation \eqref{eq:T:expand}, solving $T\hat\varphi_p=\cdots$ for
$\hat\varphi_p$ is also a simple matter of matching terms order by
order.

In our 192-digit calculation for $\mu_0=1$, we reached order $\nu=149$
in equation \eqref{fdGammaN} and computed $\tau_n$ and $d_n$ for
$0\le n\le 74$ and $\tilde\tau_{19,n}$ and $\tilde d_{19,n}$ for $0\le
n\le(\nu-19)/2=65$ via the qd-algorithm \cite{cuyt}.  Let us briefly
let $x$ denote $\epsilon^2$ rather than a spatial variable.  The $[m/k]$
Pad\'e approximant of the formal
power series $\sum_{n=0}^\infty\tau_nx^n$ is defined
\cite{cuyt,lorentzen:book} as the rational function
\begin{equation}\label{eq:pade:def}
  [m/k]_\tau(x) = P(x)/Q(x)
\end{equation}
that satisfies $P(x)-Q(x)\sum_{n=0}^{m+k}\tau_{n}x^n=O(x^{m+k+1})$,
where $P$ and $Q$ are polynomials of degree $m$ and $k$, respectively,
and $Q(0)=1$. The truncated continued fraction
$\frac1{x}\gaussk_{n=0}^N \frac{d_nx}{1}$ gives $[m/k]_\tau(x)$ with
$m=\lfloor N/2\rfloor$ and $k=\lceil N/2\rceil$, so that $m+k=N$ and
$m=k$ or $m=k-1$. Thus, truncating equations \eqref{eq:T:expand} and
\eqref{eq:hat:eta:phi:cfrac} to include the available terms
$\{d_n\}_{n=0}^{74}$ and $\{\tilde d_{19,n}\}_{n=0}^{65}$
gives the $[37/37]_{\tau}(\epsilon^2)$ and
$\epsilon^{19}[32/33]_{\tilde\tau_{19}}(\epsilon^2)$ Pad\'e
approximants of $T$ and $\hat\varphi_{19}$,
respectively.  We monitored the
floating-point arithmetic errors as explained in the electronic
supplementary material.

\begin{figure}[t]
  \begin{center}
    \includegraphics[width=\linewidth]{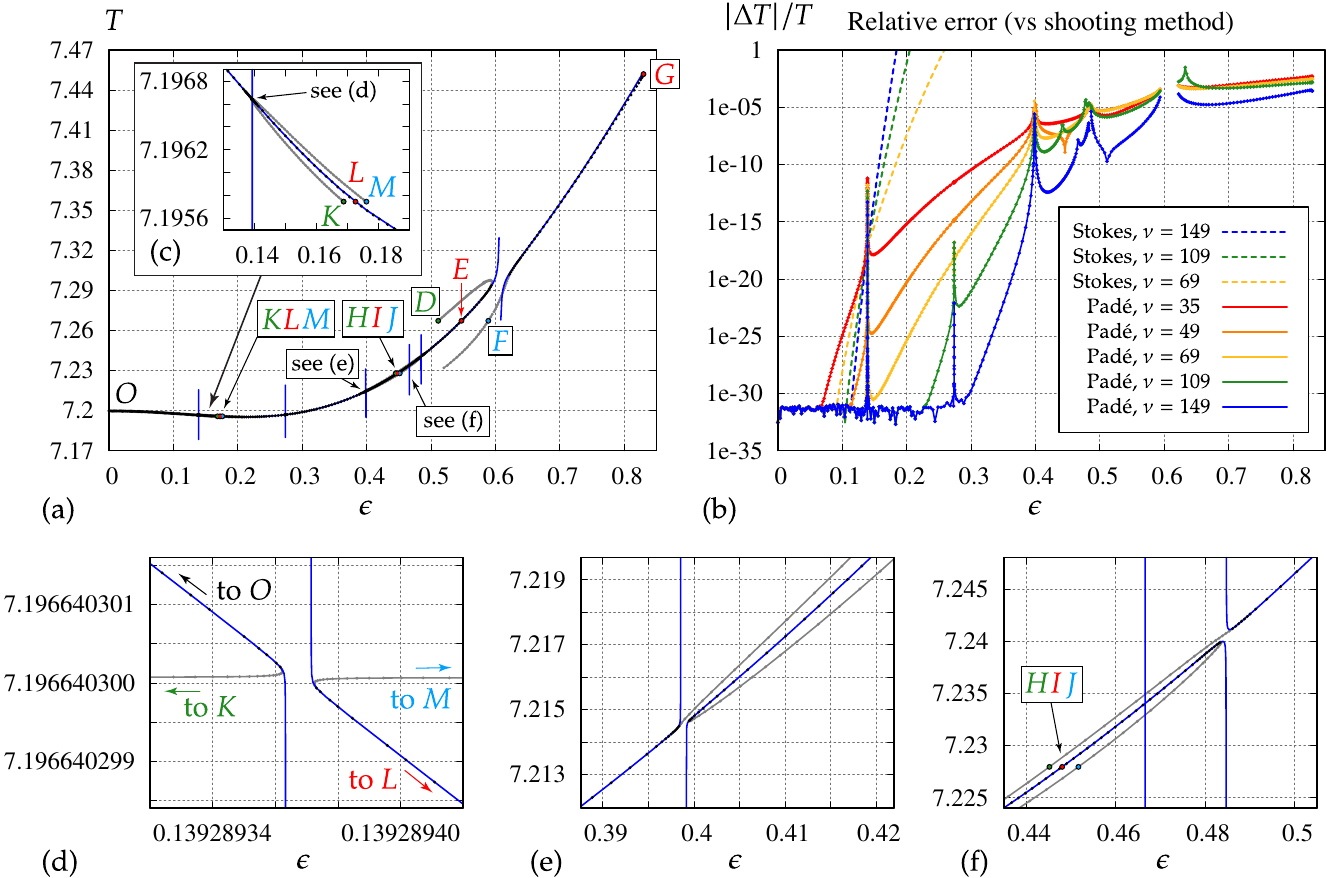}
  \end{center}
  \caption{\label{fig:bif:T1} Dependence of the period $T$ on the
    amplitude $\epsilon$ for standing waves of unit depth. (a) The
    black and gray markers are the shooting method results for
    solutions on the main and side branches, respectively. The blue
    curves are the $149^\text{th}$-order Pad\'e approximation, which
    has poles where the curves diverge from the main branch. (b)
    Relative error between the Stokes and Pad\'e expansion
      solutions at various orders and the shooting method results.
      (c,d,e,f) a closer look at the alignment of Pad\'e poles and
      imperfect bifurcations. (d) the gap contains four Pad\'e poles
      and zeros. }
\end{figure}

Figure~\ref{fig:bif:T1}(a) shows the period $T$ of the unit-depth
standing wave solutions of figure~\ref{fig:bif:evol:06}.  The blue
curve shows $[37/37]_\tau(\epsilon^2)=P(\epsilon^2)/Q(\epsilon^2)$,
which is a $149^\text{th}$-order approximation of $T$ since the first
incorrect term of its Taylor series is $O(\epsilon^{150})$. This
Pad\'e approximation has pole singularities at values of $\epsilon$
where $Q(\epsilon^2)=0$. This causes the blue curve to approach
$\pm\infty$ as $\epsilon$ approaches each pole.  In
figure~\ref{fig:bif:T1}(b), we quantify the agreement between
the shooting method solutions and the Stokes and Pad\'e expansions of
various orders. As in figure~\ref{fig:bif:evol:1}, the shooting method
solutions have been grouped into main and side branches, plotted with
black and gray markers, respectively.  For each solution on the main
branch, we compute its crest to trough height $\epsilon$ and use that
for the amplitude of the Stokes and Pad\'e expansions.  We use the
shooting method solution as a reference when computing
the relative error $|\Delta T|/T$ in the plot, where
\begin{equation}\label{eq:Delta:T:phi}
  \Delta T = T^\text{expansion} - T^\text{shooting}, \qquad
  \Delta \hat\varphi_p = \hat\varphi_p^\text{expansion} - \hat\varphi_p^\text{shooting}.
\end{equation}
The dashed lines in
figure~\ref{fig:bif:T1}(b) are the Stokes expansions
$\sum_{n=0}^{(\nu-1)/2}\tau_n\epsilon^{2n}$ of order
$\nu\in\{69,\,109,\,149\}$ while the solid lines are the Pad\'e
expansions $[m/k]_\tau(\epsilon^2)$ of order
$2(m+k)+1=\nu$, where
$m=\lfloor(\nu-1)/4\rfloor$ and $k=\lceil(\nu-1)/4\rceil$ for
$\nu\in\{35,\,49,\,69,\,109,\,149\}$.

The error curves in figure~\ref{fig:bif:T1}(b) reach a floor of
$10^{-32}$ as that is the accuracy limit of the shooting method in
quadruple-precision. In this region, $\Delta T$ and
  $\Delta\hat\varphi_p$ in equation \eqref{eq:Delta:T:phi} are dominated by the
  error in the shooting method since more precision was used in the
  Pad\'e and Stokes expansions. The Stokes expansions are extremely
accurate up to $\epsilon=0.1$, but then rapidly lose accuracy as
$\epsilon$ crosses $\rho_\nu$ in equation \eqref{eq:rho:nu:def}, which is
$0.139289367$ for $\nu\in\{109,\,149\}$ and coincides with the
amplitude where the $KLM$ bifurcation occurs in
figures~\ref{fig:bif:evol:1}(c) and~\ref{fig:bif:T1}(a). For each
shooting-method solution on the main branch, $\epsilon$ has a fixed
value and the Pad\'e approximants continue to improve in accuracy as
$\nu$ increases, even if multiple bifurcations have occurred at
smaller values of $\epsilon$.  For a given order $\nu$, the
errors in the Pad\'e approximation are largest for large $\epsilon$,
and in regions where the main branch transitions into the side
branches. In these transition regions, the output values $T$ and
$\epsilon$ from the shooting method carry the most error due to the
large condition numbers observed there in
figure~\ref{fig:bif:evol:1}(i). Poles in the Pad\'e approximants
  of $T$ make it possible to more accurately follow the main branch
toward the side branches and jump across disconnections in the
bifurcation curves. However, once $\epsilon$ approaches a pole too
closely, accuracy is lost.

\begin{figure}[t]
  \begin{center}
    \includegraphics[width=\linewidth]{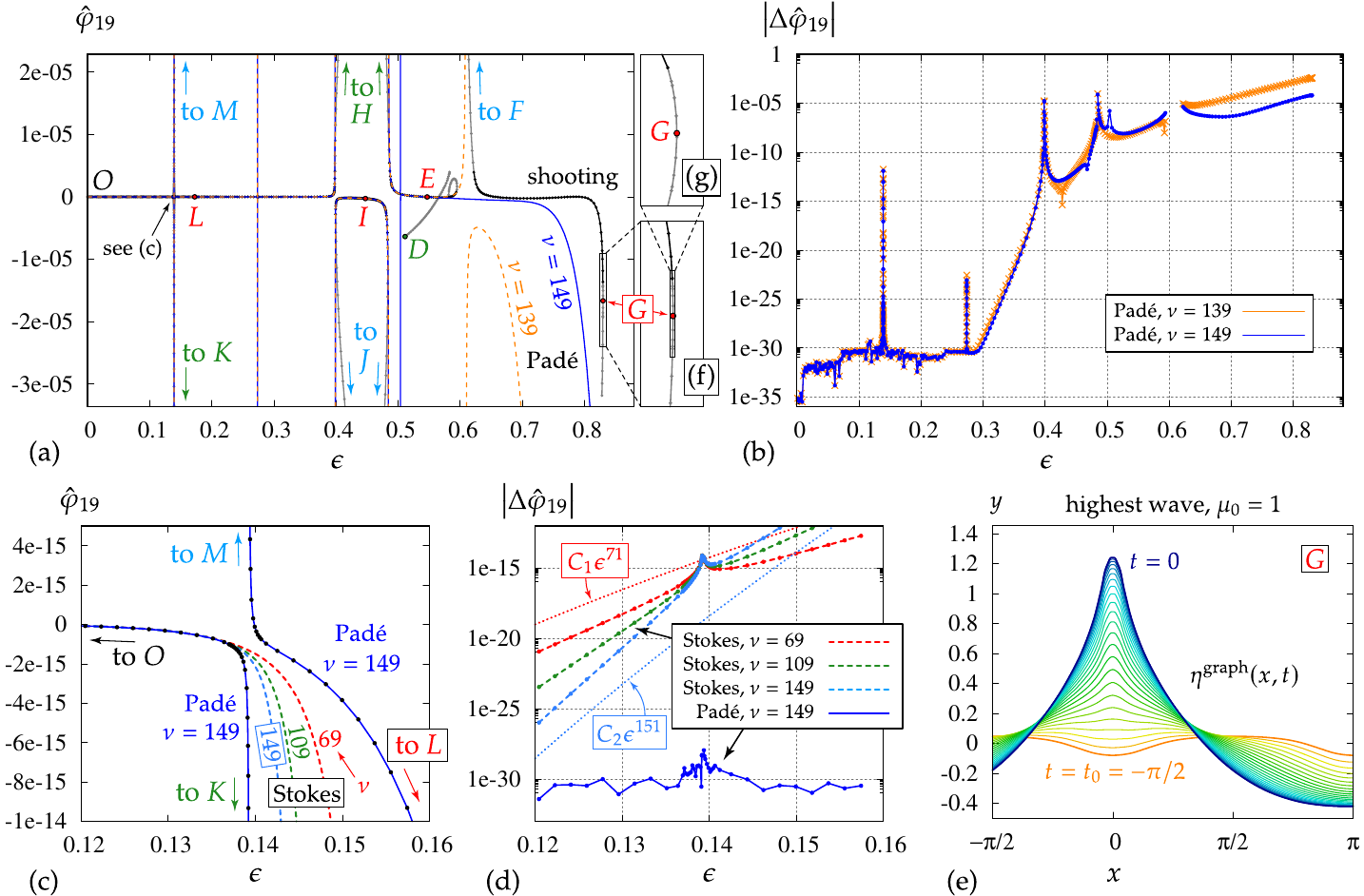}
  \end{center}
  \caption{\label{fig:bif:19} Unit-depth standing waves exhibit
    numerous imperfect bifurcations that are visible in plots of
    $\hat\varphi_{19}$ versus $\epsilon$. (a) The black and gray
    markers are the shooting method results. The solid blue curves and
    dashed orange curves show the $149^\text{th}$ and
    $139^\text{th}$-order Pad\'e approximants.  (c) A closer look at
    the first imperfect bifurcation.  The Pad\'e approximant cleanly
    jumps to the new branch while the Stokes expansions cannot.  (b,d)
    Difference between the shooting method and the Stokes or Pad\'e
    results in (a,c). (e) The highest wave, solution $G$ in
    (a,f,g), is not singular.}
\end{figure}

Figure~\ref{fig:bif:19}(a) shows $\hat\varphi_{19}$ versus
$\epsilon$ for unit-depth standing waves. The shooting method data is
the same as in figure~\ref{fig:bif:evol:1}(c), but the
$y$-axis has been scaled by a factor of 20 to better view the
imperfect bifurcations connecting the main-branch to the
side-branches. The blue and dashed orange curves are the
$149^\text{th}$ and $139^\text{th}$-order Pad\'e approximants,
$\epsilon^{19}[32/33]_{\tilde\tau_{19}}(\epsilon^2)$ and
$\epsilon^{19}[30/30]_{\tilde\tau_{19}}(\epsilon^2)$, respectively.
In addition to correctly navigating the $KLM$
bifurcation and both sides of the $HIJ$ bubble structure, both Pad\'e
approximants predict a bifurcation at $\epsilon=0.273808$. We computed
additional shooting method solutions with $\epsilon$ in this
neighborhood and find that there is indeed a bifurcation here
corresponding to the $(p,j)=(37,7)$ harmonic resonance. Details on how
the resonance was identified are given in the electronic supplementary
material. We would not have known to look for a bifurcation here
without computing the Pad\'e poles.  Figure~\ref{fig:bif:19}(b)
shows that both Pad\'e approximants plotted in panel (a) agree with
the shooting method to an absolute error of
$|\Delta\hat\varphi_{19}|<10^{-29}$ for $0\le\epsilon\le0.3$, except
in the transition regions to the side branches near
$\epsilon=0.139289$ and $\epsilon=0.273808$. We plotted the absolute
error since the shooting method involves computing $O(1)$ quantities
such as $T$ whereas $\hat\varphi_{19}$ is $O(\epsilon^{19})$. This
causes the shooting method to lose relative accuracy at small
amplitude in higher-frequency modes such as $\hat\varphi_{19}$.

Figure~\ref{fig:bif:19}(c,d) shows that the
$149^\text{th}$-order Pad\'e approximant of $\hat\varphi_{19}$
maintains absolute errors below $10^{-28}$ as it navigates the jump
across the disconnection in the bifurcation curve at
$\epsilon=0.139289$. The errors in (d) correspond to the points
  shown in (c). Further out on the side branches to $K$ and $M$, the
  errors are larger, leading to the spike in
  $|\Delta\hat\varphi_{19}|$ near $\epsilon=0.139289$ in (b). The
Stokes expansions cannot change course fast enough to follow the side
branch to solution $K$ in figure~\ref{fig:bif:19}(c), and
cannot jump branches since they are polynomials. Prior to this first
imperfect bifurcation, the error in the Stokes expansions converge at
the expected order, $O(\epsilon^{\nu+2})$.  This is demonstrated in
figure~\ref{fig:bif:19}(d) for $\nu\in\{69,149\}$ by comparing
the dashed error curves with the dotted lines showing
$C_1\epsilon^{71}$ and $C_2\epsilon^{151}$. The constants $C_i$ were
chosen to position the dotted lines near the corresponding error
curves without obscuring the plot.

The travelling water wave of maximum crest-to-trough height has
  a $120^\circ$ corner angle \cite{lhf:78}. By contrast, we find that
  the unit-depth standing wave of maximum height, solution $G$ in
  figure~\ref{fig:bif:19}(a,e,f,g), is smooth.  Snapshots of its time
evolution are given in figure~\ref{fig:bif:19}(e) at times
$t\in\mathscr{T}_{24}$ from equation \eqref{eq:scrT:def}. We located
solution $G$ using 8th degree polynomial interpolation
from the shooting method results to represent $\epsilon$ as a
function of $\hat\varphi_{19}$. The nine interpolation points
are the gray markers in figure~\ref{fig:bif:19}(g), where panels
  (f) and (g) give magnified views of the bifurcation curve in panel
  (a) near solution $G$.  Maximizing the polynomial gives
$\hat\varphi_{19}=-1.665013\times10^{-5}$ for solution $G$,
which has the maximum wave height of
$2\epsilon$ with $\epsilon=0.83016190$. The two Pad\'e
  approximants plotted in figure~\ref{fig:bif:19}(a) deviate from the
  shooting method solutions before solution $G$ is reached.  The
$139^\text{th}$-order approximant (orange dashed line) has a pole at
$\epsilon=0.6093$ that helps navigate the start of the $DEF$
bifurcation, but breaks down after that. The $149^\text{th}$-order
approximant (blue line) has a spurious pole at $0.5040$ and does not
`see' the $DEF$ bifurcation, but does a better job of tracking the
final turning point to the highest wave $G$. A pole is considered
spurious if it does not persist across multiple consecutive orders
$\nu$ of the Pad\'e approximation, or if there is no evidence of an
actual bifurcation at this location using the shooting method. The
pole in $[37/37]_\tau(\epsilon^2)$ at $\epsilon=0.4666$ in
figure~\ref{fig:bif:T1}(f) also appears to be spurious, and agrees
with a zero of the numerator to 11 digits. Such pole-zero pairs are
called Froissart doublets \cite{gonnet:13}.

\subsection{Branch cuts between turning points in the bifurcation curves}

Next we investigate how the Pad\'e approximants are able to navigate
the disconnections in the bifurcation curves so accurately.  The
$149^\text{th}$-order Pad\'e approximants of $T$ and
$\hat\varphi_{19}$ both contain four closely spaced poles and
zeros that lie in the gap between the turning points in $\epsilon$
shown in figure~\ref{fig:bif:T1}(d). From the shooting method, we find
that these turning points are located at $\epsilon_L=0.139289362345$
and $\epsilon_R=0.139289372366$.  The emergence of multiple Pad\'e
poles and zeros in this gap of width
  $\epsilon_R-\epsilon_L=1.0021\times10^{-8}$ suggests that each
function being approximated, $T(\epsilon)$ and
$\hat\varphi_{19}(\epsilon)$, has a branch cut from $\epsilon_L$ to
$\epsilon_R$ on the real $\epsilon$-axis \cite{stahl:97}.
Generalizing, it suggests that $a_p(t;\epsilon)$, $b_p(t;\epsilon)$,
$c_p(t;\epsilon)$, $S(\epsilon)$ and $h(t;\epsilon)$ in equation \eqref{Stokes}
each have a branch cut from $\epsilon_L$ to $\epsilon_R$. These
functions are real-valued for $\epsilon<\epsilon_L$ on the branches to
$O$ and to $K$ in figure~\ref{fig:bif:T1}(d), and for
$\epsilon>\epsilon_R$ on the branches to $M$ and to $L$, which
suggests that the branch points at $\epsilon_L$ and $\epsilon_R$ have
square root singularity structures.  There are no solutions in the gap
between $\epsilon_L$ and $\epsilon_R$. Analytic continuation around
the branch points into the gap (if it is possible) would lead to water
waves with a complex period, a complex fluid depth, etc., which would
be difficult to interpret physically.  The $149^\text{th}$-order
Pad\'e approximants of $T$ and $\hat\varphi_{19}$ also have two poles
in the gap between branches near $\epsilon=0.399$. Turning points have
been observed before \cite{smith:roberts:99,water2}, but we are not
aware of branch cuts being discussed previously in the context of
standing water waves.

\begin{figure}[t]
  \begin{center}
    \includegraphics[width=\linewidth]{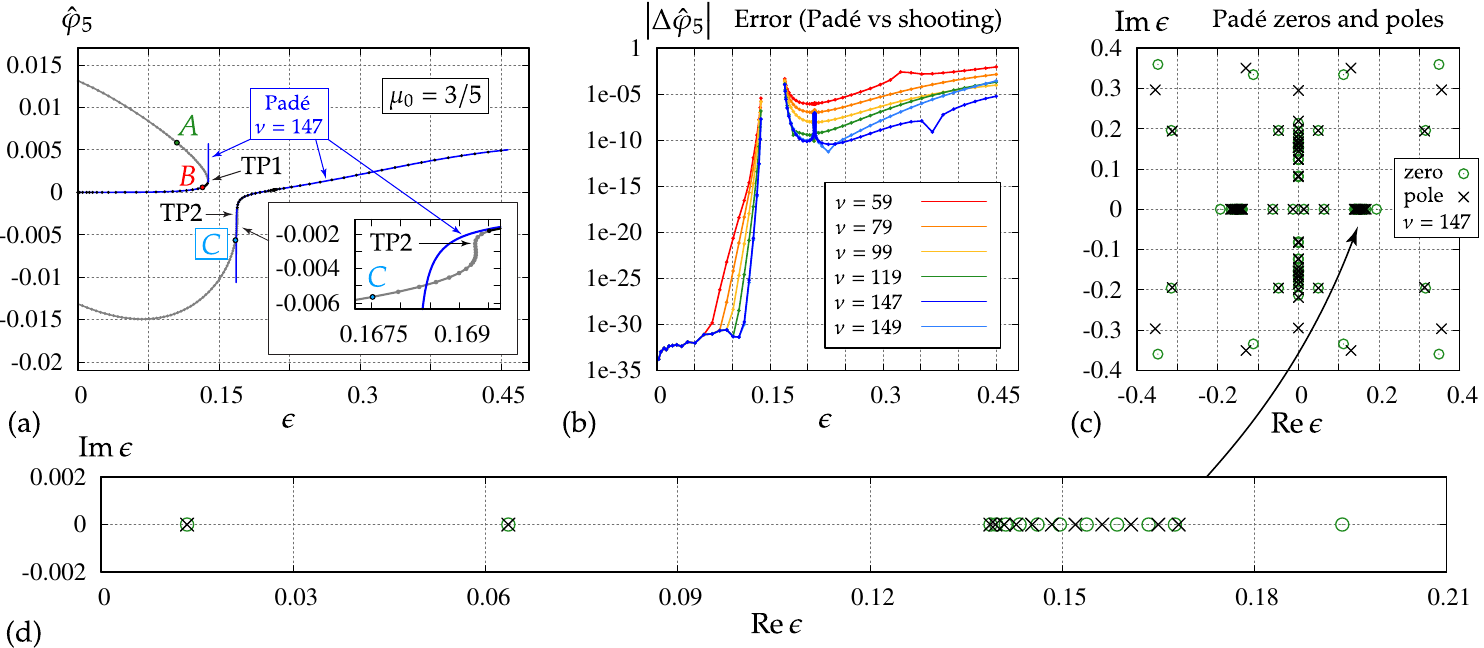}
  \end{center}
  \caption{\label{fig:pade:06} (a) The bifurcation curve for
      depth $\mu_0=3/5$ computed via the shooting method exhibits a
      large gap between the turning points labeled TP1 and TP2. (b)
      The error in the Pad\'e approximants of $\hat\varphi_5$ from equation
      \eqref{eq:shooting:fourier} generally decreases as the order
      $\nu$ increases, but not monotonically; $\nu=147$ is more
      accurate than $\nu=149$ for large $\epsilon$.  (c,d) a large
      cluster of interlaced poles and zeros of the 147th-order Pad\'e
      approximant of $\hat\varphi_5$ provides evidence of a branch cut
      in the gap.}
\end{figure}

Further evidence that standing wave families contain branch cuts is
given for the $\mu_0=3/5$ case in figure~\ref{fig:pade:06}. The blue
curve in figure~\ref{fig:pade:06}(a) is the
$147^\text{th}$-order Pad\'e approximant
$\epsilon^5[35/36]_{\tilde\tau_5}(\epsilon^2)$ of $\hat\varphi_5$
while the black and gray markers are the shooting method results of
figure~\ref{fig:bif:evol:06}.  In figure~\ref{fig:pade:06}(b),
we see that the error in the Pad\'e approximant of order $\nu$
generally decreases as $\nu$ increases, but not monotonically. The
most accurate approximant with $\nu\le149$ over the range
$0.25\le\epsilon\le 0.45$ turns out to be of order $\nu=147$.
Figure~\ref{fig:pade:06}(c) shows the zeros and poles of
$[35/36]_{\tilde\tau_5}(\epsilon^2)$ in the complex plane, without the
factor of $\epsilon^5$. These zeros and poles are the square roots of
the zeros of $P$ and $Q$ in equation \eqref{eq:pade:def}.  We obtain
  identical results whether $P$ and $Q$ are formed using continued
  fraction recurrence relations \cite{cuyt} or by finding the
nullspace of a Toeplitz matrix \cite{gonnet:13}. We use Mathematica to
find the zeros of $P$ and $Q$.  Each of these steps is done in
high-precision arithmetic to match the 192-digit precision of
the Stokes expansion of $\hat\varphi_5$ that we computed. The number
of zeros and poles that appear in the gap of width $0.030726$
between the turning points $\epsilon_L=0.13854288$ and
$\epsilon_R=0.16926877$ generally increases with the order $\nu$. We
interpret this to mean that $\hat\varphi_5(\epsilon)$ has a branch cut
on the real $\epsilon$-axis from $\epsilon_L$ to $\epsilon_R$.
Figure~\ref{fig:pade:06}(d) shows that for $\nu=147$, the
Pad\'e approximant has 11 poles and 10 zeros on this branch cut.  This
is consistent with the behavior one expects from Pad\'e approximants
of a Cauchy-Stieltjes integral \cite{allen:75} or a more general
analytic function with branch point singularities
\cite{stahl:97,yamada:14}.

\section{Conclusion}\label{sec:conclusion}

We have derived a recursive algorithm to compute successive terms of
the Stokes expansion for finite-depth standing water waves and
implemented it in arbitrary-precision arithmetic on a
supercomputer. One advantage of the conformal mapping framework over
previous graph-based approaches \cite{roberts:83,marchant:87} is that
the arguments of the hyperbolic functions in equation \eqref{eq:re:expand}
depend only on time, which reduces the cost of re-expanding the
composite power series that arise.

We carried out extensive numerical experiments to verify the
correctness of the Stokes expansions by comparing them to standing
waves computed by a shooting method \cite{water2} that we implemented
in double and quadruple-precision using adaptive meshes and numerical
continuation. While previous studies
\cite{mercer:94,smith:roberts:99,okamura:99,water2} have established a
connection between nearby harmonic resonances and the branching
structure of families of standing water waves, we look specifically at
how small divisors in the recurrence activate new growth patterns
among the Stokes coefficients. We find that the Stokes coefficients
rapidly settle into geometric growth patterns as the expansion order
$\nu$ increases, but the growth rate sometimes jumps in response to
new small divisors entering the recurrence. This led us, in equation
  \eqref{eq:rho:nu:def}, to define an inverse growth rate
factor~$\rho_\nu$ with the property that for fixed $\epsilon$,
  successive terms of the Stokes expansion transition from geometric
  decay to geometric growth when $\rho_\nu$ drops below $\epsilon$.
A clear connection between the large jumps in $\rho_\nu^{-1}$ in
figure~\ref{fig:growthFD}(a,b) and the corresponding small divisors is
demonstrated for $\mu_0\in\{1/4,3/5\}$ in the electronic supplementary
material.

In the examples we presented in detail, $\rho_\nu$ aligns with the
amplitude where an imperfect bifurcation is observed using the
shooting method. In these cases, we observe, for the first time,
clusters of poles and zeros in the Pad\'e approximants of the Stokes
expansion near $\epsilon=\rho_\nu$ that suggest that previously
observed \cite{smith:roberts:99} turning points in the bifurcation
curves are branch point singularities of an analytic function. These
poles and zeros allow the Pad\'e approximants to jump across
disconnections in the bifurcation curves with remarkable accuracy on
both sides of the branch cut. For unit-depth standing waves, the
$149^\text{th}$-order Pad\'e approximant of the period maintains 30
digits of accuracy for amplitudes up to $\epsilon=0.3$, which is
beyond the first two disconnections we identified, at
$\epsilon=\rho_{149}=0.139289$ and $\epsilon=0.273808$.  Neither of
these disconnections are `observable' in double-precision using the
shooting method alone. Mercer \& Roberts \cite{mercer:94} noted that
high-frequency resonances are likely to be extremely weak (made
  quantitative by estimates from \cite{roberts:81}), so only the
dominant, low-frequency resonances are observable in a finite
truncation of the problem carried out numerically. Our high-precision
numerics coupled with Pad\'e techniques make it possible to locate and
compute them.

For $\mu_0=1/16$, the smallest divisor for $2\le p\le 24773$ is
$\lambda_{2,2}$, which is not associated with a harmonic resonance
as $\lambda_{p,j}(\mu_0)=0$ has no solutions for
  $0<\mu_0<\infty$ and $p-j\in2\mbb Z$ unless $p\ge5$ and $j\ge3$; see
  \S\ref{sec:res:depths}.  As a result, there is no imperfect
bifurcation associated with this small divisor. Instead, as shown in
the electronic supplementary material, the closest poles to the origin
in the $109^\text{th}$ order Pad\'e approximant of the period lie on
the imaginary $\epsilon$-axis rather than the real axis. These
  poles are clustered just outside
  the radius $q(0)^{-1}$ computed from the Domb-Sykes plot in
  figure~\ref{fig:growthFD}(d), so there is likely a
  branch cut on the imaginary axis. This example shows that sometimes
  $q(0)^{-1}$ does not correspond to a real amplitude $\epsilon$ where
  a bifurcation exists.

Consistent with previous studies
\cite{mercer:94,smith:roberts:99,water2,rycroft:13,shelton:stand}, we
find that following the side branches of an imperfect bifurcation
using the shooting method activates secondary standing waves that
oscillate on top of the primary wave with different amplitudes and
phases on different bifurcation branches. This leads to non-uniqueness
for fixed values of wave amplitude or period. For example, solutions
$ABC$ in figure~\ref{fig:bif:evol:06}(c) have the same period
$T=8.45592$ but different wave heights $\epsilon$. Since solutions $A$
and $C$ appear to oscillate around solution $B$, we regard $B$ as the
primary wave and the perturbation from $B$ to $A$ or from $B$ to $C$
as a finite-amplitude secondary wave. Nonlinearity can distort these
secondary waves to deviate visibly from oscillating as scalar
multiples of $\cos(px)\cos(jt)$. This is shown in the
electronic supplementary material for both finite-amplitude and
small-amplitude secondary waves, the latter having nearly the
same period as the primary wave under the linearized equations of
motion about the primary wave.

Finite-amplitude secondary waves are specific perturbations that
maintain time-periodicity of the composite wave under the nonlinear
evolution equations. The linear stability of standing waves and
short-crested waves to arbitrary perturbations is an interesting
problem that has been studied, e.g., in
\cite{mercer:92,ioualalen:96,water:stab1}. Further exploration of the stability
of standing waves, e.g., near the bifurcations studied in the present
paper, as well as the long-time dynamics of unstable perturbations,
are interesting avenues for future research.

As one increases the order of the truncated Stokes expansion, new
small divisors occasionally enter the recurrence that lead to new
Pad\'e poles close to the origin.  We find that the Pad\'e
approximants continue to improve in accuracy at a given amplitude
$\epsilon$ as new features of the bifurcation curve emerge at lower
amplitude. This was shown in figure~\ref{fig:bif:T1}(b) for $\mu_0=1$
and in figure~\ref{fig:pade:06}(b) for $\mu_0=3/5$. In the latter
case, the poles at $\epsilon=0.013403$ and $\epsilon=0.063553$ do not
appear until $\nu=93$ and $\nu=117$, respectively. The pole at
$\epsilon=0.013403$ appears shortly before the large jump in
$\rho_\nu^{-1}$ for $\mu_0=3/5$ in figure~\ref{fig:growthFD}(b).
Rigorous existence proofs of standing waves
\cite{plotnikov01,iooss05,alazard15} and temporally quasi-periodic
water waves \cite{berti2016quasi} employ a Nash-Moser iteration to
rapidly converge to a solution through a sequence of less regular
spaces.  However, this only establishes existence for values of the
amplitude parameter in a Cantor set. The gaps
  $\epsilon_R-\epsilon_L$ between turning points in the bifurcation
  curves are generally smaller for higher wave number resonances.  For
  example, the gap for $p=19$ in figures~\ref{fig:bif:T1}(d)
  and~\ref{fig:bif:19}(c) is six orders of magnitude smaller than the
  gap for $p=5$ in figure~\ref{fig:pade:06}(a).  As more poles and
  branch cuts appear at higher orders in the Pad\'e approximants, new
  gaps in parameter space emerge in which there is no solution.  An
  interesting question is whether this process of removing smaller and
  smaller gaps leaves behind a cantor set of values of $\epsilon$
  where the Pad\'e approximants converge to a solution of the standing
  wave problem.

Theorem~\ref{thm:nonresonant} shows that for almost every fluid depth,
and every rational depth, the divisors $\lambda_{p,j}$ are bounded
below by a slowly decaying function of the wave number $p$.
These lower bounds appear to limit the rate at which $\rho_\nu$
  from \eqref{eq:rho:nu:def} approaches 0 as $\nu\to\infty$.  In the
  electronic supplementary material, we provide an example to show
  that if the analogue of $\rho_\nu$ approaches zero too rapidly, the
  Pad\'e approximants do not converge to the underlying analytic
  function between its branch cuts.

In the present work, we focused on continued fraction expansions of
single components of the Stokes expansion, namely $T$, $\hat\varphi_p$
and $\hat\eta_p$ in equations \eqref{eq:T:expand} and
\eqref{eq:hat:eta:phi:cfrac}.  Roberts \cite{roberts:83} and Marchant
\& Roberts \cite{marchant:87} also considered Pad\'e approximants of
scalar quantities associated with Stokes expansions of short-crested
waves. It would be natural in future work to study multivariate
  rational approximations of the solutions \cite{guillaume:00}.
  For example, if the continued fraction coefficients $d_n$ and
$\tilde d_{p,n}$ in equations \eqref{eq:T:expand} and
\eqref{eq:hat:eta:phi:cfrac} are well-approximated by rational
functions of the depth parameter and have simple pole singularities at
a resonant depth, the truncated Pad\'e approximants at a given order
$\nu$ would become bivariate rational functions of depth and
amplitude.  This would use Robert's idea \cite{roberts:81} to extend
the validity of a Stokes expansion past discontinuities in the
bifurcation curve (to larger values of $\epsilon$) for many depths
simultaneously. It would also provide a satisfactory answer to a
concern raised by Concus \cite{concus:64} that a small perturbation of
the depth would cause discontinuous changes in the Stokes expansion
coefficients.  Although these coefficients change discontinuously, the
solution itself (the Pad\'e approximant) depends continuously on depth
in a neighborhood of the resonant depth on the large-amplitude side of
the discontinuity.  There will only be finitely many resonant depths
for a given truncation order of the Stokes expansion, and this
approach could be used to represent solutions in a neighborhood of any
of them.  One could also use Pad\'e techniques for Fourier series
\cite{daras:pade:Lp} to obtain a rational function of $e^{-iw}$ to
approximate the Fourier series one obtains by truncating the sums in
equation \eqref{ansatz} to a finite range $1\le p\le p_\text{max}$.  This would
generalize (from travelling waves to standing waves) the results of
Dyachenko et al \cite{sergey:16} and Lushnikov \cite{lushnikov:jfm:16}
on branch point singularities in the upper half-plane obtained by
analytic continuation of the conformal map.  Standing waves would have
the new feature that these singularities evolve in time.

\vspace*{1.5pc}
\noindent
\textbf{Acknowledgments.}\,
The authors thank the anonymous reviewers for many valuable comments
and suggestions that strengthened the results of the paper.
AA was supported in part by the ARCS Foundation through the
ARCS Scholar program.  JW was supported in part by the National
Science Foundation under award number DMS-1716560 and by the
U.S. Department of Energy, Office of Science, Office of Advanced
Scientific Computing Research under Contract No. AC02-05CH11231.
This research used the Lawrencium cluster at the Lawrence Berkeley
National Laboratory and the Savio cluster at UC Berkeley.

\vspace*{1.5pc}
\noindent
\textbf{Data Accessibility.}\, \parbox[t]{4in}{
Our source code and data are available at \\
https://doi.org/10.5281/zenodo.15585007. \cite{abassi:semi1:zenodo}
}

\bibliographystyle{RS}

%\bibliography{refs}

% \input{supplementary1}

%\appendix

\newpage
%\vspace*{0.5in}

\begin{center}
  \large Electronic supplementary material
\end{center}

\setcounter{section}{0}

\makeatletter
\def\@seccntformat#1{\@ifundefined{#1@cntformat}%
   {\csname the#1\endcsname\space}%    default
   {\csname #1@cntformat\endcsname}}%  enable individual control
\newcommand\section@cntformat{\thesection.\space} % section-level
\makeatother
\renewcommand{\thesection}{S\arabic{section}}
\counterwithin{equation}{section}
% \counterwithin{figure}{section}
% \counterwithin{table}{section}

In \S\ref{sec:sdiv:growth} we plot the smallest divisor that
arises at each wave number for the depths considered in
\S\ref{sec:growth} and examine how the growth patterns in the
Stokes coefficients change when unusually small divisors enter the
recurrence. We also fill in details on the Domb-Sykes plot analysis in
figure~\ref{fig:growthFD}(d,e,f) and give an example
where the rapid growth of the Stokes expansion coefficients
corresponds to Pad\'e poles on the imaginary $\epsilon$-axis near the
origin.
In \S\ref{sec:pos:rad}, we show that there is no fluid depth for
which $\lambda_{p,j}$ in equation \eqref{eq:lam:lamcap} can be uniformly
bounded away from 0 for $p\ge2$.  We also show that this is not
  the case for the gravity-capillary standing wave problem, i.e.,
  $\lambda_{p,j}^\text{cap}$ can sometimes be bounded away from zero
  in spite of the density of resonant bond numbers.  We then show
that the Pad\'e approximants of an analytic function with a sequence
of branch cuts that accumulate at the origin may or may not converge
to the function at larger amplitudes, depending on how rapidly the
analogue of $\rho_\nu$ from equation \eqref{eq:rho:nu:def} converges
to zero as $\nu\to\infty$.
In \S\ref{sec:proof32}, we prove Theorem~\ref{thm:nonresonant},
which states that for almost every fluid depth, the small
divisors are bounded below by a slowly-decaying function of the wave
number.  We demonstrate for $\mu_0\in\{1/16,1/4,0.2499\}$ that
Theorem~\ref{thm:nonresonant} gives a good prediction of how fast the
small divisors decay in practice. We also state and prove a theorem on
the presence of many large divisors in the proximity of any small
divisor.
In \S\ref{FDSecODEStokesCoefficients}, we present a brief
derivation of the forcing terms that appear in the ODEs of
\S\ref{FDSecSolutionAlgorithm}. In
\S\ref{sec:comp:aspects}, we discuss computational aspects of
the algorithm and provide implementation details for our
arbitrary-precision parallel algorithm. In
\S\ref{sec:floating:point}, we discuss the effects of
finite-precision arithmetic and how to estimate floating-point
errors. In \S\ref{sec:secondary:waves}, we investigate the
secondary standing waves that are activated with different phases and
amplitudes by following the side branches of the bifurcation curves of
\S\ref{sec:bif}, focusing on how nonlinearity affects the shapes
of the secondary waves.  Finally, in \S\ref{sec:identify:hr}, we
show how to identify which harmonic resonance is responsible for a
bifurcation branch by studying the singular vector corresponding to
the smallest singular value of the Jacobian of a solution near the
imperfect bifurcation.

\section{Small divisors, growth rates, and imaginary Pad\'e poles}
\label{sec:sdiv:growth}

The jumps in growth rate in figure~\ref{fig:growthFD}(a,b)
appear to be caused by new small divisors entering the recurrence at
certain orders when solving equation \eqref{eq:lam:pj:def} for
$\alpha_{p,n,j}$. This alters the growth patterns of the Stokes
expansion coefficients $\alpha_{p,n,j}$, $\beta_{p,n,j}$,
$\gamma_{p,n,j}$, $\mu_{n,j}$ and $\sigma_n$ on subsequent iterations.
This observation is implicitly made in \cite{roberts:83,marchant:87},
though the authors focus on zero divisors of nearby resonant depths
rather than small divisors of the actual recurrence.
Figure~\ref{fig:smallDiv} shows the smallest divisor $\lambda_p$
associated with each spatial mode, defined in equation \eqref{eq:lam:p:def}
above. Notable small divisors at each depth are labeled with triples
$(p,j,\lambda_p)$, where $j$ is the argmin in \eqref{eq:lam:p:def}.

These small divisors have a strong effect on the Stokes expansion
coefficients.  In the case $\mu_0=1/4$, the smallest divisor seen for
$2\le p\le 47$ is $\lambda_2=0.226$.  It then drops by a factor of
13.6 at $p=48$ to $\lambda_{48}=0.0166$.  The $(p,j)=(48,14)$ mode
becomes active at order $\nu=48$, but it starts out much smaller in
magnitude than the largest mode of that order. Specifically,
$\alpha_{48,0,14}=6.62\times 10^{66}$ while
$\alpha_{2,23,2}=-1.735\times 10^{82}$. From that point, as $\nu$
increases through even integers, $\alpha_{48,(\nu-48)/2,14}$ grows
faster than any other mode of order $\nu$. By the time $\nu=70$,
$\alpha_{48,11,14}=1.463\times10^{123}$ overtakes
$\alpha_{2,34,2}=4.22\times10^{121}$ as the largest (in magnitude)
coefficient $\alpha_{p,n,j}$ with $p+2n=\nu$ and $j\in E_\nu$. This is
precisely where the jump in the growth rate
$\rho_\nu^{-1}=\sqrt{A_\nu/A_{\nu-2}}$ of the norms $A_\nu$ defined in
equation \eqref{eq:A:nu:def} appears in figure~\ref{fig:growthFD}(a)
for the case $\mu_0=1/4$.  Similar observations hold for the case
$\mu_0=3/5$, where the small divisor $\lambda_{65,11}=0.03166$ leads
to the jump in $\rho_\nu^{-1}$ at $\nu=103$.  The case $\mu_0=1$ is
interesting as there is a cluster of 3 moderately small divisors
$\lambda_{p,j}$ with $(p,j)\in\{(7,3),(12,4),(19,5)\}$.  We find that
$\alpha_{19,n,5}$ is among the largest modes for $n\ge23$, which is
the transition region where $\rho_\nu^{-1}$ moves up to the next
plateau for $\nu=19+2n\ge65$ in figure~\ref{fig:growthFD}(b)
for $\mu_0=1$. Small divisors are not the only consideration in
determining which mode is largest: the right-hand side
$-S_{p,n,j}$ in equation \eqref{eq:lam:pj:def} depends on the previously
computed $\alpha_{q,m,l}$ in a complicated way. This causes other
modes $\alpha_{p,n,j}$ with $(p,j)$ near $(19,5)$ to also be large for
$n\ge23$ and $\mu_0=1$.

\begin{figure}[t]
  \begin{center}
    \includegraphics[width=\linewidth,trim=8 0 0 0,clip]{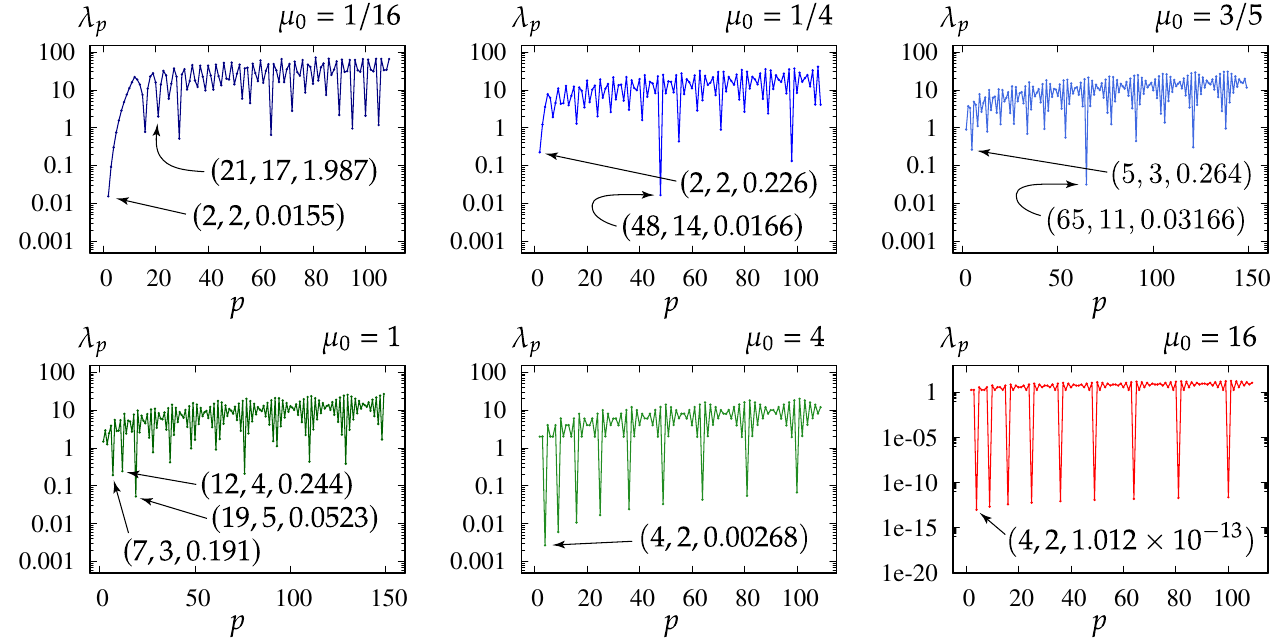}
  \end{center}
  \caption{\label{fig:smallDiv} Smallest divisor that will arise in
    the calculation for each spatial Fourier mode $p\ge2$. The labels
    $(p,j,\lambda_p)$ give the parameters of the smallest
    $\lambda_p$'s encountered, where $j$ is the argmin in
    equation \eqref{eq:lam:p:def}.  }
\end{figure}

Before discussing the case of imaginary Pad\'e poles, we fill in some
details omitted in \S\ref{sec:growth} on how we estimate the
radius of convergence of the Stokes expansion from the Domb-Sykes
plots in figure~\ref{fig:growthFD}(d,e,f). We fit the data
  $\big(s_\nu,\rho_\nu^{-1}\big)$ with the polynomial $q(s)$ of degree
$d$ that minimizes $f=\sum_{\nu\in \mc I}|\rho_\nu^{-1} -
q(s_\nu)|^2w_\nu$, where
\begin{equation}
  s_\nu=1/\nu, \qquad w_\nu=C/(D-\nu)^2, \qquad D=4+\max_{\nu\in\mc
    I}\,\nu,
\end{equation}
and $C$ is chosen so that $\sum_{\nu\in\mc I}w_n=1$.  We used
the parameters
\begin{equation*}
  \begin{array}{c|c|c|c}
      & \mu_0=1 & \mu_0=1/16 & \mu_0=\infty \\ \hline
    d & 2 & 4 & 8 \\
    \mc I & \{\nu\,:\,10\le\nu\le54\}\cap2\mbb Z &
    \{\nu\,:\,20\le\nu\le109\}\cap2\mbb Z &
    \{\nu\,:\,20\le\nu\le149\}
  \end{array}
\end{equation*}
After finding $q(s)$, we estimate
$\lim_{\nu\to\infty}\rho_\nu^{-1}=q(0)$ to obtain the extrapolated
radius of convergence $q(0)^{-1}$ from the ratio test.
Omitting the odd integers eliminates oscillations in the residual
$[\rho_\nu^{-1}-q(s_\nu)]$ that arise for $\mu_0=1/16$ and
$\mu_0=1$. The results for $q(0)$ agree to all the digits reported in
figure~\ref{fig:growthFD}(d,e,f) if we instead omit the even
integers. When $\mu_0=\infty$, these oscillations are not present, so
we include both odd and even integers in $\mc I$.  More details on
choosing the polynomial order $d$ to maximize accuracy without
over-fitting will be given elsewhere \cite{abassi:semi2}. This choice
of weight $w_\nu$ favors accuracy for larger values of $\nu$, where
$s_\nu$ is closer to 0. As discussed in \S\ref{sec:growth} and
shown in figure~\ref{fig:growthFD}(f), when $\rho_\nu^{-1}$
  jumps from one plateau height to another, the extrapolated value
  $q(0)$ increases, indicating that a new singularity
  $\epsilon_*\in\mbb C$ has been detected closer to the origin, near
  the smaller radius $|\epsilon_*|\approx q(0)^{-1}$.  These
  singularities appear as poles in the Pad\'e expansions and do not
  vanish when the order $\nu$ of the Pad\'e approximation increases
  and new poles emerge closer to $\epsilon=0$. Instead, as shown in
  figure~\ref{fig:pade:06}(d), they often `fill in' to become
  clusters of pole-zero pairs in what appear to be branch
  cuts. Regardless of whether $\rho_\nu\to0$ as $\nu\to\infty$, the
  predicted values of $q(0)^{-1}$ from the Domb-Sykes plots using
  different sets $\mc I$ for the extrapolation are useful for
  locating singularities.  It will be shown in \cite{abassi:semi2}
that no jumps in $\rho_\nu^{-1}$ are encountered up to order
$\nu=641$ for $\mu_0=\infty$, and the polynomial $q(x)$ computed here
from $20\le\nu\le149$ satisfies
$\max_{150\le\nu\le641}\big(|\rho_\nu^{-1}-q(s_\nu)|/\rho_\nu^{-1}\big)\le8.0\times10^{-11}$.
This suggests that the radius of convergence for the infinite depth
case (the Schwartz and Whitney expansion) is positive:
$q(0)^{-1}=0.301262103>0$.  However, we have some doubts about this
conclusion as Pad\'e approximants of individual components of the
solution ($T$, $\hat\varphi_p$ and $\hat\eta_p$) possess extremely
closely spaced Froissart doublets inside this radius, similar to what
happens in the $\mu_0=1/16$ case reported below. We will explore this
in more detail in \cite{abassi:semi2}. We are not yet able to reach
such high orders in finite depth.

For $\mu_0\in\{1/4,3/5,1\}$, the small divisors
$\big\{\{\lambda_{48,14}\}, \{\lambda_{5,3},\lambda_{65,11}\},
\{\lambda_{7,3},\lambda_{12,4},\lambda_{19,5}\}\big\}$ lead to
imperfect bifurcations in the family of solutions.  But for
$\mu_0=1/16$, the small divisor $\lambda_{2,2}$ is not associated with
a harmonic resonance since the first resonance in finite depth occurs
at $(p,j)=(5,3)$.  Figure~\ref{fig:imaginary:poles}(a,b) shows the poles and zeros of the
$109^\text{th}$-order Pad\'e approximant of the period~$T$.  As $\nu$
increases, the poles and zeros become more densely distributed on the
imaginary axis with an accumulation point emerging at the blue circle
of radius $q(0)^{-1}=0.000267885$, which is the extrapolated radius of
convergence of the Stokes expansion computed from the Domb-Sykes plot,
as described above.  The dotted red circle in figure~\ref{fig:imaginary:poles}(b) has radius
$\rho_{109}=0.000271628$, which is slightly larger than $q(0)^{-1}$
since $q(s)$ in figure~\ref{fig:growthFD}(d) increases as
$s\to0^+$. All but one pole on the positive imaginary axis in
figure~\ref{fig:imaginary:poles}(a,b) lie outside of the
blue circle. (The conjugate of any pole or zero is also a pole or
  zero, so we focus on those in the upper half-plane.)  The one
exception is the Froissart doublet labeled $FD$ in figure~\ref{fig:imaginary:poles}(b). This
doublet consists of a pole $z_p$ and a zero $z_0$ on the imaginary
axis near $0.00022524i$ that are separated by a relative distance of
only $|z_p-z_0|/|z_p|=4.4\times10^{-21}$. Details confirming that
64-digit (212-bit) floating-point arithmetic is sufficient to compute
this relative distance are given in \S\ref{sec:floating:point}
below.

\begin{figure}[t]
  \begin{center}
    \includegraphics[width=\linewidth]{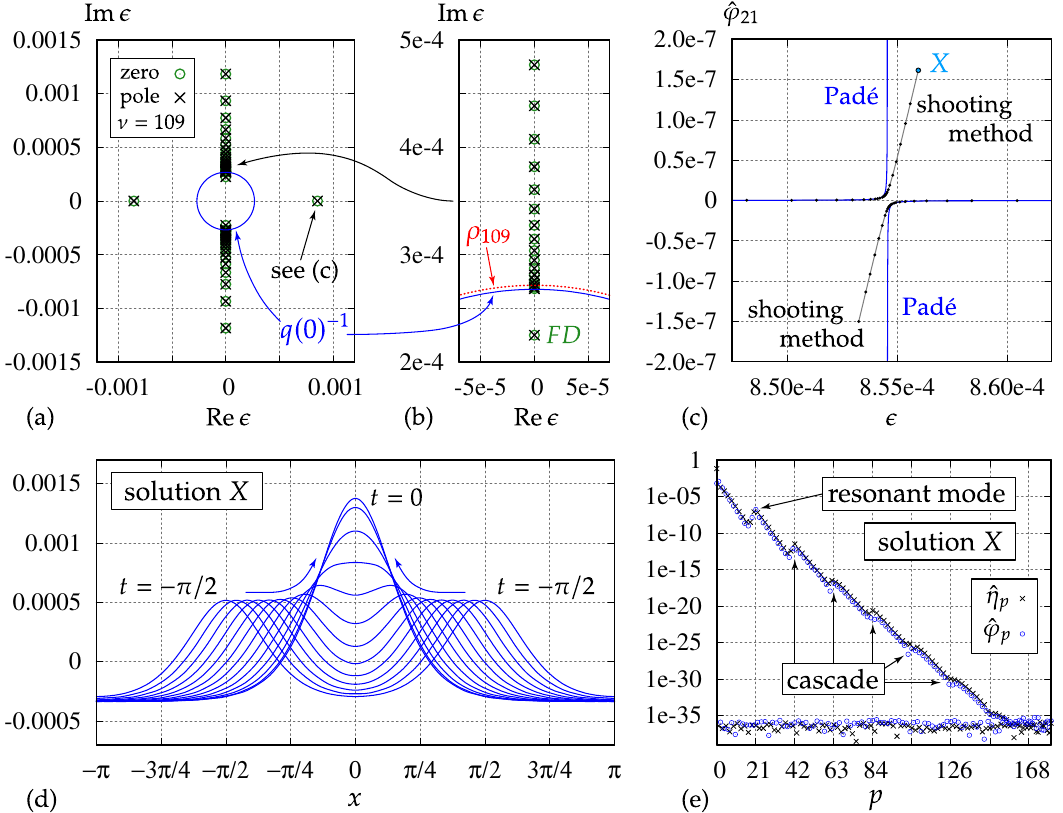}
  \end{center}
  \caption{\label{fig:imaginary:poles} Standing waves of
      dimensionless depth $\mu_0=1/16$.  (a,b) Poles and zeros of the
      $109^\text{th}$-order Pad\'e approximant of the period $T$. (c)
      A bifurcation plot of solutions near the real pole at
      $\epsilon=0.000854552$. (d) Snapshots of solution $X$ at times
      $t\in\mathscr{T}_{12}$ from equation \eqref{eq:scrT:def}. (e) The
      Fourier spectrum of the initial conditions of solution $X$.}
\end{figure}

Because the pole and zero of the Froissart doublet are so close
together, the Pad\'e approximant
$[27/27]_\tau(\epsilon^2)=P(\epsilon^2)/Q(\epsilon^2)$ agrees to more
than 20 digits on the real axis with the rational function
$[P(\epsilon^2)/(\epsilon^2-z_0^2)]/[Q(\epsilon^2)/(\epsilon^2-z_p^2)]$,
which has all its poles outside of the blue circle in
figure~\ref{fig:imaginary:poles}(a). The leading terms of their Taylor
expansions will also be close to each other, which helps explain why
$\rho_{109}$ and $q(0)^{-1}$ are close to the second-smallest pole
rather than the pole of the Froissart doublet. This doublet persists
over many consecutive orders, appearing first at order $\nu=67$, and
remains close to $0.00022524i$ for $75\le\nu\le109$. This suggests
that there is an actual singularity near this location that should
eventually cause $\rho_\nu$ to drop below $0.00022524$. If the
underlying singularity is not a simple pole, additional poles may
appear near this one in higher-order Pad\'e approximants. The other
poles on the imaginary axis in figure~\ref{fig:imaginary:poles}(a) are evidence that branch cuts
may exist along the imaginary axis above and below the blue circle.
We will see in \S\ref{sec:proof32} that
$\lambda_p\ge\max(0.0155,p^{-0.57})$ for $2\le p\le 2.5\times
10^{22}$, which is a slowly decaying lower bound. The rate at which
$\rho_\nu$ approaches zero as $\nu\to\infty$ (assuming it does so)
depends on how these small divisors interact with the terms in the
right-hand side $-S_{p,n,j}$ in equation \eqref{eq:lam:pj:def}, and on whether
the abundance of large divisors discussed in \S\ref{sec:proof32}
below are helpful in preventing rapid decay. (A slower decay rate in
  $\rho_\nu$ appears to improve the convergence of the Pad\'e
  approximants, as discussed in \S\ref{sec:pos:rad} below.)

In figure~\ref{fig:imaginary:poles}(a), we see that the
$109^\text{th}$-order Pad\'e approximant of $T$ has a pole-zero pair
on the real $\epsilon$-axis at $\epsilon=0.000854552$. (There is
  another such pair at $\epsilon=-0.000854552$).  The relative
distance between the pole and zero is $6.0\times10^{-11}$. This
pole-zero pair persists over many consecutive orders, appearing first
at order $\nu=69$ and remaining close to $\epsilon=0.00085$ for
$81\le\nu\le109$. We used the Pad\'e approximants of $T$ and of the
Fourier modes $\hat\varphi_p$ and $\hat\eta_p$ in equation
\eqref{eq:shooting:fourier} as an initial guess in the shooting method
to construct a bifurcation plot via numerical continuation near this
amplitude, shown in figure~\ref{fig:imaginary:poles}(c). The
most resonant mode, $\hat\varphi_{21}$, is plotted versus $\epsilon$.
The black markers show shooting-method solutions while the blue curve
shows the $109^\text{th}$-order Pad\'e approximant
$\epsilon^{21}[22/22]_{\tilde\tau_{21}}(\epsilon^2)$ of
$\hat\varphi_{21}$ obtained by truncating the continued fraction
\eqref{eq:hat:eta:phi:cfrac} at $n=44$. Instead of turning points that
leave a gap with no solutions, there are two solutions at each value
of $\epsilon$ in a neighborhood of $\epsilon=0.000854552$, one on each
bifurcation branch.

Figure~\ref{fig:imaginary:poles}(d) shows snapshots of the
time-evolution of solution $X$ over a quarter period, where $X$ is the
labeled solution on the upper bifurcation branch in panel (c).  Since
the depth $\mu_0=1/16$ is small in comparison to the wavelength
$2\pi$, standing waves take the form of counter-propagating solitary
waves that repeatedly collide at dimensionless times $t\in2\pi\mbb Z$
and $t\in\pi(1+2\mbb Z)$ to form rest states with localized peaks
centered at $x=0$ and $x=\pm\pi$, respectively. While secondary
standing waves are not visibly active in this solution, the Fourier
spectrum of the initial condition, shown in
figure~\ref{fig:imaginary:poles}(e), suggests that a
resonance in the $21^\text{st}$ spatial Fourier mode leads to a
cascade of peaks at wave numbers $p$ that are multiples of $21$. The
smallest divisor associated with $p=21$ is the $(p,j)=(21,17)$ mode
labeled in figure~\ref{fig:smallDiv}. Using the method described below
in \S\ref{sec:identify:hr}, we confirm that there is a solution
of the linearized water wave equations about solution $X$ that
contains 21 spatial oscillations and 17 temporal oscillations and is
nearly time-periodic with the same period as $X$, namely $T=25.0633$.
This small-amplitude, secondary standing wave solution corresponds to
the smallest singular value of the Jacobian of solution $X$, which is
$\sigma_\text{min}=1.179\times10^{-6}$. The second-smallest and
largest singular values are $0.00144$ and $1.407$, respectively.  We
omit a contour plot of the linearized solution in the interest of
space and defer a discussion of the method to
\S\ref{sec:identify:hr} below.  Solution $X$ was computed in
quadruple-precision using $M=432$ spatial gridpoints; 288 timesteps of
a $15^\text{th}$-order spectral deferred correction method \cite{dutt}
over a quarter period; and $d=140$ unknown initial conditions in equation
\eqref{eq:dof}, where $\hat\varphi_{21}$ is omitted from the list of
unknowns and specified as an amplitude parameter. It is interesting
that the $(21,17)$ mode becomes much more resonant evolving over
solution $X$ than over flat water. The divisor for this mode is not
particularly small: $|\lambda_{21,17}|=1.987$.

\section{Divergence of the Stokes expansion and convergence of its Pad\'e approximants}
\label{sec:pos:rad}

Roberts \cite{roberts:81,roberts:83} concludes from the density of
resonant depths that the asymptotic expansions of standing and
short-crested water waves have a zero radius of convergence for all
values of the depth parameter. This conclusion is re-iterated in
\cite{marchant:87,mercer:94,ioualalen:96,smith:roberts:99}.
We agree that this is likely to be true but disagree that it
automatically follows from the density of resonant depths. For a given
non-resonant depth, the recursion of \S\ref{sec:recursive:alg}
is a specific procedure and nearby resonances enter into it only
through division by $\lambda_{p,j}$ in equation \eqref{eq:lam:pj:def}. Resonant
depths lead to zero-divisors while non-resonant depths could
potentially lead to small divisors.

A distinction should be made between nearby resonant parameters and
small divisors. We illustrate this for the gravity-capillary
wave problem in infinite depth \cite{abassi:semi2}, where the Stokes
expansion recursion involves division by $\lambda^\text{cap}_{p,j}$
from equation \eqref{eq:lam:lamcap} instead of $\lambda_{p,j}$.  We have
discovered that there are many values of $B$ in equation \eqref{eq:lam:lamcap}
for which $\lambda^\text{cap}_p=\min_{j\in p+2\mbb
  Z}|\lambda^\text{cap}_{p,j}|$ is bounded away from 0 for all
$p\ge2$, even though $B$ is an accumulation point of resonant Bond
numbers. We will explore this in more detail elsewhere
\cite{abassi:semi2}, but give the example $B=1$ here.  The equation
$y^2=x^3+4x$ is an elliptic curve whose integral points are enumerated
in \cite{tate} by constructing a generating set for the Mordell-Weil
group of rational points on the curve.  The only integer solutions
(with $y\ge0$) turn out to be $(x,y)=(0,0)$ and $(x,y)=(2,\pm4)$;
see~Example 3.11 of \cite{tate}.  Setting $x=2p$ and $y=4j$, we obtain
$16j^2=8p^3+8p$, i.e., $\lambda^\text{cap}_{p,j}=0$.  If $p$ and $j$
are integers, so are $x$ and $y$, so the only solutions of
$\lambda^\text{cap}_{p,j}=0$ are $(p,j)=(0,0)$ and $(p,j)=(1,\pm1)$.
It follows that $\lambda^\text{cap}_p=\min_{j\in p+2\mbb
  Z}\big|\frac12\big(p+p^3\big)-j^2\big|\ge1$ for $p\ge2$. However,
setting $B^\e{m}=\big(1-\frac1{2m^4}\big)$ for $m\in\mbb N$, one finds
that $\lambda^\text{cap}_{p,j}(B^\e{m})=0$ for $j=2m^3$ and $p=2m^2$,
so $B^\e{m}$ is a sequence of resonant Bond numbers converging to
$B=1$. It is an open question whether the Stokes expansion has a
positive radius of convergence in such cases.

For the finite-depth problem with zero surface tension, there is no
depth $\mu_0$ for which $\lambda_p$ is uniformly bounded away from 0
for all $p\ge2$. This is equivalent to the assertion in
Theorem~\ref{thm:nonresonant} that $\mc E_{-1/2}$ is the empty set;
see \S\ref{sec:proof32} below.  Thus, for every depth, there
will be arbitrarily small divisors eventually, but for almost every
depth, $|\lambda_{p,j}|$ can only become small when $p$ is large, as
quantified by Theorem~\ref{thm:nonresonant}. Determining whether the
algorithm of \S\ref{sec:recursive:alg} leads to a series with a
positive radius of convergence is difficult since the formulas for the
forcing terms $T^r_{p,n}$, given in equations
\eqref{eq:T2:def}--\eqref{fdBernoulliGreekT4} below, are nonlinear and
contain factors such as $q\alpha_{q,k}$ or
$\dot\alpha_{q,l}(t)=\sum_{j\in E_{q+2l}}ij\alpha_{q,l,j}e^{ijt}$ in
which a spatial or temporal derivative amplifies higher-frequency
modes in proportion to the mode index $q$ or $j$.  We agree with
previous authors
\cite{roberts:81,roberts:83,marchant:87,mercer:94,ioualalen:96,smith:roberts:99}
that it is likely that the Stokes expansion diverges for every fluid
depth. We believe that the bound
$\lambda_p\ge\min(a,p^{-\frac12-\delta})$ in
Theorem~\ref{thm:nonresonant}, together with the presence of many
large divisors (see \S\ref{sec:proof32} below), will limit the
growth rate of the Stokes coefficients so that $\rho_\nu$ in equation
\eqref{eq:rho:nu:def} approaches zero slowly. A reasonable conjecture
is that if $\mu_0\in\mc E_\delta$, there exist $C>0$ and $\nu_0\ge3$,
both depending on $\mu_0$, such that $\rho_\nu\ge C\nu^\theta$ for
$\nu\ge\nu_0$. The dependence of $\theta$ on $\delta$ would have to be
determined in the course of proving the conjecture.  One can hope for
$\theta=-\frac12-\delta$, but accounting for factors in the recurrence
associated with differentiation such as $q$ and $j$ discussed above
might require a larger shift, e.g., $\theta=-1-\delta$ or
$\theta=-\frac32-\delta$.

Limiting the rate at which $\rho_\nu$ approaches zero via $\rho_\nu\ge
C\nu^\theta$ with $C>0$, $\theta<0$ and $|\theta|$ small appears to be
critical for the convergence of the Pad\'e approximants of the Stokes
expansion.  To gain intuition, let $\beta\in\{1,2,3,4,6,8,12,16\}$ and
consider the function
\begin{equation}\label{eq:fz:beta}
  f(z) = \sum_{l=0}^\infty \frac{z^{50l}}\pi\int_{-r_l}^{r_l}
  \frac{\sqrt{r_l^2-s^2}}{a_l+s-z}\,ds, \qquad
  a_l = \frac1{1+\big[(l+1)\ln(l+e)\big]^\beta},\qquad
  r_l = 10^{-2-5l},
\end{equation}
which contains an infinite sequence of branch cuts
$[a_l-r_l,a_l+r_l]\subset\mbb R$ that shrink in size as they approach
the origin. While $f(z)$ is infinitely differentiable at $z=0$, its
Maclaurin series has a radius of convergence of zero. The truncated
continued fraction $z^{-1}\gaussk_{n=0}^N \frac{d_nz}1$ of the formal
power series $\sum_{\nu=0}^\infty c_\nu z^\nu$ with
$c_\nu=\frac1{\nu!}f^\e{\nu}(0)$ gives the $[m/k]$ Pad\'e approximant
of $f(z)$ of order $N$, where $m=\lfloor N/2\rfloor$ and $k=\lceil
N/2\rceil$. The $d_n$ are obtained from the $c_\nu$ using the
quotient-difference (qd) algorithm
\cite{cuyt,lorentzen:book}. We plan to study this example in
  more detail in future work but report our preliminary results here.
  We can show that $\rho_\nu^{-1}:=\big|c_\nu/c_{\nu-1}\big|$ climbs
  through an infinite staircase with flat plateau regions separated by
  localized jumps, similar to what we imagine will happen in
  figure~\ref{fig:growthFD}(a,b) as $\nu\to\infty$. We can also show
  that $\rho_\nu^{-1}\sim (\nu/50)^\beta$, i.e.,
  $\lim_{\nu\to\infty}\frac{1/\rho_\nu}{(\nu/50)^\beta}=1$. Thus,
  increasing $\beta$ increases the growth rate of $\rho_\nu^{-1}$.
The question is whether the Pad\'e approximants converge to $f(x)$
outside of the branch cuts, e.g., at $x=3/4$.  Through
numerical tests up to order $N=1600$ using 5000 digits of precision
and $\beta\in\{1,2,3,4,6\}$, we find that the error $\big|f(x) -
  x^{-1}\gaussk_{n=0}^N \frac{d_nx}1\big|$ at $x=3/4$ decreases
geometrically as $N$ increases, but the decay rate gets worse as
$\beta$ increases. Repeating this for $\beta=16$, the error decreases
initially but reaches a barrier that prevents the Pad\'e approximants
from converging to $f(3/4)$. The cases $\beta\in\{8,12\}$ also
  exhibit barriers, but at smaller thresholds than $\beta=16$. It is not
  clear whether these barriers obstruct convergence of the Pad\'e
  approximants or merely delay it.

The reason for the breakdown in convergence has to do with whether
poles of the Pad\'e approximant continue to be distributed to
low-index branch cuts to improve the quadrature approximation
\cite{allen:75} of the Cauchy-Stieltjes integrals in
\eqref{eq:fz:beta} as new branch cuts are encountered.  The factor
$z^{50l}$ shifts the Maclaurin series of the $l^\text{th}$ integral in
\eqref{eq:fz:beta} to higher orders of the expansion, so a new branch
cut is encountered when $N$ is a multiple of 50. The newly
encountered branch cuts draw some of the Pad\'e poles away from
lower-index branch cuts, but the overall trend when $\beta$ is small
is that increasing $N$ increases the number of poles in each branch
cut of index $l\le\lfloor N/50\rfloor$, often favoring lower-index
branch cuts and leaving the most recently encountered branch cuts
devoid of poles. By contrast, when $\beta=16$, the high-index branch
cuts rapidly acquire poles at the expense of low-index branch cuts.
When $\beta=16$, the number of poles in the $l=0$ branch cut
$[a_0-r_0,a_0+r_0]=\big[\frac{49}{100},\frac{51}{100}\big]$ decreases
from $25$ at $N=49$ to $14$ at $N=524$. It then alternates between
14 and 15 for $524\le N\le798$ and remains equal to 14 for
$798\le N\le1600$, with no sign of rebounding. The error in the
$N^\text{th}$-order Pad\'e approximant at $x=3/4$ does not improve
after $N=11$ in this case. We plan to investigate this example
  in more detail in future work.

\section{Lower bounds on small and large divisors}
\label{sec:proof32}

In this section we prove Theorem~\ref{thm:nonresonant} and perform a
numerical test to show that the bounds in the theorem are indicative
of what happens in practice. We also state and prove a theorem that as
the wave number $p$ increases, the spacing between potentially small
divisors increases, as do the size and number of large divisors near
every potentially small divisor, defined as a divisor bounded by
$\coth(\mu_0)$ in magnitude.  Most of these will not be a new
`smallest divisor seen so far,' but even satisfying
$|\lambda_{p,j}|\le\coth(\mu_0)$ becomes increasingly unlikely. Recall
from \S\ref{sec:res:depths} that
\begin{equation}
  \lambda_{p,j}= p\frac{\tanh(p\mu_0)}{\tanh\mu_0} - j^2, \qquad\quad
  \lambda_p = \min_{j\in p+2\mbb Z} \big|\lambda_{p,j}\big|,
\end{equation}
and that we write $\lambda_{p,j}(\mu_0)$ and $\lambda_p(\mu_0)$ in
contexts where multiple depths $\mu_0$ are being discussed.

{
\renewcommand{\thetheorem}{\ref{thm:nonresonant}}
\begin{theorem}
  For each $\delta>0$, the set
  \begin{equation}\label{eq:E:del:def}
    \mc E_\delta = \Big\{\mu_0>0 \;\,\Big\vert\;\, \exists \; a>0
    \;\, \text{such that} \;\, \forall\;p\ge2, \; 
    \lambda_{p}(\mu_0)\ge \min\big(a,p^{-\frac12-\delta}\big) \Big\}
  \end{equation}
  has full Lebesgue measure. For $\delta\le0$, $\mc E_\delta$ has
  Lebesgue measure 0.  For $\delta\le-\frac12$, $\mc E_\delta$ is the
  empty set.  If $\delta>\frac12$ and $\mu_0>0$ is rational, then
  $\mu_0\in\mc E_\delta$.
\end{theorem}
\addtocounter{theorem}{-1}
}

\begin{proof}
To prove the first assertion, we will show that the complement $\mc
E_\delta^c=(0,\infty)\setminus\mc E_\delta$ has measure zero. Fix
$\delta>0$ and $\mu_0\in \mc E_\delta^c$. Then either $\mu_0$ is a
resonant depth or we can construct a sequence
$\{(p_i,a_i)\}_{i=1}^\infty$ with the properties that $a_1=1$ and, for
$i\ge1$, $a_{i+1}=\lambda_{p_i}$ with $p_i$ the smallest integer
$p\ge2$ satisfying $\lambda_p<\min(a_i,p^{-\frac12-\delta})$. An
induction argument shows that for $i\ge1$, the last element of the
finite sequence $\{\lambda_2,\lambda_3,\dots,\lambda_{p_i}\}$ is the
unique smallest element, and $p_{i+1}>p_i$. For each $i\ge1$, we can
choose $j_i\ge1$ to have the same parity as $p_i$ and to satisfy
$|\lambda_{p_i,j_i}|=\lambda_{p_i}$. (The argmin of the formula for
  $\lambda_p$ with $p\ge2$ is never $j=0$.) This procedure yields a
sequence $\{(p_i,j_i)\}_{j=1}^\infty$ satisfying
\begin{equation}\label{eq:pi1:seq}
  p_{i+1}>p_i\ge2, \qquad
  0<\lambda_{p_i}<p_i^{-\frac12-\delta}, \qquad
  \lambda_{p_{i+1}}<\lambda_{p_i}, \qquad
  \big|\lambda_{p_i,j_i}\big|=\lambda_{p_i}, \qquad \big(i\ge1\big).
\end{equation}
From $\big|\lambda_{p_i,j_i}\big|<p_i^{-\frac12-\delta}$ and the
triangle inequality, we have
\begin{equation}
  \begin{aligned}
    \big| p_i - j_i^2\tanh\mu_0\big|\;\, &= \;\, \Big| \tanh(\mu_0)\lambda_{p_i,j_i}
    + p_i\big[1 - \tanh(p_i\mu_0) \big] \Big| \\  &<\;\,
  \tanh(\mu_0)\,p_i^{-\frac12-\delta} + p_i\big[ 1 - \tanh(p_i\mu_0) \big].
  \end{aligned}
\end{equation}
Since $\mu_0>0$ and $0<(1-\tanh x)\le 2e^{-2x}$ for $x>0$, there
exists $p_*$ large enough that
\begin{equation}
  p^{\frac32+\delta}\big[ 1 - \tanh(p\mu_0)
  \big] < \big(1-\tanh\mu_0\big), \qquad\quad (p\ge p_*).
\end{equation}
Choose $i_*$ large enough that $p_i\ge p_*$ for $i\ge i_*$. Then
\begin{equation}\label{eq:tanh:bound1}
  \big|p_i - j_i^2\tanh\mu_0 \big|< p_i^{-\frac12-\delta}, \qquad \big(i\ge i_*\big)
\end{equation}
and, since $p_i\ge2$,
\begin{equation}\label{eq:ji:pi:bound}
  j_i^2\tanh\mu_0 - p_i
  \;<\; p_i^{-\frac12-\delta}
  \;<\; \jt\frac12p_i \quad \Rightarrow \quad
  \frac32p_i> j_i^2\tanh\mu_0, \qquad \big(i\ge i_*\big).
\end{equation}
Equation \eqref{eq:tanh:bound1} now gives
\begin{equation}
  \big|j_i^2\tanh\mu_0 - p_i\big|< p_i^{-\frac12\delta}\left(
    {\jt\frac23} j_i^2 \tanh\mu_0 \right)^{-\frac12-\frac12\delta}
   < j_i^{-1-\delta}, \qquad
  \big(i\ge i_*\big),
\end{equation}
where we increased $i_*$ if necessary to achieve
$p_i^{-\frac12\delta}<\left( \frac23 \tanh\mu_0
  \right)^{\frac12+\frac12\delta}$ for $i\ge i_*$.  We conclude that\,
$\tanh\mu_0$\, belongs to the set
\begin{equation}\label{eq:Fdelta:def}
  \mc F_\delta = \Big\{ x\in\mbb R \;\,\Big\vert\;\, \exists \;\,
  \text{infinitely many pairs}
  \;\,(p,j)\in\mbb Z\times\mbb N \;\,\text{s.t.} \;\,
  \Big| x - \frac{p}{j^2} \Big| < \frac1{j^{3+\delta}} \Big\}.
\end{equation}
Borosh \& Fraenkel \cite{borosh} proved that the Hausdorff dimension
of $\mc F_\delta$ is $\frac{3}{3+\delta}$. Since this is less than 1,
its Lebesgue measure is zero. We have established that
\begin{equation}
  \mu_0\in \wtil{\mc F}_\delta = \tanh^{-1}\big( \mc F_\delta\cap (0,1) \big).
\end{equation}
The inverse hyperbolic tangent function is absolutely continuous and
increasing on any compact interval $[x_1,x_2]\subset(0,1)$, so
$\tanh^{-1}\big(\mc F_\delta\cap [x_1,x_2] \big)$ has measure zero by
Theorem~7.18 of \cite{rudin:cx}. It follows that $\wtil{\mc F}_\delta$
has measure zero. We conclude that $\mc E_\delta^c$ is a subset of the
union of $\wtil{\mc F}_\delta$ with the countable set of resonant
depths, and hence has measure zero.

Next fix $\delta\le-\frac12$. Then $\mc E_\delta$ coincides with the
set
\begin{equation}\label{eq:mcE:def}
  \mc E = \Big\{\mu_0>0 \;\,\Big\vert\;\, \exists \; a>0 \;\,
  \text{such that} \;\, \forall\;p\ge2, \; 
  \lambda_{p}(\mu_0)\ge a\Big\}.
\end{equation}
Both inclusions $\mc E_\delta\subset\mc E$ and $\mc E\subset\mc
E_\delta$ follow from reducing $a$ to 1 if necessary and noting that
$\min\big(a,p^{-\frac12-\delta}\big)=a$ for $p\ge2$. We claim that
$\mc E$ is the empty set. Suppose $\mc E$ is not empty and
$\mu_0\in\mc E$. Then there is an $a>0$ such that
\begin{equation}
  \big|p - j^2\tanh\mu_0\big|\;\, \ge \;\, a\tanh\mu_0 - p\,\big[1 - \tanh(p\mu_0)\big],
  \qquad\big(\,p\ge2,\;j\in p+2\mbb Z\,\big).
\end{equation}
Since $0<(1-\tanh x)\le 2e^{-2x}$ for $x>0$, there is a $p_*\ge2$ such
that
\begin{equation}\label{eq:equidistr}
  \big|p - j^2\tanh\mu_0\big|\;\,\ge\;\,  \frac a2\tanh\mu_0, \qquad\quad
  \big(\,p\ge p_*,\; j\in p+2\mbb Z\, \big).
\end{equation}
If $\tanh\mu_0=m/d$ is rational, then setting $j=2p_*d$ and
$p=4p_*^2md$ causes the left-hand side to be zero, a
contradiction. (The factors of 2 and 4 ensure that $j\in p+2\mbb Z$.)
Now suppose $\tanh\mu_0$ is irrational.  We first observe that if $j$
is large enough that $j^2\tanh\mu_0\ge p_*$, then we can round
$j^2\tanh\mu_0$ down or up to obtain an integer $p\ge p_*$ with the
same parity as $j$ and such that the left-hand side of equation
\eqref{eq:equidistr} is less than or equal to 1. This implies that
$\frac a2\tanh\mu_0\le1$.  We know from Weyl's equidistribution
theorem \cite{stein} that $m^2\tanh\mu_0$ is equidistributed on the
unit interval modulo 1, so there exists an $m$ large enough that
$4m^2\tanh\mu_0\ge(p_*+1)$ holds, and such that $m^2\tanh\mu_0$ modulo
1 lies in the interval $\big(0,\frac a8\tanh\mu_0\big)$. This implies
there is an integer $l$ such that $\big|l-m^2\tanh\mu_0\big|<\frac
a8\tanh\mu_0$.  Multiplying by 4 and setting $(p,j)=(4l,2m)$, we
obtain $p$ and $j$ of the same parity such that
$\big|p-j^2\tanh\mu_0\big|<\frac a2\tanh\mu_0$, which contradicts
\eqref{eq:equidistr} once we confirm that $p\ge p_*$. For this, we use
$p-j^2\tanh\mu_0>-\frac a2\tanh\mu_0\ge-1$, which gives
$p>4m^2\tanh\mu_0-1\ge p_*$, as required.

Now suppose $\delta\le0$.  We claim that $\mc E_\delta$ has measure
zero.  We already proved that $\mc E_\delta$ is empty for
$\delta\le-\frac12$, so suppose $-\frac12<\delta\le0$. Let
$\theta=-\big[\delta-(-\frac12)\big]=-\frac12-\delta$, which satisfies
$-\frac12\le\theta<0$.  Let $x_1$ and $x_2$ satisfy $0<x_1<x_2<1$ and
suppose $\mu_0\in\mc
E_\delta\cap\tanh^{-1}\!\big([x_1,x_2]\big)$. Then there is an $a>0$
such that
\begin{equation}
  \big|p-j^2\tanh\mu_0\big| \ge \min\big(a,p^\theta\big)\tanh\mu_0 - p\big[1-\tanh(p\mu_0)\big],
  \qquad \big(\,p\ge 2,\, j\in p+2\mbb Z\big).
\end{equation}
Since $\theta<0$, there is a $p_*$ such that
$\min\big(a,p^\theta\big)=p^\theta$ for $p\ge p_*$. We can increase
$p_*$ if necessary so that
$p^{1-\theta}\big[1-\tanh(p\mu_0)\big]<\frac12\tanh\mu_0$ for $p\ge
p_*$.  Thus,
\begin{equation}\label{eq:pi:theta:lower}
  \big| p - j^2\tanh\mu_0\big| > \frac12p^\theta\tanh\mu_0 \ge
  \frac{x_1}2p^\theta, \qquad \big(\,p\ge p_*,\, j\in p+2\mbb Z\,\big).
\end{equation}
We will use this inequality to show that $\tanh\mu_0$ belongs to a set
of measure zero.  Let $C=(2x_2)^\theta x_1$ and
$C_1=\frac142^{2\theta}C$. Since $-\frac12\le\theta<0$ and
$0<x_1<x_2<1$, we have $0<C_1<C<2^\theta x_1^{1+\theta}<1$. Schmidt
\cite{schmidt:64} proved that for almost every $x\in\mbb R$, the
number $N(M,x)$ of integers $m$ in the range $1\le m\le M$ for which
there exists an integer $l$ satisfying $|m^2x-l|\le
\frac12C_1m^{2\theta}$ satisfies $N(M,x)=\Psi(M)+O(\Psi(M)^{2/3})$ as
$M\to\infty$, where $\Psi(M)=\sum_{m=1}^M C_1m^{2\theta}$. The
intervals
$I_m=\big[{-}\frac12C_1m^{2\theta},\frac12C_1m^{2\theta}\big]$ are
required to be nested ($I_m\supset I_{m+1}$) as $m$ increases, which
is true here since $\theta<0$.  We see that
$\lim_{M\to\infty}\Psi(M)=\infty$ since $C_1>0$ and $-1\le2\theta<0$.
We conclude from Schmidt's theorem that the following set has full
Lebesgue measure
\begin{equation}
  \mc G_{\theta,C_1} = \Big\{ x\in\mbb R \;\,\Big\vert\;\, \exists \;\,
  \text{infinitely many pairs}
  \;\,(l,m)\in\mbb Z\times\mbb N \;\,\text{s.t.} \;\,
  \big| m^2x - l \big| \le \frac12C_1 m^{2\theta} \Big\}.
\end{equation}
Freezing $x_1$, $x_2$ and the corresponding $C$ and $C_1$, the set
$\mc G=\mc G_{\theta,C_1}\cap[x_1,x_2]$ has full measure $x_2-x_1$.
If $x\in\mc G$, there is a sequence
$\big\{(l_i,m_i)\big\}_{i=1}^\infty$ such that
$\big|m_i^2x-l_i\big|\le\frac12C_1m_i^{2\theta}$ and $m_{i+1}>m_i$ for
$i\ge1$. Multiplying by 4 and setting $(p_i,j_i)=(4l_i,2m_i)$, we find
that $p_i$ and $j_i$ have the same parity and
$\big|j_i^2x-p_i\big|\le2C_1m_i^{2\theta}=\frac12Cj_i^{2\theta}$.
Since $C<1$, $\theta<0$ and $j_i\ge1$,
\begin{equation}
  -\frac12\le p_i-j_i^2x\le\frac12, \qquad\quad (i\ge1).
\end{equation}
Choose $i_*$ large enough that $j_i^2x_1>\frac12$ for $i>i_*$. Then
\begin{equation}
  0\quad < \quad j_i^2x_1-\frac12 \quad\le\quad p_i \quad\le\quad j_i^2x_2+\frac12
  \quad < \quad 2j_i^2x_2, \qquad \big( i\ge i_* \big).
\end{equation}
For $i>i_*$, we have $j_i^2>p_i/(2x_2)>0$, so
\begin{equation}\label{eq:pi:theta:upper}
  \big|p_i-j_i^2x\big| \quad\le\quad\frac12C(j_i^2)^\theta \quad < \quad
  \frac12C\Big(\frac {p_i}{2x_2}\Big)^\theta \quad =\quad \frac{x_1}2p_i^\theta,
  \qquad\quad \big(i>i_*\big).
\end{equation}
Since there are infinitely many pairs $(p_i,j_i)$ with the same parity
satisfying \eqref{eq:pi:theta:upper}, $\tanh\mu_0$ satisfying
\eqref{eq:pi:theta:lower} does not belong to $\mc G$.  This shows that
\begin{equation}\label{eq:Edelta:subset}
  \mc E_\delta\cap \tanh^{-1}\big([x_1,x_2]\big) \quad\subset\quad
  \tanh^{-1}\Big( [x_1,x_2]\setminus \mc G \Big).
\end{equation}
Since $[x_1,x_2]\setminus\mc G$ has measure zero and $\tanh^{-1}(x)$
is monotonic and absolutely continuous on the compact interval
$[x_1,x_2]$, Theorem~7.18 of \cite{rudin:cx} ensures that the
right-hand side of equation \eqref{eq:Edelta:subset} has measure zero.  Since
$[x_1,x_2]$ was an arbitrary subinterval of $(0,1)$, we conclude that
$\mc E_\delta$ has measure zero.

Finally, we claim that if $\delta>\frac12$ and $\mu_0>0$ is rational,
say $\mu_0=m/d$ with $m,d\in\mbb N$, then $\mu_0\in\mc
E_\delta$. Lambert's continued fraction \cite{lorentzen:book} is
$\tanh(\mu_0)=\frac
md\raisebox{-3pt}{$+$}\frac{m^2}{3d}\raisebox{-3pt}{$+$}\frac{m^2}{5d}
\raisebox{-3pt}{$+$}\frac{m^2}{7d}
\raisebox{-3pt}{$+\cdots$}$. Theorem 4.1 of \cite{hancl:15} implies
that $\tanh\mu_0$ has irrationality exponent $2$. This means that for
any $\beta>2$, there exists a constant $C>0$ such that
$\big|\tanh\mu_0-\frac pq\big|\ge Cq^{-\beta}$ for all $(p,q)\in\mbb
N^2$. Let $\delta'=\frac12\big(\frac12+\delta\big)$, which satisfies
$\frac12<\delta'<\delta$. We set $\beta=\frac32+\delta'>2$ to obtain
$C$.  Specializing to $q=j^2$ then gives
\begin{equation}\label{eq:j2p:0}
  \big|j^2\tanh\mu_0 - p\big|\ge \frac{C}{j^{1+2\delta'}}, \qquad\quad
    \big( p,j\,\in\,\mbb N \,\big).
\end{equation}
If $\mu_0$ were not in $\mc E_\delta$, then since it is also
non-resonant by Lemma~\ref{fdTranscendentalNumberTheoryLemma}, we
could construct a sequence $\{(p_i,j_i)\}_{i=1}^\infty$ satisfying
equations \eqref{eq:pi1:seq}--\eqref{eq:ji:pi:bound}. From
equations \eqref{eq:tanh:bound1} and \eqref{eq:ji:pi:bound}, we have
\begin{equation}\label{eq:j2p:1}
  \big|j_i^2\tanh\mu_0-p_i\big| \; < \; \Big(\frac23j_i^2\tanh\mu_0\Big)^{-\frac12-\delta}
  \; \le \; \frac{C}{j_i^{1+2\delta'}}, \qquad\qquad \big(i\ge i_*\big),
\end{equation}
where we increased $i_*$ if necessary so that
$\big(\frac23\tanh\mu_0\big)^{-\frac12-\delta}j_i^{2(\delta'-\delta)}\le
C$ for $i\ge i_*$. This contradicts \eqref{eq:j2p:0}, so $\mu_0\in\mc
E_\delta$ as claimed.  \hspace*{\fill} \raisebox{-3pt}{\qed}
\end{proof}

We performed three numerical experiments, with
  $\mu_0\in\{1/16,0.25,0.2499\}$, to study the rate at which finite-depth small
  divisors approach zero in practice. We chose $\mu_0=1/16$ since
there were no jumps in $\rho_\nu^{-1}$ in figure~\ref{fig:growthFD}(a)
for that case.  Figure~\ref{fig:sd:1:16}(a) shows the first
24 pairs $(p,\lambda_p)$ for which $\lambda_p<\big(\min_{2\le q\le
    p-1}\lambda_q\big)$.  We find that $\lambda_p\ge\lambda_2$ for
$2\le p\le 24\,773$. It is difficult to imagine a code ever being
implemented that could reach $p=24\,774$. While
Theorem~\ref{thm:nonresonant} only guarantees that $\mu_0=1/16$
belongs to $\mc E_\delta$ for $\delta>1/2$ (by virtue of being
  rational), it appears to belong to $\mc E_{0.07}$, using $a=0.0155$
in the definition \eqref{eq:E:delta:def} to capture the first point
$\lambda_2$. Decreasing $\delta$ much further would require reducing
$a$. The first divisor $\lambda_{p,j}$ smaller than $10^{-5}$ in
magnitude is $p=714\,638\,949\,293$ and $j=3\,383\,653$, and the first
below $10^{-11}$ is $p\approx2.504\times10^{22}$ and
$j\approx6.333\times10^{11}$, with precise integer values given as
subscripts in the figure. We checked all possibilities with $1\le
j\le10^{12}$, which covers $2\le p\le 6.24\times10^{22}$.
We repeated this for $\mu_0=0.25$, which exhibits a large jump in
$\rho_\nu^{-1}$ in figure~\ref{fig:growthFD}(a) due to a nearby harmonic
resonant depth (at $\mu_0\approx 0.249977976$), and for $\mu_0=0.2499$, to
compare to the $\mu_0=0.25$ case. The results are shown in
figure~\ref{fig:sd:1:16}(b). In both cases, $\mu_0$ appears to
belong to $\mc E_\delta$ with $\delta=0.1$, where $a=\lambda_{6037}=0.000447$
for $\mu_0=0.25$ and $a=\lambda_{48}=0.0587$ for $\mu_0=0.2499$.  As
with $\mu_0=1/16$, new small divisors roughly follow the dashed orange
line $p^{-1/2}$ in the log-log plot, and lie above the line
$p^{-0.6}$, except for a few initial outliers that are accounted for
with the parameter $a$ in equation \eqref{eq:E:delta:def}. These three
experiments confirm that the sets $\mc E_\delta$ in
Theorem~\ref{thm:nonresonant} are well-suited to describe what
actually happens in practice.

\begin{figure}[t]
  \begin{center}
    \includegraphics[width=\linewidth]{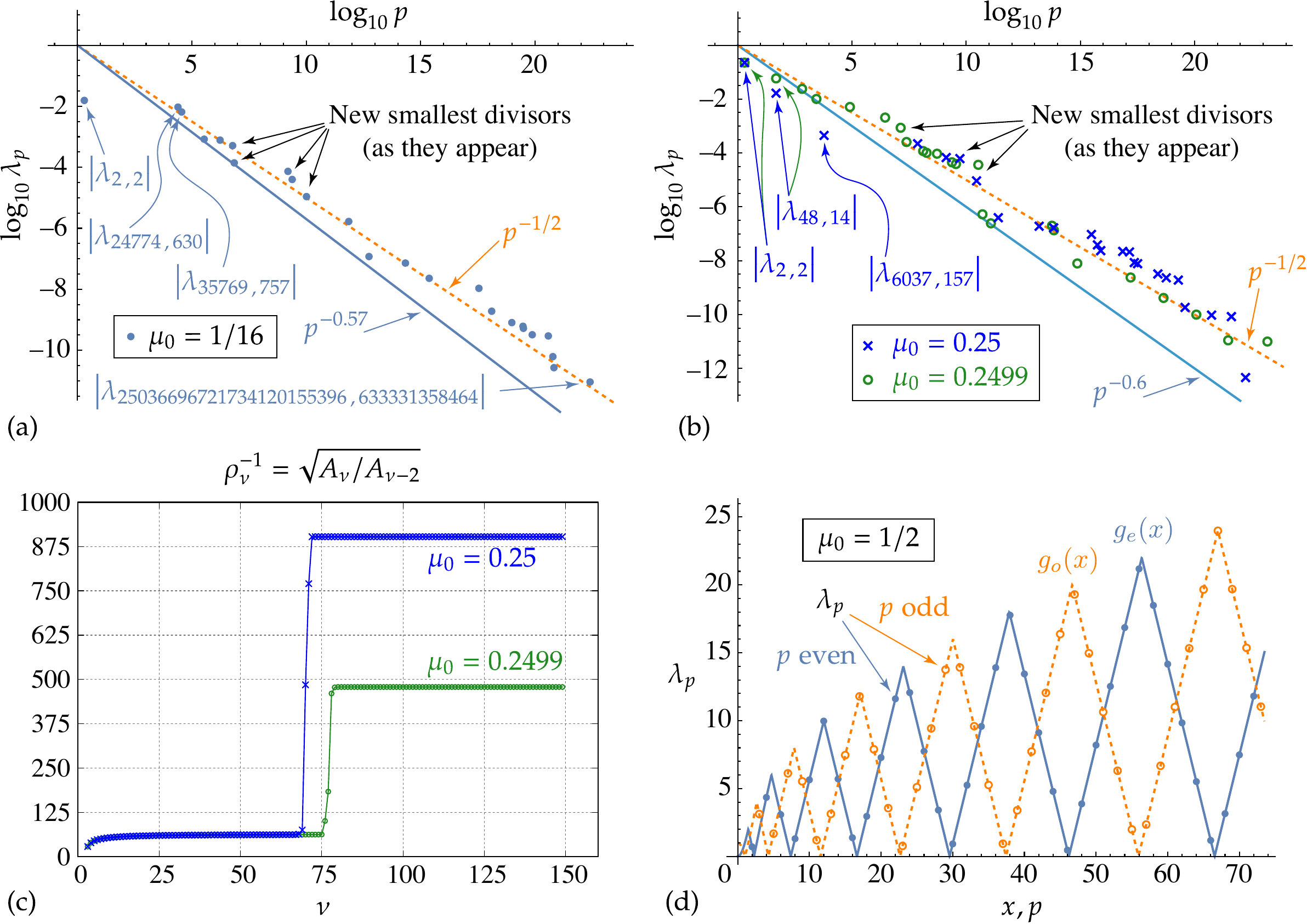}
  \end{center}
  \caption{\label{fig:sd:1:16} Behavior of $\lambda_p=|\lambda_{p,j}|$
    for $\mu_0\in\{1/16,0.25,0.2499,1/2\}$. (a,b) On the rare
    occasion that a new smallest divisor enters the recurrence, it
    does so near $p^{-1/2}$ (orange dashed line), remaining
      above the blue line $p^{-\frac12-\delta}$ once $p$ is large
      enough, and above $\min\big(a,p^{-\frac12-\delta}\big)$ for
      $p\ge2$.  Here $a=\lambda_2=0.0155$ and $\delta=0.07$ for
      $\mu_0=1/16$; $a=\lambda_{6037}=0.000447$ and $\delta=0.1$ for
      $\mu_0=0.25$; and $a=\lambda_{48}=0.0587$ and $\delta=0.1$ for
      $\mu_0=0.2499$.  This suggests that $1/16\in\mc E_{0.07}$ and
      $\{0.25,0.2499\}\subset\mc E_{0.1}$ in
      Theorem~\ref{thm:nonresonant}.  (c) The $(p,j)=(48,14)$ mode is
      nearly resonant for both $\mu_0=0.25$ and $\mu_0=0.2499$, but
      $|\lambda_{48,14}|$ is $0.282$ times smaller for $\mu_0=0.25$
      than for $\mu_0=0.2499$, which causes a larger and earlier jump
      in $\rho_\nu^{-1}$ in the former case by the mechanism described
      in \S\ref{sec:sdiv:growth}. (d) Each $\lambda_p$ lies on
    a sawtooth-shaped curve, one for $p$ even (blue lines) and the
    other for $p$ odd (orange dashed lines). }
\end{figure}

The main difference in the pattern of new smallest divisors for
$\mu_0=0.25$ versus $\mu_0=0.2499$ in figure~\ref{fig:sd:1:16}(b)
is that $\mu_0=0.25$ is nearly resonant for
$(p,j)\in\big\{(48,14),(6037,157)\big\}$, which causes
$\lambda_p=|\lambda_{p,j}|$ to lie well below the line $p^{-0.6}$ in
the plot for $p\in\{48,6037\}$. By contrast, $\mu_0=0.2499$ has a
weaker resonance at $(p,j)=(48,14)$, with $\lambda_p$ slightly below
$p^{-0.6}$, and $\lambda_{6037}=9.46$ is not small in this
case. Because the $(p,j)=(48,14)$ resonance is weaker when
$\mu_0=0.2499$ than when $\mu_0=0.25$, the jump in
$\rho_\nu^{-1}$ in figure~\ref{fig:sd:1:16}(c) is smaller
for $\mu_0=0.2499$ than for $\mu_0=0.25$, and is also delayed, since it
takes longer for the new growth pattern in the expansion coefficients
(excited by $\lambda_{48,14}$) to become the dominant mechanism for
growth, as described in \S\ref{sec:sdiv:growth} for
$\mu_0=0.25$. If it were possible to compute the series to much higher
order, we expect there would be another large jump in $\rho_\nu^{-1}$
for the $\mu_0=0.25$ case due to the small value of
$\big|\lambda_{6037,157}\big|$, but that eventually both cases
$\mu_0\in\{0.2499,0.25\}$ would have similar sequences of jumps as new
small divisors are encountered, since both of these choices of $\mu_0$
belong to $\mc E_\delta$ with $\delta=0.1$ and have similar patterns
of new smallest divisors for large $p$.

Next we discuss large divisors and the spacing between potentially
small divisors.  Figure~\ref{fig:sd:1:16}(d) shows that
$\lambda_p$ is obtained by sampling two sawtooth-shaped curves with
progressively larger `teeth.' The blue curve is
$g_e(x)=\min_{j\in2\mbb Z}|f(x)-j^2|$ and the dashed orange curve is
$g_o(x)=\min_{j\in2\mbb Z+1}|f(x)-j^2|$, where
\begin{equation}
  f(x)=x\frac{\tanh(x\mu_0)}{\tanh(\mu_0)}.
\end{equation}
Here $x\ge0$ is a continuous variable and
\begin{equation}
  \lambda_p=g_e(p), \qquad p\in\{2,4,6,\dots\}, \qquad\qquad
  \lambda_p=g_o(p), \qquad p\in\{3,5,7,\dots\}.
\end{equation}
We plotted the case of $\mu_0=1/2$ (instead of $1/16$ or $1/4$) to reduce
the oscillation frequency of $g_e(x)$ and $g_o(x)$ and show more sampled
values of $(p,\lambda_p)$ on each monotonic segment of these
functions.  Theorem S3.3 below confirms the idea that $\lambda_p$ can
only be small for $p$ near the zeros of $g_e(x)$ or $g_o(x)$, which
are close to the centers of the intervals $I_j$ that parameterize the
V-shaped troughs of these curves. If $p$ is close enough to a zero of
one of the curves that $\lambda_p\le\coth(\mu_0)$, we will show that
the next opportunity for this to happen again is near a zero of the
other curve. The spacing between these zeros grows linearly with the
index $j$, as does the height of the $j^\text{th}$ peak of $g_e(x)$ or
$g_o(x)$. Since a zero of one curve is close to the peak of the other,
there will be many large values of $\lambda_q$ for $q$ near $p$ of the
opposite parity. We will show that if a small divisor excites growth
in the Stokes coefficients of mode $(p,j)$, there will be large
divisors suppressing growth in modes $(q,k)$ with $k\ne j$ or $q\ne p$
of the form $q=p\pm l$ for $l=\{1,3,5,\dots,l_\text{max}\}$, where
$l_{\text{max}}\approx j\tanh(\mu_0)$.  As an extreme example, when
$\mu_0=1/16$, the smallest value of $\lambda_p$ that occurs in the
range $2\le p\le 6.24\times 10^{22}$ is $9.13\times 10^{-12}$, where
$p\in I_j$ with $j=633\,331\,358\,464$. The two nearest neighbors have
$\lambda_{p\pm1}\ge1.267\times10^{12}$ and each odd $q$ satisfying
$|q-p|\le3.95\times10^{10}$ yields $\lambda_q\ge 6.33\times10^{11}$.
We note that $\lambda_q\lambda_p\ge5.78$ for each of these $q$'s, so
the large divisors suppress modes $(q,k)$ more strongly than the small
divisor amplifies mode $(p,j)$.

Before stating the theorem, we need to establish notation. Let
$\{x_j^e\}_{j=0}^\infty$ and $\{x_j^o\}_{j=0}^\infty$ denote the
location of successive peaks and zeros of $g_e(x)$ and $g_o(x)$,
respectively. At $x_0^e=x_0^o=0$, $g_o(x)$ has a peak while $g_e(x)$
has a zero. Since $f(1)=1$, $g_o(x)$ has a zero at $x_1^o=1$. Let
$\mbb N_o=\{1,3,5,\dots\}$, $\mbb N_e=\{2,4,6,\dots\}$,
$I_j=(x_{j-1}^o,x_{j+1}^o]$ for $j\in\mbb N_o$ and
$I_j=(x_{j-1}^e,x_{j+1}^e]$ for $j\in\mbb N_e$. Let $*$ denote the
symbol $e$ or $o$ and let $j\in\mbb N_*$. Then $g_*(x)=|f(x)-j^2|$ for
$x\in I_j$. In our enumeration of peaks and zeros, $x_j^*$ is the zero
of $g_*(x)$ in $I_j$ and $x_{j+1}^*$ is the peak at the right endpoint
of $I_j$. So $x=x_j^*$ is the solution of $f(x)=j^2$ while
$x=x_{j+1}^*$ is the solution of $f(x)-j^2=(j+2)^2-f(x)$, i.e.,
$f(x)=[(j+1)^2+1]$.  The value of $g_*(x)$ at $x=x_{j+1}^*$ is
$2j+2$. Between a zero and peak of one curve is a peak and zero of the
other, $\{f(x^*_j),f(x^\dagger_j),f(x^\dagger_{j+1}),
f(x^*_{j+1})\}=\{j^2,j^2+1,(j+1)^2,(j+1)^2+1\}$, where
$\dagger\in\{e,o\}$ and $\dagger\ne *$. The order of the points is
then $\{x_1^o,x_1^e,x_2^e,x_2^o, x_3^o,x_3^e,x_4^e,x_4^o,\dots\}
=f^{-1}\big(\{1,2,4,5,9,10,16,17,\dots\}\big)$. Since $f(x)$ is an
increasing bijection of $[0,\infty)$ to $[0,\infty)$, so is
$f^{-1}(y)$. For $*\in\{e,o\}$ and $j\in\mbb N_*$, we define $\mc
P_j=\{p\in\mbb N_*\,:\,p\ge2\,,\,p\in I_j\}$ so that
$\lambda_p=|\lambda_{p,j}|$ for $p\in\mc P_j$.  For each $p\ge2$ there
is precisely one $j$ such that $p\in\mc P_j$.  This is because
$\cup_{j\in\mbb N_o}I_j=(0,\infty)$ and $\cup_{j\in\mbb
  N_e}I_j=(x_1^e,\infty)$. The only question is whether $x_1^e<2$,
which follows from $2<f(2)$. (We showed that $p<f(p)<p^2$ for $p\ge2$
  in \S\ref{sec:res:depths}.) Since $x_1^e>x_1^o=1$, $\mc P_j=
\mbb N_*\cap I_j$ when $j\ge2$.

\begin{lemma} If $x\mu_0>1$ then $f'(x)>\coth(\mu_0)$. \end{lemma}

\begin{proof} Use
  $\tanh(x\mu_0)>[1-\sech^2(x\mu_0)]$\, in\, $f'(x)=\frac{\tanh(x\mu_0)+
    x\mu_0\sech^2(x\mu_0)}{\tanh(\mu_0)}$. \hspace*{\fill} \qed
\end{proof}

\begin{lemma}
  Suppose $M>0$ is large enough that $0<u(1-\tanh u)<\frac13\mu_0$ for
  all real $u\ge M$. Then if $y\ge \frac M{\mu_0\tanh\mu_0}$ and
  $x=f^{-1}(y)$, there exists $\theta\in\big(0,\frac13\big)$ such
  that\, $x=y\tanh(\mu_0)+\theta$.
\end{lemma}

\begin{proof}
  We know $x=f^{-1}(y)$ exists and $x>0$. Since $\tanh(x\mu_0)<1$ and
  $x\frac{\tanh(x\mu_0)}{\tanh\mu_0}=y$, we have $x>y\tanh(\mu_0)\ge
  M/\mu_0$. So $u = x\mu_0\ge M$ and $\theta = \frac{u(1-\tanh
      u)}{\mu_0}\in\big(0,\frac13\big)$.  Finally, $y = f(x) =
  x\frac{1-(1-\tanh(x\mu_0))}{\tanh\mu_0}
  =\frac{x-\theta}{\tanh\mu_0}$. \hspace*{\fill} \qed
\end{proof}

\begin{theorem}
  Let $M=\max\big(6,\log\big(\frac{36}{\mu_0}\big)-6\big)$ and suppose
  $j\in\mbb N$ with $j\ge\big(\sqrt{M}\coth(\mu_0)+1\big)$.  Then
  there is at most one $p\in\mc P_j$ with $\lambda_p\le\coth(\mu_0)$.
  If $p\in\mc P_j$ with $\lambda_p\le\coth(\mu_0)$ and $q\ne p$ is an
  integer satisfying $|q-p|\le\big[(j-1)\tanh(\mu_0)-\frac43\big]$,
  then $q\ge2$ and $\lambda_q>\coth(\mu_0)$.  If, additionally, $q-p$
  is odd, then $\lambda_q>j$. If $k$ has the same parity as $p$ and
  $k\ne j$, then $|\lambda_{p,k}|>\frac83j$.
\end{theorem}

\begin{proof}
  First we check that $M$ satisfies the hypotheses of lemma
  S3.2. Since $0<u(1-\tanh u)= \frac{2u}{e^{2u}+1}<2ue^{-2u}$, it
  suffices to show that $2ue^{-2u}\le\frac13\mu_0$ for $u\ge M$. Since
  $u\ge M\ge6$, we have $u\le 6e^{u-6}$ and $2ue^{-2u}\le
  12e^{-(u+6)}\le 12e^{-(M+6)}\le \frac13\mu_0$.

  We observe here that if $\mu_0\ge36e^{-12}=2.21\times 10^{-4}$ then
  $M=6$, which covers typical fluid depths.  Since
  $\tanh(\mu_0)<\mu_0$, the condition on $j$ ensures that $(j-1)^2\ge
  M\coth^2(\mu_0)> \frac{M}{\mu_0\tanh\mu_0}$. Since $j$ is an integer
  and $j\ge\big(\sqrt{6}\coth(\mu_0)+1\big)> \big(\sqrt6+1\big)\approx
  3.45$, we also have $j\ge4$.

  Let $*\in\{e,o\}$ denote the parity of $j$.  Since $j\ge4\ge2$, $\mc
  P_j=\mbb N_*\cap I_j$ and the endpoints of $I_j$ satisfy
  $x_{j\pm1}^*=[(j\pm1)^2+1]$. For each $x\in I_j$, $y=f(x)\ge
  f(x_{j-1}^*)= [(j-1)^2+1]> \frac M{\mu_0\tanh\mu_0}$. By lemmas S3.2
  and S3.1, $x\mu_0=\mu_0(y\tanh\mu_0+\theta)> M>1$ and
  $f'(x)>\coth\mu_0$. By the mean value theorem, for any $p,q\in\mc
  P_j$ we have $\lambda_p+\lambda_q=|f(p)-j^2|+|f(q)-j^2|\ge
  |f(p)-f(q)| = |f'(r)(p-q)|>|p-q|\coth(\mu_0)$, where $r$ is a real
  number between $p$ and $q$. Since $p,q\in\mbb N_*$,
  $|p-q|\ge2$. Choosing $p\in\mc P_j$ to minimize $\lambda_p$ and
  assuming $\lambda_q\le\coth(\mu_0)$ with $q\ne p$ forces
  $\lambda_p+\lambda_q\le 2\coth(\mu_0)$, a contradiction.

  Suppose $p\in\mc P_j$ with $\lambda_p\le\coth(\mu_0)$.  Let
  $\dagger\in\{e,o\}$ with $\dagger\ne*$. Recall that $x_j^*$ is the
  zero of $g_*(x)$ on $I_j$ and $x_{j\pm1}^\dagger$ are the adjacent
  zeros of $g_\dagger(x)$, so
  $f(\{x_j^*,x_{j\pm1}^\dagger\})=\{j^2,(j\pm1)^2\}$.  Since
  $\lambda_p=|f(p)-j^2|\le\coth(\mu_0)$, we know $f(p)\ge
  y_1:=(j^2-\coth\mu_0)>(j-1)^2$, where we used $\coth(\mu_0)<(2j-1)$
  in the last inequality, which follows from
  $(j-1)\ge\sqrt{M}\coth(\mu_0)$.  Since
  $(j-1)^2\ge\frac{M}{\mu_0\tanh\mu_0}$, there exist
  $\theta_1,\theta^\dagger_{j-1}\in(0,\frac13)$ such that
  \begin{equation}
    x_1=f^{-1}(y_1)=j^2\tanh(\mu_0)-1+\theta_1, \qquad
    x_{j-1}^\dagger=(j-1)^2\tanh(\mu_0)+\theta^\dagger_{j-1}.
  \end{equation}
  Since $f^{-1}(y)$ is monotonic, $p\ge x_1$. Thus,
  $p-x_{j-1}^\dagger\ge x_1-x_{j-1}^\dagger >
  \big[(2j-1)\tanh(\mu_0)-\frac43\big].$ Since
  $|q-p|\le\big[(j-1)\tanh(\mu_0)-\frac43\big]$, we conclude that
  $q-x_{j-1}^\dagger>j\tanh(\mu_0)$. The mean value theorem then gives
  $f(q)-(j-1)^2=f'(r)(q-x_{j-1}^\dagger)>j$, where $r$ is between
  $x_{j-1}^\dagger$ and $q$. Similarly, $\lambda_p\le\coth(\mu_0)$
  gives $f(p)\le y_2:=(j^2+\coth\mu_0)$, so $p\le x_2=f^{-1}(y_2)$ and
  \begin{equation*}
      x_{j+1}^\dagger-p \;\ge\; x_{j+1}^\dagger-x_2 \;=\;
      \big[(j+1)^2-j^2\big]
      \tanh(\mu_0)+\theta_{j+1}^\dagger-1-\theta_2
      \;>\; (2j+1)\tanh(\mu_0)-\frac43.
  \end{equation*}
  The bound on $|q-p|$ gives $x_{j+1}^\dagger-q>(j+2)\tanh(\mu_0)$.
  Applying the mean value theorem again gives $(j+1)^2-f(q)>(j+2)$.
  We have shown that $(j-1)^2+j<f(q)<(j+1)^2-(j+2)$. Since $j\ge4$,
  $f(1)=1<13<f(q)$, so $q\ge2$. If $q\in\mbb N_\dagger$,
  $\lambda_q=g_\dagger(q)= \min\big(f(q)-(j-1)^2,(j+1)^2-f(q)\big)>
  \min(j,j+2)=j$. Otherwise, we use $(j-1)^2+1<f(q)<(j+1)^2+1$ to
  conclude that $q\in\mc P_j=\mbb N_*\cap I_j$, and therefore
  $\lambda_q>\coth(\mu_0)$.

  Finally, if $k\in\mbb N_*$ and $k\ne j$, then
  $\lambda_{p,k}=|f(p)-k^2|\ge |j^2-k^2|-|f(p)-j^2|$. The first term
  is minimized by $k=j-2$, and $|f(p)-j^2|\le\coth(\mu_0)\le
  M^{-1/2}(j-1)\le 6^{-1/2}(j-1)$. Using $j\ge4$ and $M\ge6$, we have
  $\lambda_{p,k}\ge (4j-4) -
  6^{-1/2}(j-1)\ge(4-6^{-1/2})\frac34j>\frac83j$.
   \hspace*{\fill} \raisebox{-3pt}{\qed}
\end{proof}

\section{Derivation of the ODEs governing the Stokes coefficients}
\label{FDSecODEStokesCoefficients}

In this section we briefly derive the equations of motion for the
Stokes expansion coefficients \eqref{Stokes} from the governing
equations \eqref{CRAnsatzed}--\eqref{BernoulliAnsatzed} of the spatial
Fourier modes.  Using \eqref{Stokes} to expand the algebraic equation
(\ref{CRAnsatzed}) in powers of $\epsilon$, we obtain
\begin{align}
  \label{CRGreekT2}
  & \beta_{p,n} + p\gamma_{p,n} +  T^2_{p,n} = 0, \qquad\qquad
  \Big( p\in\mbb N \;,\; n\in \mbb N\cup\{0\} \Big), \\
  \notag
  & T^2_{p,n} = \cosh(p\mu_0) \Bigg(\sum_{q=1}^{p-1} 
    \sum_{k=0}^{n} 
    \frac{q}{2\cosh(q\mu_0)\cosh[(p-q)\mu_0]} \alpha_{q,k} \beta_{p-q,n-k}  
    \\
    \label{eq:T2:def}
    & \qquad\qquad + 
    \sum_{q=1}^{n}\sum_{k=0}^{n-q} 
    \frac{q}{2\cosh(q\mu_0)\cosh[(p+q)\mu_0]}  \alpha_{q,k} 
    \beta_{p+q,n-q-k}
    \\ \notag
    & \qquad\qquad - \sum_{q=1}^{n}\sum_{k=0}^{n-q} 
    \frac{p+q}{2\cosh(q\mu_0)\cosh[(p+q)\mu_0]} \alpha_{p+q,k} 
    \beta_{q,n-q-k}  
    \Bigg).
\end{align}
We match the notation $T^2_{p,n}$ introduced by Amick \& Toland
\cite{amick1987semi} for the analogous forcing term in the
infinite-depth case. To avoid listing special cases, empty sums are
always taken to mean zero.  Equation (\ref{fdKFSp0}) gives
differential equations for the $\dot\mu_n$, namely
\begin{equation}\label{fdKinematicGreekP0}
  \begin{split}
    & \dot{\mu}_{n} + T^1_{0,n} = 0, \qquad\qquad \Big( n\in\mbb N\cup\{0\} \Big), \\
    & T^1_{0,n} =  \sum_{q=1}^{n} 
    \sum_{k=0}^{n-q}\sum_{l=0}^{n-q-k} \frac{q}{2\cosh^2(q \mu_0)}  
    \alpha_{q,k} 
    \dot{\alpha}_{q,l} s_{2q,n-q-k-l} \\  &  \qquad\qquad + \sum_{q=1}^{n} 
    \sum_{k=0}^{n-q}\sum_{l=0}^{n-q-k} \sum_{m=0}^{n-q-k-l} 
    \frac{q^2}{2\cosh^2(q\mu_0)} \alpha_{q,k} \alpha_{q,l} \dot{\mu}_{m}
    c_{2q,n-q-k-l-m}.
  \end{split}
\end{equation}
There is no analogous forcing term $T^1_{0,n}$ in the infinite-depth
case, but Amick and Toland only defined $T^r_{p,n}$ for
$r\in\{2,3,4\}$, so we make use of the omitted $r=1$ index.  We note
that
\begin{equation}
  T^1_{0,0}=0 \qquad \Rightarrow \qquad \dot\mu_0=0,
\end{equation}
consistent with $\mu_0 L/2\pi$ being the depth of the bottom boundary
in physical space, which remains stationary as the standing wave
evolves in time.  Finally, (\ref{fdKFSp1}) gives the differential
equation
\begin{align}
  \label{fdKinematicGreekP1}
  \dot{\alpha}_{p,n} &- p\gamma_{p,n} + T^3_{p,n} = 0, \qquad\qquad \Big(\,
    p \in \mbb N\;,\; n\in\mbb N\cup\{0\} \, \Big), \\
  \label{kinematicGreekT3}
  \begin{split}
          T^3_{p,n} &=  \frac{1}{ s_{p,0}} \vast[ 
	\sum_{q=0}^{n-1} 
	\lp \dot{\alpha}_{p,q} - p\gamma_{p,q} \rp s_{p,n-q}
	+ 2 p \sum_{q=0}^{n-1} 
	\sum_{k=1}^{n-q}
	\alpha_{p,q} \dot{\mu}_k  c_{p,n-q-k} \\ & \quad\; + \cosh(p\mu_0) 
	\vast( - \sum_{q=1}^{p-1} \sum_{k=0}^{n} \sum_{l=0}^{n-k} 
	  \frac{(p-q)\alpha_{p-q,k}\dot{\alpha}_{q,l}
            s_{p-2q,n-k-l}}{2\cosh[(p-q)\mu_0]\cosh(q\mu_0)}
	\\ &  \qquad +  \sum_{q=1}^{n} 
	\sum_{k=0}^{n-q}\sum_{l=0}^{n-q-k} 
	\frac{ \big[\lp p+q\rp 
	    \alpha_{p+q,k}\dot{\alpha}_{q,l} +
            q\alpha_{q,k}\dot{\alpha}_{p+q,l}\big]
		s_{p+2q,n-q-k-l}}{2\cosh(q\mu_0)\cosh((p+q)\mu_0)}
	\\ &  \qquad + \sum_{q=1}^{p-1} \sum_{k=0}^{n-1} \sum_{l=0}^{n-k-1} 
	\sum_{m=1}^{n-k-l} 
	\frac{(p-q)q\alpha_{p-q,k}\alpha_{q,l} \dot{\mu}_m 
		c_{p-2q,n-k-l-m}}{2\cosh[(p-q)\mu_0]\cosh(q\mu_0)}
	\\ &  \qquad + \sum_{q=1}^{n-1} 
	\sum_{k=0}^{n-q-1}\sum_{l=0}^{n-q-k-1} \sum_{m=1}^{n-q-k-l} 
	\frac{ 
		q(p+q)\alpha_{q,k}\alpha_{p+q,l} \dot{\mu}_m 
		c_{p+2q,n-q-k-l-m}}{\cosh(q\mu_0)\cosh[(p+q)\mu_0]}
	\vast)\vast],
	\end{split}
\end{align}
while (\ref{BernoulliAnsatzed}) gives
\begin{align}
  \label{fdBernoulliGreekP1}
  \dot{\gamma}_{p,n} &+ \sigma_0 \tanh(p\mu_0) \alpha_{p,n} + T^4_{p,n} = 0, \qquad
  \Big(\,p\in\mbb N\;,\; n\in\mbb N\cup\{0\}\,\Big), \\
  \label{fdBernoulliGreekT4}
  \begin{split} T^4_{p,n} &= \frac{1}{c_{p,0}} \vast[
      \sum_{q=0}^{n-1} 
      \dot{\gamma}_{p,q}c_{p,n-q} 
      +
      \sum_{q=0}^{n-1} 
      \sum_{k=0}^{n-q} \alpha_{p,q} \sigma_k
      s_{p,n-q-k}
      \\ & \quad + \cosh(p\mu_0)
      \vast(
	\sum_{q=1}^{n} 
	\sum_{k=0}^{n-q}\sum_{l=0}^{n-q-k}  
	\frac{\beta_{q,k}\beta_{p+q,l}c_{p+2q,n-q-k-l}}{2\cosh(q\mu_0)\cosh[(p+q)\mu_0]}
	\\
	&  \qquad - \sum_{q=1}^{p-1} \sum_{k=0}^{n} \sum_{l=0}^{n-k} 
	\frac{\beta_{p-q,k} \beta_{q,l} 
	  c_{p-2q,n-k-l}}{4\cosh[(p-q)\mu_0]\cosh(q\mu_0)}
	\\
	&  \qquad -\sum_{q=1}^{n} 
	\sum_{k=0}^{n-q} \sum_{l=0}^{n-q-k} \frac{\lp \beta_{p+q,k} 
	  \dot{\alpha}_{q,l}+\beta_{q,k}\dot{\alpha}_{p+q,l} \rp c_{p,n-q-k-l}
	}{2\cosh(q\mu_0)\cosh[(p+q)\mu_0]}
	\\ &   \qquad +  \sum_{q=1}^{p-1} \sum_{k=0}^{n} \sum_{l=0}^{n-k} 
	\frac{\beta_{p-q,k} 
	  \dot{\alpha}_{q,l} c_{p,n-k-l}
	}{2\cosh[(p-q)\mu_0]\cosh(q\mu_0)} 
	\vast) \vast].
  \end{split}
\end{align}

\section{Computational aspects}\label{sec:comp:aspects}

In this section we examine the practical aspects of computing the
Stokes coefficients in Fourier space efficiently on a parallel
computer.  We represent the functions $\mu_n(t)$, $\alpha_{p,n}(t)$,
$\beta_{p,n}(t)$, $\gamma_{p,n}(t)$ and $B_{p,n}(t)$ through the real
coefficients ($\mu_{n,j}$, $\alpha_{p,n,j}$, etc.)  that appear in the
trigonometric polynomial representations \eqref{fdFourierSeriesMu},
\eqref{fdFourierSeriesAlpha}, \eqref{fdFourierSeriesBeta},
\eqref{fdFourierSeriesGamma} and \eqref{eq:Bqm:form}.  We reduce
memory costs by only storing the Fourier modes that are present in the
primed sums in those equations.  We also do not store $c_{q,n,j}$ or
$s_{q,n,j}$ in equation \eqref{eq:cqm:sqm:form} since they equal
$\frac12(B_{q,m,j}\pm B_{-q,m,j})$, due to \eqref{eq:cqn:sqn:bell}.

In the algorithm of \S\ref{sec:recursive:alg}, summarized in
figure~\ref{FDAlgorithmImage}, from the point that $\mu_N(t)$ has just
been computed in the previous iteration to the point that
$T^1_{0,N+1}(t)$ is evaluated in order to compute $\mu_{N+1}(t)$, the
coefficients $B_{q,m,j}$ that will be needed by any of the
$c_{q,m}(t)$ and $s_{q,m}(t)$ that appear in the formulas for the
$T^r_{p,n}$ satisfy $0\le m\le N$ and $|q|+2m\le 2N+3$.  Thus,
immediately after $\mu_N(t)$ becomes known, we compute the new Bell
polynomials $B_{q,m}(t)$ and $B_{-q,m}(t)$ with $(q,m)\in\{(1,N)\}\cup
L_{2N+2}^\circ\cup L_{2N+3}^\circ$.

It is clear that the time complexity of the recursive algorithm is
dominated by the computation of the forces $T^2_{p,n}$, $T^3_{p,n}$,
and $T^4_{p,n}$ for $(p,n)\in L_\nu^\circ$ with $\nu\in\{2N+2,2N+3\}$.
Unlike the infinite-depth case in \cite{amick1987semi}, our forces
$T^3_{p,n}$ and $T^4_{p,n}$ are no longer quadratic functions of
previously computed quantities ($\alpha_{j,k}$, $\dot\alpha_{j,k}$,
  $\beta_{j,k}$, etc.), but are now quartic and cubic, respectively.
This is because of the conformal depth function $h(t)$ and the
hyperbolic trigonometric functions it introduces. It may be possible
to introduce additional auxiliary variables to accumulate intermediate
pairwise products to reduce this complexity. We did not pursue this
idea for the finite-depth case but succeeded with this strategy for
the infinite-depth case with or without surface tension. These results
will be reported elsewhere \cite{abassi:semi2}.

Although the triple and quadruple sums reduce the maximum order
$\nu_\text{max}$ that is feasible with available computational
resources relative to the infinite-depth problem, we were able to
compute the solution to very high order ($\nu_\text{max}=149$) by
designing our code to run on a supercomputer using a hybrid MPI/OpenMP
framework \cite{chopp:book} using MPFR \cite{mpfr:toms} for
multiple-precision arithmetic. Every sum appearing in the forces,
regardless of the number of indices, is computed in parallel using MPI
and OpenMP reductions. Each thread of each MPI task accumulates a
partial sum of the terms it is responsible for. For example, the sum
\begin{equation}\label{eq:a:mu:B:term}
\sum_{q=1}^{p-1} \sum_{k=0}^{n-1} \sum_{l=0}^{n-k-1}
	\sum_{m=1}^{n-k-l} 
	\frac{(p-q)q\alpha_{p-q,k}\alpha_{q,l} \dot{\mu}_m 
	  c_{p-2q,n-k-l-m}}{2\cosh[(p-q)\mu_0]\cosh(q\mu_0)}
\end{equation}
appears in the formula \eqref{kinematicGreekT3} for $T^3_{p,n}$. When
a thread processes one of the terms of this sum, it computes the
inverse FFTs of the temporal Fourier coefficients of each
factor, namely
\begin{equation}\label{eq:a:mu:B}
  \alpha_{p-q,k,j}, \quad 
  \alpha_{q,l,j}, \quad 
  ij\mu_{m,j}, \quad 
  B_{p-2q,n-k-l-m,j}, \quad 
  B_{2q-p,n-k-l-m,j},
\end{equation}
to obtain values for $\alpha_{p-q,k}(t)$, $\alpha_{q,l}(t)$,
$\dot\mu_m(t)$ and $B_{\pm(p-2q),n-k-l-m}(t)$ for $t\in[0,2\pi)$ on a
uniform grid $\mc G_M = \{2\pi j/M\}_{j=0}^{M-1}$, with enough grid
points $M$ to resolve $T^3_{p,n}(t)$ with no aliasing errors.  Since
$T^3_{p,n}(t)$ is a trigonometric polynomial of degree $p+2n$, the
minimum grid size is $M_\text{min} = 2(p+2n+1)$. We choose the
smallest integer $M\ge M_\text{min}$ of the form
$M=2^{m_2}3^{m_3}5^{m_5}$ with $m_2\ge1$, $m_3\in\{0,1\}$ and
$m_5\in\{0,1\}$, which are grids for which the FFT and inverse FFT are
particularly fast. The value of $M$ increases as the computation
progresses to higher orders $\nu=p+2n$. Examples include $M=240$ for
$\nu=109$ and $M=320$ for $\nu=149$.  We wrote a custom FFT library to
work efficiently with the MPFR data type to avoid allocation of
temporary variables as much as possible; otherwise it is a standard
radix-2, 3 and 5 FFT algorithm, optimized as in \cite{fft:2345}. We
also wrote specialized MPI communication routines to send sequences of
MPFR numbers using character strings for the mantissas (exported in
  base 32) and integers for the exponents.

Continuing with the example in \eqref{eq:a:mu:B:term}, the Fourier
modes in equation \eqref{eq:a:mu:B} are written into complex arrays of size
$M/2+1$, indexed by $0\le j\le M/2$. Each set of modes in equation
\eqref{eq:a:mu:B} fits in this array size without truncation, and is
zero-padded to fill up the space. Multiplying $\mu_{m,j}$ by $ij$
gives the Fourier coefficients of $\dot\mu_m(t)$. We use the c2r
version \cite{recipes} of the inverse FFT, which assumes
negative-index Fourier modes are the complex conjugate of
positive-index modes (without storing them) and returns real function
values on the uniform grid $\mc G_M$. We then evaluate
$c_{p-2q,n-k-l-m}$ from $B_{\pm(p-2q),n-k-l-m}$ on $\mc G_M$. All the
factors in \eqref{eq:a:mu:B:term} are now known on the uniform grid,
and are multiplied together pointwise. Each thread of each MPI task is
assigned a subset of the indices $q$, $k$, $l$ and $m$ in the sum
\eqref{eq:a:mu:B:term} and accumulates the partial sum over these
indices.  This is repeated for the other sums in the formula
\eqref{kinematicGreekT3} for $T^3_{p,n}(t)$.  These results are
combined with those of the other threads and nodes at the end via
parallel reduction.  Finally, a forward FFT is taken to convert from
physical space back to Fourier space, where the solution of the ODEs
for $\alpha_{p,n}(t)$, etc., is `read off' from the Fourier
representations of the forces. Computing the time derivative $\dot
T^3_{p,n}$ in \eqref{eq:II:star} is also easily performed in Fourier
space.  We compute the Bell polynomials through a similar procedure in
which $B_{q,n}(t)$ is accumulated on a uniform grid in time via the
recursion \eqref{eq:Bqn:recur}. Taking the FFT of the sum gives the
Fourier coefficients $B_{q,n,j}$, which are the representation stored
in memory.

Although it would be possible to process all the lattice points within
$L_{2N+2}^\circ$ independently in parallel, followed by all the points
in $L_{2N+3}^\circ$, we elected to process the lattice points
sequentially and parallelize the computation at the level of
individual sums in the forces. This is simpler and leads to
near-perfect load balancing without having to worry about how the
number of terms in the sums in the forcing terms $T^r_{p,n}$ varies
with $p$ and $n$ at a given level $p+2n=\nu$.

\section{Effects of finite-precision arithmetic}
\label{sec:floating:point}

We computed the expansion coefficients $\alpha_{p,n,j}$,
$\beta_{p,n,j}$, $\gamma_{p,n,j}$, $\mu_{n,j}$ and $\sigma_n$ for the
dimensionless fluid depths listed in \eqref{eq:mu0:list}.  Our code
employs MPFR with a fixed mantissa size, so running the calculation
multiple times with different precisions allows us to observe the
accumulation of roundoff errors in the lower-precision results.
Figure~\ref{fig:sigErr}(a) shows the relative error in $\sigma_n$
in a 64-digit (212-bit) calculation at each depth $\mu_0$ using a
90-digit (300-bit) calculation for the reference solution. We use
$\sigma_n$ since it is a scalar quantity that is influenced by
roundoff errors in all the other coefficients up to order
$\nu=2n$. Other measures of error, such as the relative error in the
vector $\vec\alpha^\e\nu$ containing the $\alpha_{p,n,j}$ of order
$p+2n=\nu$ with $j\in E_\nu$, lead to similar results.

\begin{figure}[t]
  \begin{center}
    \includegraphics[width=\linewidth]{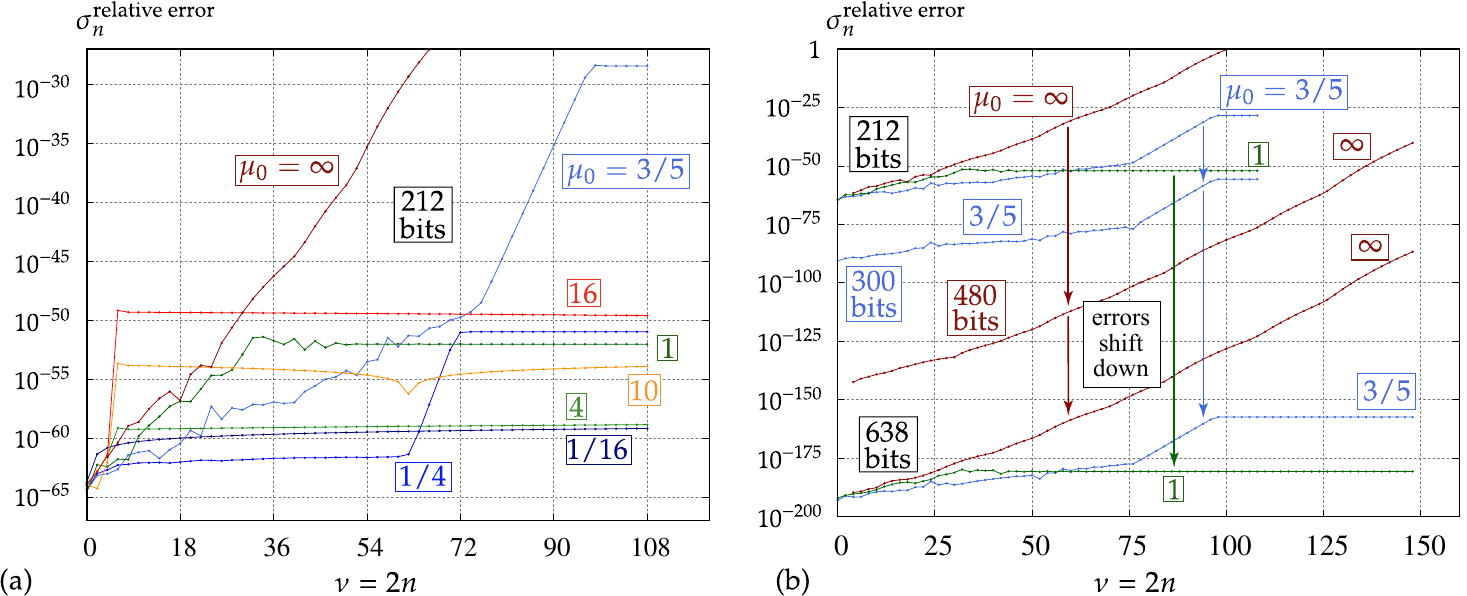}
  \end{center}
  \caption{\label{fig:sigErr} Relative errors in the Stokes expansion
    coefficient $\sigma_n$ due to floating-point errors in the
    algorithm of \S\ref{sec:recursive:alg}. Errors are estimated
    by repeating the calculation with more precision. Unlabeled
    numbers in boxes give the fluid depth, $\mu_0$, of each curve.
    The vertical shifts in panel (b) are close to
    $2^{212-480}=2.1\times10^{-81}$, $2^{480-638}=2.7\times10^{-48}$,
    $2^{212-638}=5.8\times10^{-129}$, $2^{212-300}=3.2\times10^{-27}$,
    and $2^{300-638}=1.8\times10^{-102}$.}
\end{figure}

In all seven finite-depth cases, we find that the relative error
exhibits two types of behavior, one where it saturates to a
steady-state value over several iterations, and one where it grows
until it reaches another plateau level. In the plateau regions, the
absolute error grows at a similar rate to $\sigma_n$ itself (so the
  relative error remains flat), while in the growth regions, the
absolute error grows faster than $\sigma_n$. This growth could be
partly due to an increasing amount of cancellation in the formulas for
the forcing terms $T^r_{pn}$ at higher order, where large terms of
similar size and opposite sign are added together. Additionally, the
recursion may cause these roundoff-error perturbations to grow at a
faster rate than the solution itself, e.g., through a similar process
to losing digits when computing minimal solutions of three-term scalar
recurrence relations. In the infinite-depth case, the relative error
grows steadily without entering any plateau regions.

Figure~\ref{fig:sigErr}(b) shows the relative error in $\sigma_n$ for
fluid depths $\mu_0\in\{3/5,1,\infty\}$ computed to order $\nu=149$
using 64 digits (212 bits) and 192 digits (638 bits). Also plotted are
a 90-digit (300-bit) calculation for $\mu_0=3/5$ and a 144-digit
(480-bit) calculation for $\mu_0=\infty$. We used a 256-digit (850-bit)
calculation for the reference solution when computing errors
in the 192-digit cases. For both finite and infinite depth, we find that
increasing the precision by $b$ bits causes the error curves to shift
down by a factor of approximately $2^{-b}$ while retaining their shape
(aside from small fluctuations). One could potentially use the
lower-precision calculation to estimate the error in the
higher-precision result by assuming that nearly identical growth and
plateau regions will occur. However, all errors reported in this paper
are from a lower-precision calculation checked against an auxiliary
higher-precision calculation.

To compute the continued fraction expansion coefficients in equations
\eqref{eq:T:expand} and \eqref{eq:hat:eta:phi:cfrac} for $\mu_0=1$, we
use both the standard and progressive forms of the qd-algorithm
\cite{cuyt} in 192-digit (638-bit) floating-point arithmetic and
compare the results to each other to estimate the accuracy of $d_n$
and $\tilde d_{19,n}$. The relative error between the two calculations
is zero for $d_0$ and $\tilde d_{19,0}$ and grows from $10^{-192}$ for
$d_1$ to $10^{-137}$ for $d_{74}$, and from $10^{-192}$ for $\tilde
d_{19,1}$ to $10^{-135}$ for $\tilde d_{19,65}$.  This observed loss
of precision in the continued fraction coefficients is consistent with
the condition numbers one encounters (namely $1.6\times10^{54}$ for
  $T$ and $7.7\times10^{55}$ for $\hat\varphi_{19}$) if one solves for
the polynomial coefficients of $P(x)$ and $Q(x)$ in equation
\eqref{eq:pade:def} directly from $\tau_0,\dots,\tau_{74}$ or
$\tilde\tau_{19,0},\dots,\tilde\tau_{19,65}$ by computing the
nullspace of a Toeplitz matrix \cite{gonnet:13}.  The errors in
$\tau_n$ and $\tilde\tau_{19,n}$ from computing the Stokes expansions
in finite-precision arithmetic will also affect the accuracy of the
continued fraction expansions. We repeated the entire calculation with
256 digits (850 bits) and find that the relative errors in $d_n$ and
$\tilde d_{19,n}$ for the 192-digit calculation are uniformly less that
$10^{-119}$, which is far smaller than the errors in the shooting
method.

For the results of figure~\ref{fig:imaginary:poles} in the
$\mu_0=1/16$ case, we computed the Stokes expansion and its Pad\'e
approximants twice, once with 64 digits (212 bits) and once with 128
digits (424 bits). Using the latter calculation to measure error in the
former shows that the maximum relative error in any pole or zero in
figure~\ref{fig:imaginary:poles} is bounded by
$9.6\times10^{-29}$. Thus, the 64-digit calculation has enough accuracy
to distinguish the pole $z_p$ from the zero $z_0$ in the Froissart
doublet labeled $FD$ in figure~\ref{fig:imaginary:poles}(b),
which differ from each other by $|z_p-z_0|/|z_p| = 4.4\times10^{-21}$.

\section{Secondary standing waves and the nonlinear
  deformation of resonant modes}
\label{sec:secondary:waves}

In this section we investigate the secondary standing waves that
oscillate on top of the primary wave with different amplitudes and
phases on different bifurcation branches. Such secondary waves have
been reported previously for standing waves in finite depth
\cite{mercer:94,smith:roberts:99,okamura:99,water2}, three-dimensional
standing waves \cite{rycroft:13}, and gravity-capillary standing waves
\cite{water2,shelton:stand}. Here we explore the effects of
nonlinearity on the shapes of the secondary waves, which deviate from
the sinusoidal patterns one would get from linearization about the
flat rest state.

\begin{figure}[b]
  \begin{center}
    \includegraphics[width=\linewidth]{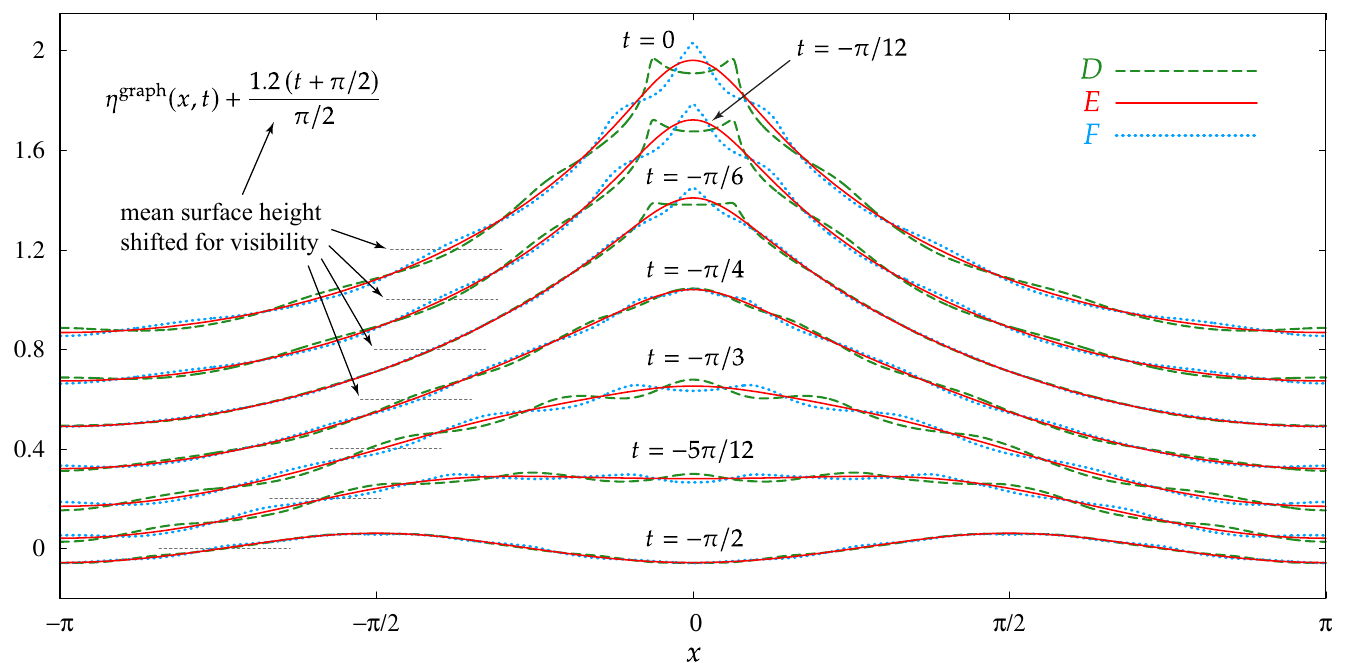}
  \end{center}
  \caption{\label{fig:evol1defBig} Snapshots of the time-evolution of
    the unit-depth standing wave solutions $D$, $E$ and $F$ from the
    bifurcation plots in
    figures~\ref{fig:bif:evol:1}--\ref{fig:bif:19} at the times
    $t\in\mathscr{T}_6$. Vertical offsets were added to the wave
    profiles at successive times for visibility.}
\end{figure}

Figure~\ref{fig:evol1defBig} shows snapshots of the wave profile
$\eta^\text{graph}(x,t)$ for solutions $D$, $E$ and $F$ from the
$\mu_0=1$ bifurcation plots in
figures~\ref{fig:bif:evol:1}--\ref{fig:bif:19} at the dimensionless
times $t\in\mathscr{T}_6$ from equation \eqref{eq:scrT:def}. These three
solutions have a common period, $T=7.267295$, which is 9.4\% larger
than small-amplitude waves in the linear regime at this depth; see
figure~\ref{fig:bif:T1}. Just like solutions $ABC$ at depth
$\mu_0=3/5$ in figure~\ref{fig:bif:evol:06}, the non-uniqueness of
solutions with this period is due to three possible amplitudes of a
secondary standing wave with characteristics of a nearby harmonic
resonance that evolves on top of the primary wave.  Solutions $ABC$
are near the $(5,3)$ resonant depth ($0.6232354$) while solutions
$DEF$ are near the $(7,3)$ resonant depth ($1.039719$). We define the
primary wave to be solution $E$. The secondary wave of solution $F$ is
in phase with solution $E$, which sharpens the crest at $t=0$ and
increases the crest-to-trough height, $\epsilon$, relative to solution
$E$. For solution $D$, the secondary wave is out of phase with
solution $E$, causing a dimple to form at the wave crest at $t=0$ and
decreasing $\epsilon$.  These changes in $\epsilon$ are also evident
in the bifurcation plot of figure~\ref{fig:bif:T1}(a). In
figure~\ref{fig:evol1defBig}, solutions $D$ and $F$ oscillate around
solution $E$ with seven spatial oscillations that deviate visibly
from being sinusoidal perturbations. The oscillations are largest near
the wave crest at $x=0$ when the wave reaches maximum amplitude at $t=0$.

\begin{figure}[t]
  \begin{center}
    \includegraphics[width=\linewidth]{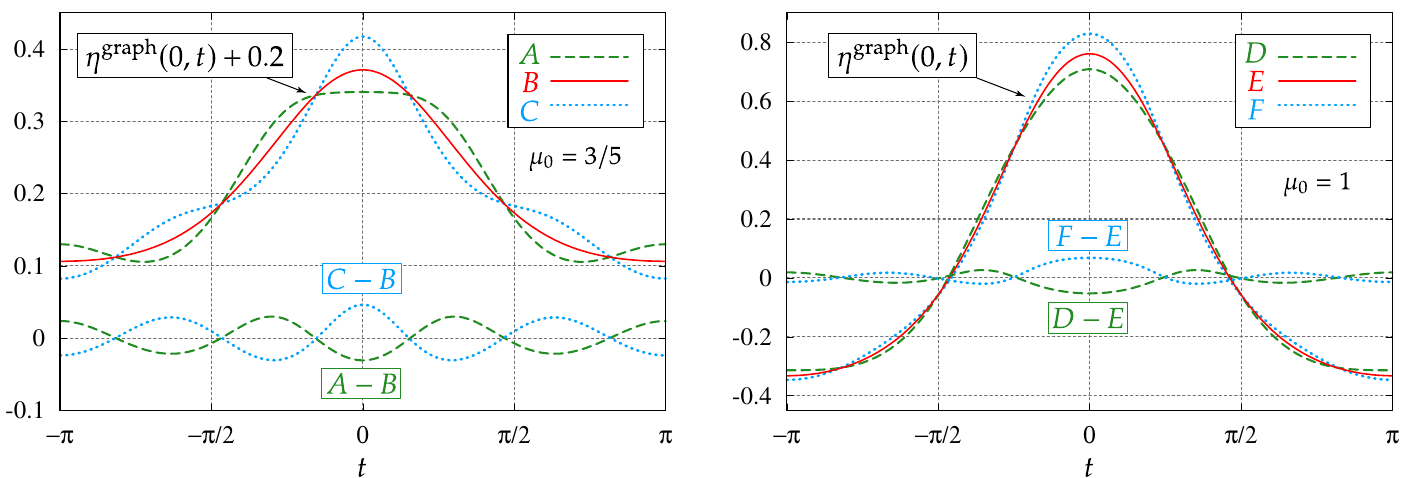}
  \end{center}
  \caption{\label{fig:x0evol1and06} Time-evolution of the wave profile
    above the symmetry point $x=0$ over one period for the $\mu_0=3/5$
    standing waves $ABC$ from figure~\ref{fig:bif:evol:06} (left) and
    the $\mu_0=1$ standing waves $DEF$ from
    figures~\ref{fig:bif:evol:1}--\ref{fig:bif:19}
    and~\ref{fig:evol1defBig} (right).  The curve labeled $A-B$ shows
    the difference
    $\eta^\text{graph}_A(0,t)-\eta^\text{graph}_B(0,t)$, with similar
    formulas for the other cases.}
\end{figure}

In figure~\ref{fig:x0evol1and06}, we plot the time evolution of the
wave profile at $x=0$ over one cycle of the wave for solutions $ABC$
from figure~\ref{fig:bif:evol:06} (left) and solutions $DEF$ from
figures~\ref{fig:bif:evol:1}--\ref{fig:bif:19}
and~\ref{fig:evol1defBig} (right).  We also plot the perturbations
required to move from $B$ to $A$ and $B$ to $C$ (left) and from $E$ to
$D$ and $E$ to $F$ (right). These are specific perturbations from one
standing wave solution of the fully nonlinear water wave equations to
another, viewing solutions $B$ and $E$ as the primary waves and these
perturbations as the secondary waves. We have not investigated the
stability of solutions $B$ and $E$ to arbitrary perturbations
\cite{ioualalen:96,water:stab1}. In the left panel, a
vertical offset of $0.2$ was added to the wave profiles to separate
them from the perturbation plots.  In both panels, the secondary waves
execute three cycles over one period of the composite wave. They
deviate visibly from being sinusoidal perturbations, with non-uniform
oscillations that grow largest near $t=0$.  This is especially true in
the right panel due to nonlinear effects being stronger for
larger-amplitude waves.

\begin{figure}[p]
  \begin{center}
    \includegraphics[width=.89\linewidth]{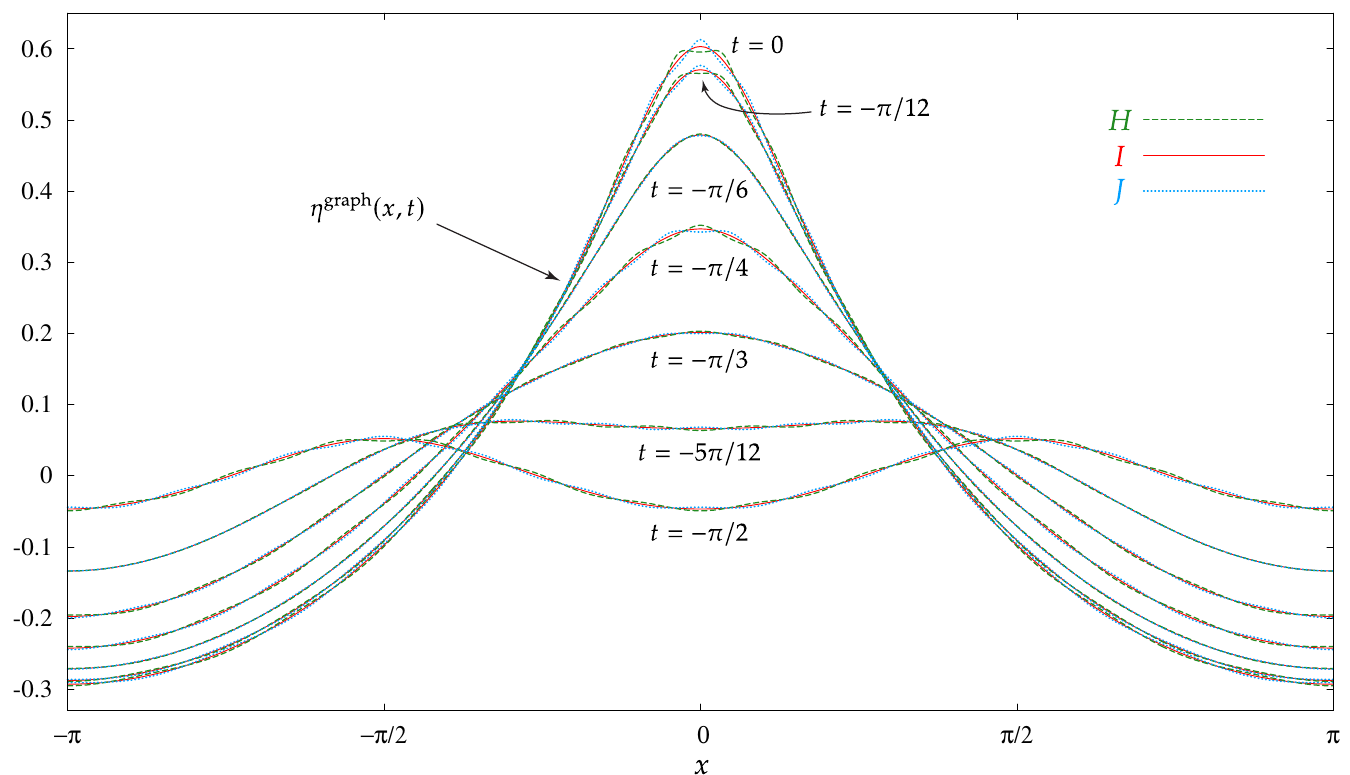}\\[5pt]
    \includegraphics[width=.89\linewidth]{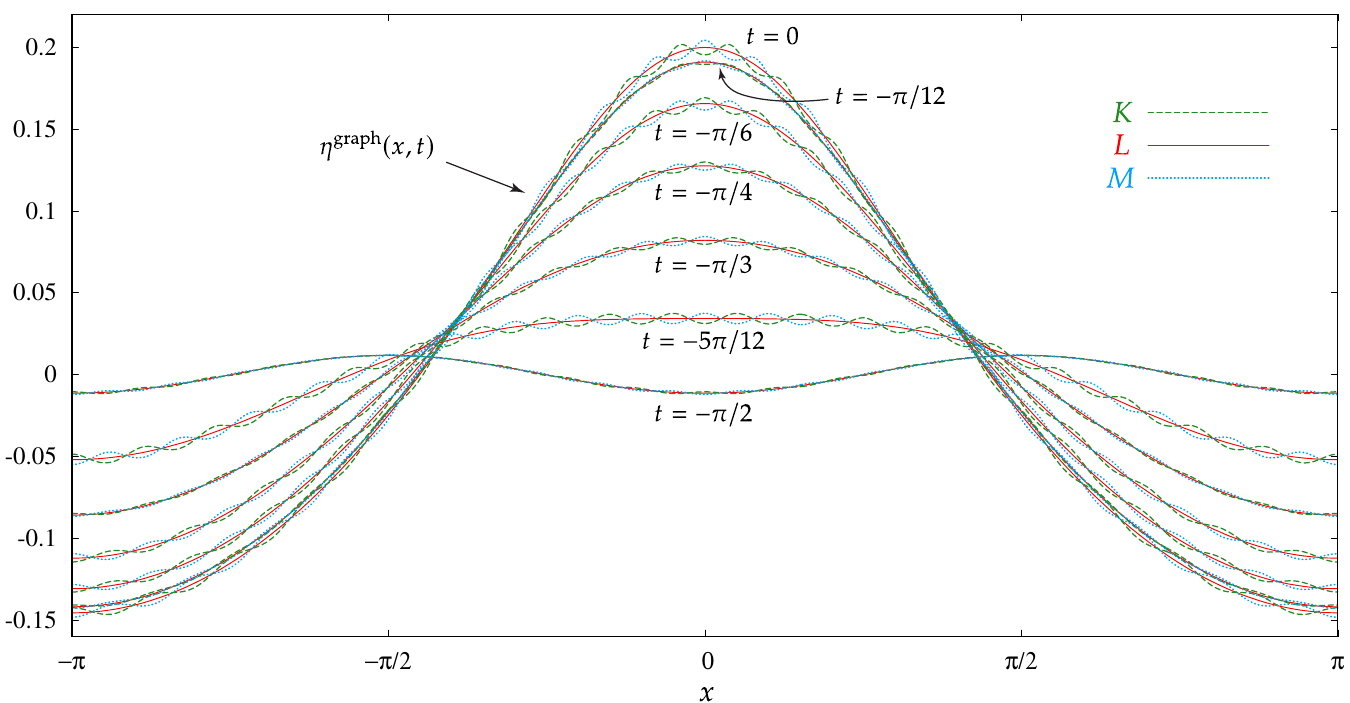}
  \end{center}
  \caption{\label{fig:evol1:hm} Snapshots of the unit-depth standing
    waves labeled $HIJ$ and $KLM$ in
    figures~\ref{fig:bif:evol:1}--\ref{fig:bif:19} at times
    $t\in\mathscr{T}_6$.}
\end{figure}

\begin{figure}[p]
  \begin{center}
    \includegraphics[width=\linewidth]{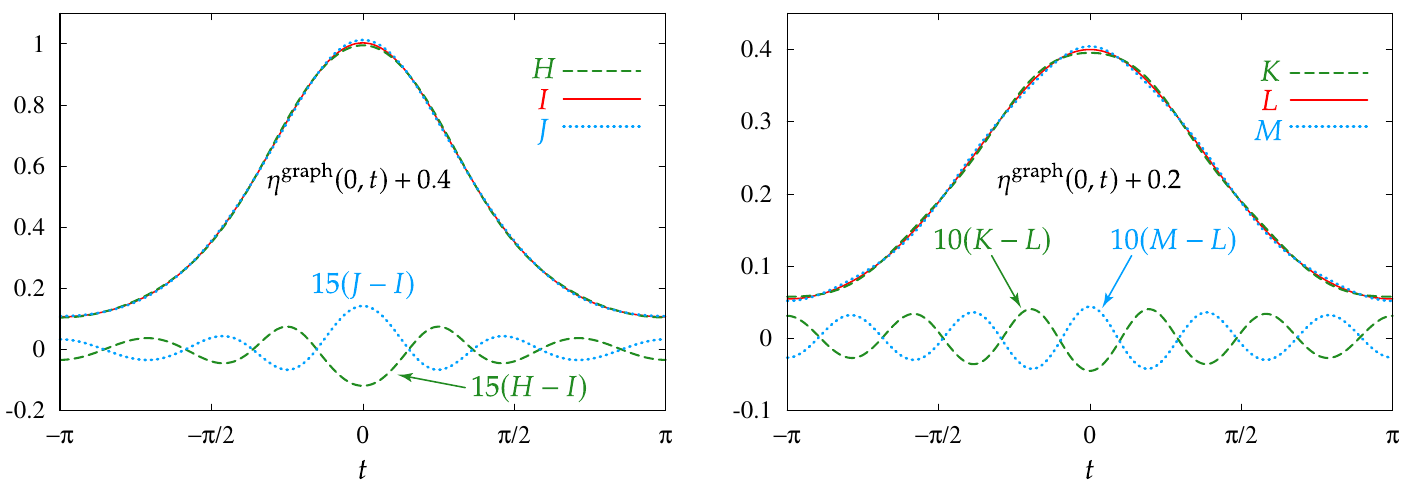}
  \end{center}
  \caption{\label{fig:evol1:x0} Time-evolution of the unit-depth
    standing waves labeled $HIJ$ and $KLM$ evaluated at $x=0$. The
    curve labeled $15(J-I)$ shows
    $15\big[\eta^\text{graph}_J(0,t)-\eta^\text{graph}_I(0,t)\big]$,
    with similar formulas for the other cases.}
\end{figure}

Figure~\ref{fig:evol1:hm} shows snapshots of the unit-depth standing
waves $HIJ$ and $KLM$ from the bifurcation plots of
figures~\ref{fig:bif:evol:1}--\ref{fig:bif:19} at times
$t\in\mathscr{T}_6$, while figure~\ref{fig:evol1:x0} shows the
time-evolution of the wave profiles above the symmetry point $x=0$
over one period.  The waves labeled $HIJ$ have common period
$T=7.227964$ while the waves labeled $KLM$ have common period
$T=7.195747$.  Solutions $H$ and $J$ oscillate around solution $I$
with twelve spatial oscillations and four temporal oscillations
while solutions $K$ and $M$ oscillate around solution $L$ with 19 spatial
oscillations and five temporal oscillations.  Combined with the results
in figures~\ref{fig:evol1defBig} and~\ref{fig:x0evol1and06}, this
confirms that these bifurcation branches correspond to the three
approximate resonances $(p,j)\in\{(7,3),(12,4),(19,5)\}$ in the
cluster of small divisors for $\mu_0=1$ in figure~\ref{fig:smallDiv}.
The secondary standing wave sharpens the wave crest at $t=0$ for
solutions $J$ and $M$ and leads to dimples at the wave crest at $t=0$
for solutions $H$ and $K$. This also causes the crest-to-trough height
$\epsilon$ in figure~\ref{fig:bif:T1} to increase for solutions $J$
and $M$ relative to $I$ and $L$, respectively, and to decrease for
solutions $H$ and $K$. In figure~\ref{fig:evol1:x0}, we multiplied the
perturbation plots by 15 (left) or 10 (right) to better see the
deviation from sinusoidal behavior in the secondary standing
waves. This deviation is more pronounced for solutions $HIJ$ as they
have larger amplitude than solutions $KLM$.

\section{A method of identifying harmonic resonances}
\label{sec:identify:hr}

After using numerical continuation to follow the $DEF$, $HIJ$ and
$KLM$ bifurcation branches in figure~\ref{fig:bif:evol:1}, we noticed
a persistent pole-zero pair in the Pad\'e approximant of $T$ in
figure~\ref{fig:bif:T1} near $\epsilon_*=0.2738080600$. Although the
pole and zero agree with each other to 24 leading digits for the
$149^\text{th}$-order Pad\'e approximant of $T$, it turns out to be an
actual imperfect bifurcation rather than a spurious Froissart doublet
\cite{gonnet:13}. Our goal in this section is to develop a method of
identifying which harmonic resonance is responsible for such a
bifurcation that has been located via Pad\'e techniques. We wish to
avoid relying on numerical continuation to extend the bifurcation
branches far enough that the secondary waves become visible, as this
is expensive.

We used the shooting method to compute sixteen additional solutions at
amplitudes $\epsilon_*\pm\delta_k$, where
$\delta_k=10^{-4.4433-0.4772k}$ for $0\le k\le 7$.  (This was an
  arbitrary choice with the feature that the distance to $\epsilon_*$
  decreases geometrically as $k$ increases.) These solutions had to be
computed in quadruple-precision to see the effects of the
bifurcation. We plotted the Fourier modes $\hat\eta_p$ ($p$ even) and
$\hat\varphi_p$ ($p$ odd) of the initial conditions
\eqref{eq:shooting:fourier} of the shooting method results as
functions of $\epsilon$ for $1\le p\le50$ and found that
$\hat\eta_{36}$ undergoes the largest jump when $\epsilon$ crosses
$\epsilon_*$.  We then computed the $149^\text{th}$-order Pad\'e
approximant $\epsilon^{36}[28/28]_{\tilde\tau_{36}}(\epsilon^2)$ of
$\hat\eta_{36}$ to see if it accurately predicts the shooting method
results near this bifurcation. This is confirmed in
figure~\ref{fig:bif027fit}(a), where all sixteen values of $\hat\eta_{36}$
from the shooting method results lie on the Pad\'e curve. (The errors,
  not shown, range from $3.5\times10^{-28}$ at $\epsilon_*-\delta_0$
  to $2.3\times10^{-22}$ at $\epsilon_*+\delta_7$.)  The four
solutions closest to $\epsilon_*$, with
$\epsilon=\epsilon_*\pm\delta_k$, $k\in\{6,7\}$, are labeled $N$, $P$,
$Q$ and $R$.

Rather than follow the side branches further by numerical continuation
in order to directly observe the secondary waves, we make use of the
fact that the Jacobian $J_{mk}=\partial r_m/\partial\theta_k$ from the
shooting method is nearly singular near an imperfect bifurcation,
where $r_m$ and $\theta_k$ were defined in equations \eqref{eq:obj:fcn} and
\eqref{eq:dof}. In this step, we drop $T$ from the vector $\theta$ in
\eqref{eq:dof} instead of one of the Fourier modes of the initial
condition. This is the Jacobian in the variant of the algorithm where
$T$ is specified as the bifurcation parameter.  At solution $P$, the
smallest singular value of $J$ is $3.4\times 10^{-9}$. The second
smallest is $2.3\times10^{-6}$ and the largest is $1.395$.  The green
circles in figure~\ref{fig:bif027fit}(b) show the magnitudes
of the components of the singular vector corresponding to the smallest
singular value. This is a right singular vector, which we denote by
$\dot\theta$. Here we use a dot for a perturbation direction or a
variational derivative with respect to this perturbation, not for a
time derivative. The components of $\dot\theta$ are the Fourier modes
of the initial conditions of the linearized water wave equations about
solution $P$, given in \cite{water2}, that minimize the norm of $\dot
r=J\dot\theta$, subject to the constraint $\|\dot\theta\|=1$. The
corresponding linearized solution about solution $P$ is denoted
$\big(\dot\eta^\text{graph}(x,t),\dot\varphi^\text{graph}(x,t)\big)$.
The components of $\dot\theta$ are ordered by interlacing
$\dot{\hat\varphi}_k^\text{graph}(t_0)$ for $k$ odd with
$\dot{\hat\eta}_k^\text{graph}(t_0)$ for $k$ even, for $1\le k\le d$.
We set $d=120$ in this calculation and used $M_1=324$ gridpoints. We
did not use adaptive grids, so $N=1$ in \eqref{eq:adapt:grids} and
\eqref{eq:obj:fcn}.

\begin{figure}[b]
  \begin{center}
    \includegraphics[width=\linewidth]{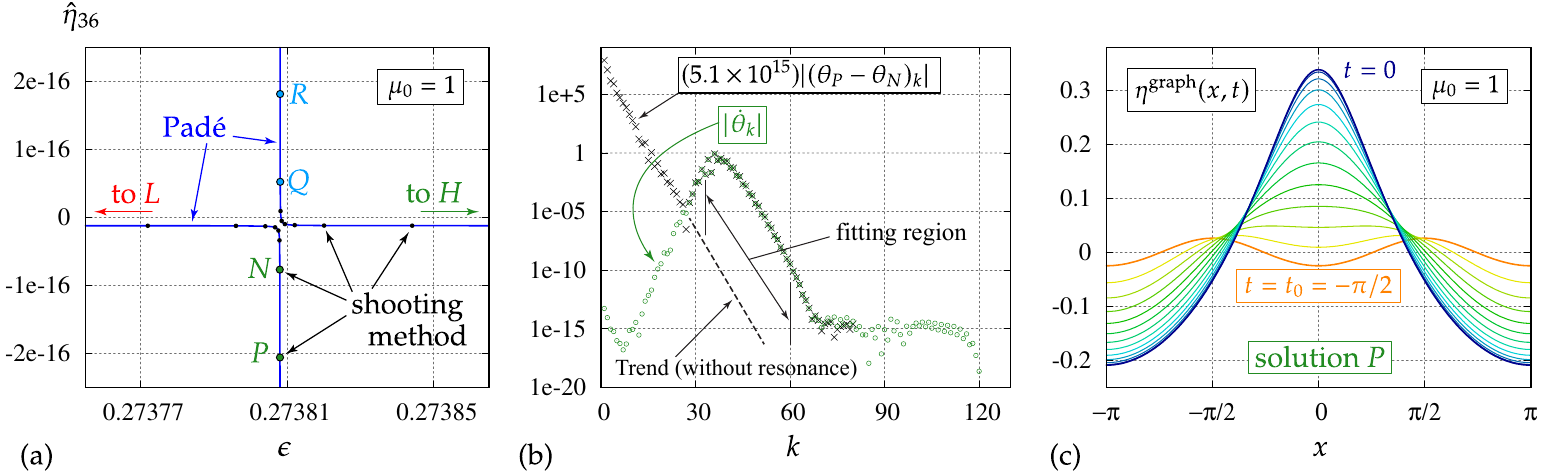}
  \end{center}
  \caption{\label{fig:bif027fit} An imperfect bifurcation near
    $\epsilon_*=0.2738080600$ is predicted by a pole in the
    $149^\text{th}$-order Pad\'e approximant of $T$. (a) We computed
    sixteen shooting method solutions near this pole, which are the black
    markers and points labeled $NPQR$. (b) To identify the harmonic
    resonance responsible for the bifurcation, we computed the
    perturbation direction $\dot\theta$ corresponding to the right
    singular vector of the Jacobian at solution $P$ with the smallest
    singular value and compared the high-frequency components of
    $(\theta_P-\theta_N)$ and $(\theta_R-\theta_Q)$ to those of
    $\dot\theta$. (c) Time-evolution of solution $P$ for
    $t\in\mathscr{T}_{12}$.}
\end{figure}

\begin{figure}[p]
  \begin{center}
    \vspace*{1.5pc}
    \includegraphics[width=\linewidth]{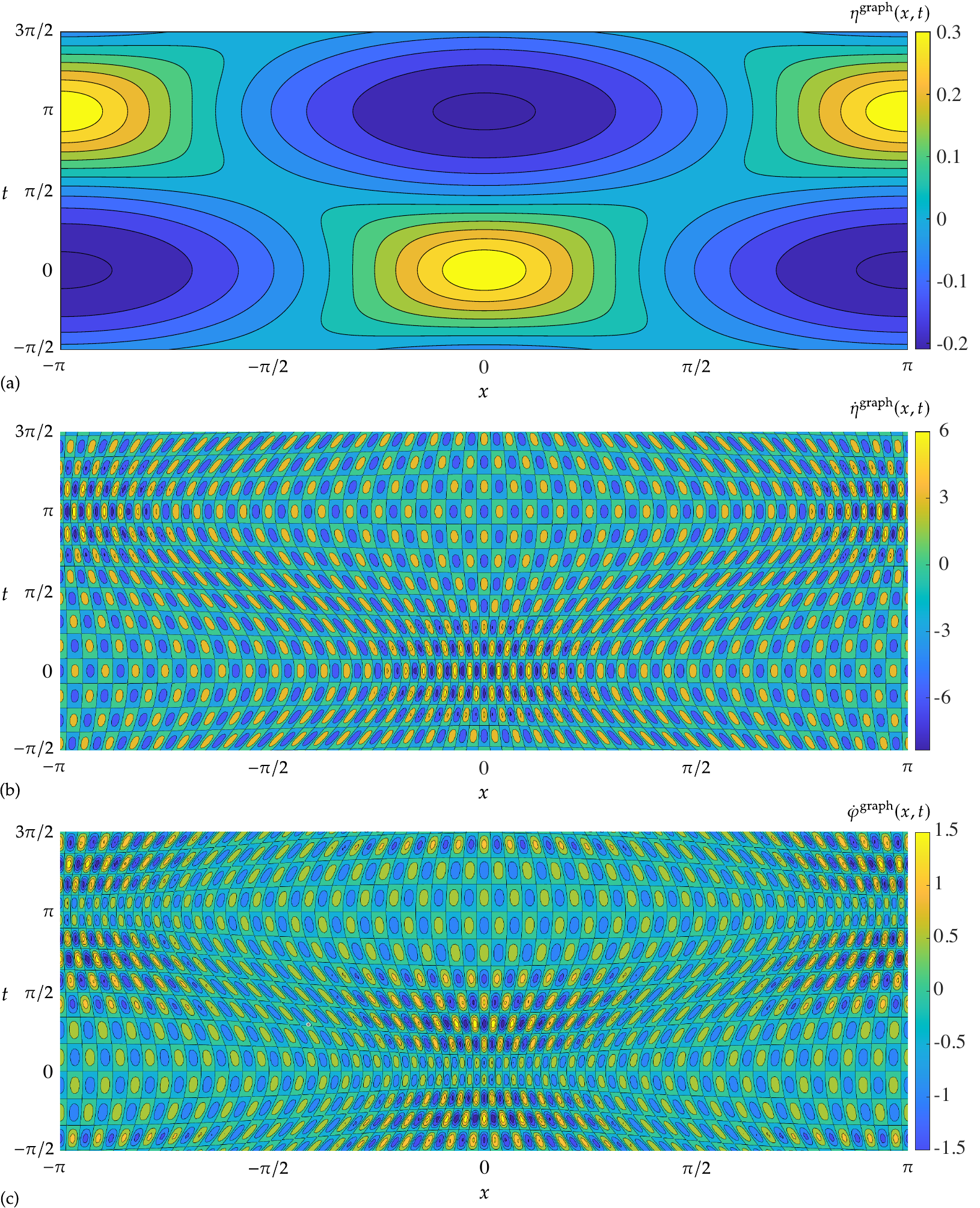}
  \end{center}
  \caption{\label{fig:lin:eta:phi} Contour plot of
    $\eta^\text{graph}(x,t)$ for solution $P$ over a full period and
    both components of the linearized solution about $P$ with initial
    conditions given by $\dot\theta$, the right singular vector of $J$
    at $P$ with the smallest singular value
    $\sigma_\text{min}=3.4\times 10^{-9}$. This linearized solution
    gives the perturbation direction to an approximate secondary
    standing wave of the same period as solution $P$. There are 37
    spatial oscillations and 7 temporal oscillations, but they are not
    uniform.}
\end{figure}

Figure~\ref{fig:bif027fit}(c) shows snapshots of solution
$P$ for $t\in\mathscr{T}_{12}$.  The same data is shown as a contour
plot in figure~\ref{fig:lin:eta:phi}(a), except that
solution $P$ is evolved over a full period $-\frac\pi2\le t\le
\frac{3\pi}2$ instead of a quarter period. The wave crest that forms
at $(x,t)=(0,0)$ appears again, shifted in space and time, at
$(x,t)=(\pm\pi,\pi)$.  Figures~\ref{fig:lin:eta:phi}(b)
and~\ref{fig:lin:eta:phi}(c) show contour plots of the linearized
solution
$\big(\dot\eta^\text{graph}(x,t),\dot\varphi^\text{graph}(x,t)\big)$
with initial conditions $\dot\theta$. The solution was normalized to
make $\dot\theta$ a unit vector in $\mbb R^d$.  The linearized
velocity potential $\dot\varphi^\text{graph}(x,t)$ is
indistinguishable from zero at $t=0$ and $t=\pi$ in the contour plot
of figure~\ref{fig:lin:eta:phi}(c). This is because
$\dot r=J\dot\theta$ satisfies
\begin{equation}\label{eq:rdot:norm}
  \|\dot r\| = \left(\frac1{M_1}
    \sum_{m=0}^{M_1-1} \dot\varphi^\text{graph}(x_{1m},0)^2\right)^{1/2} =
  \sigma_\text{min} = 3.4\times 10^{-9}.
\end{equation}
If $\dot\varphi^\text{graph}$ were exactly zero at $t=0$, a symmetry
argument \cite{mercer:92,mercer:94,water2,water:stab1} would ensure
that $(\dot\eta^\text{graph},\dot\varphi^\text{graph})$ is
time-periodic with period $T$ and $\dot\varphi^\text{graph}$ is zero
again at $t=\pi$.  The small value of $\|\dot r\|$ in equation
\eqref{eq:rdot:norm} nearly achieves the same result, where we find
that
\begin{equation}
  \begin{aligned}
   & \left(\frac1{M_1}\sum_{m=0}^{M_1-1}\left\{\Big[
        \dot\eta^\text{graph}\Big(x_{1m},\frac{3T}4\Big)
        - \dot\eta^\text{graph}\Big(x_{1m},-\frac T4\Big)\Big]^2 + \right.\right. \\
      & \hspace*{1in} \left.\left.
    \Big[\dot\varphi^\text{graph}\Big(x_{1m},\frac{3T}4\Big)
      - \dot\varphi^\text{graph}\Big(x_{1m},-\frac T4\Big)\Big]^2\right\}\right)^{1/2} =
    8.2\times10^{-8}
  \end{aligned}
\end{equation}
and
\begin{equation}\label{eq:secondary:nrms}
  \begin{aligned}
    &\left(\frac1{M_1} \sum_{m=0}^{M_1-1}
  \dot\varphi^\text{graph}(x_{1m},T/2)^2\right)^{1/2} = 1.08\times10^{-8}.
  \end{aligned}
\end{equation}
For reference on the size of the discrete $L^2$ norms in equations
\eqref{eq:rdot:norm}--\eqref{eq:secondary:nrms}, we have
\begin{equation}
  \left(\frac1{M_1}\sum_{m=0}^{M_1-1}\left[
      \dot\eta^\text{graph}\big(x_{1m},-T/4\big)^2 +
      \dot\varphi^\text{graph}\big(x_{1m},-T/4\big)^2
      \right]\right)^{1/2} = \sqrt2,
\end{equation}
which follows from discrete orthogonality of the functions $e^{ikx}$
on the grid $\{x_{1m}\}_{m=0}^{M_1-1}$ for $|k|<M_1/2$ together with
$\|\dot\theta\|=1$ and the fact that $\theta$ only contains
positive-index Fourier modes in equation \eqref{eq:dof}. We interpret
$\dot\eta^\text{graph}(x,t)$ and $\dot\varphi^\text{graph}(x,t)$ as
the perturbation direction of a nearly time-periodic, infinitesimal
secondary standing wave.  Counting the oscillations in
figures~\ref{fig:lin:eta:phi}(b) and~\ref{fig:lin:eta:phi}(c)
shows that this bifurcation corresponds to the $(p,j)=(37,7)$
harmonic resonance, but the sinusoidal pattern of the wave has been
significantly distorted as it evolves over solution $P$, the primary
wave of figure~\ref{fig:lin:eta:phi}(a).

Our final task is to determine the phase of this secondary standing
wave on the two bifurcation branches passing through $NP$ and $QR$ in
figure~\ref{fig:bif027fit}(a). The main challenge is hidden
by the extreme aspect ratio of the figure. The change in $\epsilon$
from point $N$ to point $P$ is $2.6\times10^8$ times larger than the
change in $\hat\eta_{36}$, even though it looks like the bifurcation
curve is nearly vertical from $N$ to $P$ in the plot.  Most of the
change in the initial condition $\theta$ from $N$ to $P$ is due to the
dependence on $\epsilon$ of the underlying primary wave rather than
the excitation of the secondary wave.  Our idea is to filter this out
by studying the alignment of the higher-frequency components of
$(\theta_P-\theta_N)$ with those of $\dot\theta$.  In
figure~\ref{fig:bif027fit}(b), we plot the magnitudes of the components
of $C(\theta_P-\theta_N)$ on top of those of $\dot\theta$, where
$C=5.1\times10^{15}$. This factor of $C$ visually aligns the
magnitudes of the components of $(\theta_P-\theta_N)$ with those of
$\dot\theta$ over the range $33\le k\le 60$.

The low-frequency components of $C(\theta_P-\theta_N)$ are large but
decay rapidly. The dashed line shows the trend line if these modes had
continued to decay geometrically at their initial decay rate.
Instead, there is a growth phase beginning at $k=28$ where the
components of $(\theta_P-\theta_N)$ grow by five orders of magnitude
before decaying again. We formed vectors $u$ and $v$ containing
components $33\le k\le 60$ of $(\theta_P-\theta_N)$ and $\dot\theta$,
rescaled to make $u$ and $v$ unit vectors in $\mbb R^{28}$. We find
that the angle $\Theta$ between $u$ and $-v$, computed via
$\sin(\Theta/2)=\frac12\|u-(-v)\|$, is $\Theta=5.54\times 10^{-8}$,
which shows that the high-frequency components of
$(\theta_P-\theta_N)$ are nearly perfectly aligned with those of
$-\dot\theta$.  Similarly, if we replace $u$ by components $33\le k\le
60$ of $(\theta_R-\theta_Q)$, the angle $\Theta$ between $u$ and $v$,
computed via $\sin(\Theta/2)=\frac12\|u-v\|$, is also
$\Theta=5.54\times 10^{-8}$.  It was not necessary to recompute
$\dot\theta$ at $R$ when switching from $(\theta_P-\theta_N)$ to
$(\theta_R-\theta_Q)$. The contour plots in
figure~\ref{fig:lin:eta:phi}(b,c) look identical whether we linearize
around $P$ or $R$. Since the sign of $\dot\eta^\text{graph}(0,0)$ is
positive in figure~\ref{fig:lin:eta:phi}(b), we learn that
following the bifurcation branch passing through $Q$ and $R$ leads to
a secondary standing wave that is in phase with the primary wave,
which sharpens the crest at $(x,t)=(0,0)$ and increases $\epsilon$.
Following the branch passing through $N$ and $P$ leads to a secondary
standing wave that is out of phase with the primary wave, which
flattens the crest and decreases $\epsilon$. A dimple would likely
form at the crest if one follows the bifurcation branch far enough in
that direction.

This method of studying the solution of the linearized problem about a
standing wave near an imperfect bifurcation predicted by a Pad\'e pole
on the real $\epsilon$-axis to classify the resonance responsible for
the bifurcation is, to our knowledge, new, and is much less expensive
than using numerical continuation to compute fully nonlinear solutions
far out on the bifurcation branches to directly observe the secondary
standing waves that are excited by the resonance. It is interesting
that $\hat\eta_{36}$ in figure~\ref{fig:bif027fit} responds more
strongly to the $(37,7)$ resonance than $\hat\varphi_{37}$. This shows
that the strong deformation of the shape of the $(37,7)$ resonance in
figure~\ref{fig:lin:eta:phi}(b,c) away from the
tensor product form $\cos(37x)\cos(7t)$ and the change of variables
from the graph-based formulation plotted in the figure to conformal
variables have large effects on the Fourier modes $\hat\eta_p$ and
$\hat\varphi_p$ of the initial conditions in equation
\eqref{eq:shooting:fourier}.

\end{document}